\documentclass[11pt]{article}
\usepackage{booktabs}
\usepackage{tabularx}
\usepackage{array}

\usepackage{amsthm}
\usepackage{graphicx} 
\usepackage{array} 

\usepackage{booktabs}
\usepackage{tabularx}

\usepackage{amsmath, amssymb, amsfonts, verbatim}
\usepackage{hyphenat,epsfig,subcaption,multirow}
\usepackage{nicefrac}
\usepackage{paralist}
\usepackage{mleftright}

\usepackage{mathtools}

\usepackage{pdfpages}

\usepackage[font=small, labelfont=bf]{caption}

\usepackage[dvipsnames]{xcolor}
\usepackage[ruled]{algorithm2e}

\DeclareFontFamily{U}{mathx}{\hyphenchar\font45}
\DeclareFontShape{U}{mathx}{m}{n}{
      <5> <6> <7> <8> <9> <10>
      <10.95> <12> <14.4> <17.28> <20.74> <24.88>
      mathx10
      }{}
\DeclareSymbolFont{mathx}{U}{mathx}{m}{n}
\DeclareMathSymbol{\bigtimes}{1}{mathx}{"91}

\usepackage{tcolorbox}
\tcbuselibrary{skins,breakable}
\tcbset{enhanced jigsaw}

\usepackage[normalem]{ulem}
\usepackage[compact]{titlesec}

\definecolor{DarkRed}{rgb}{0.5,0.1,0.1}
\definecolor{DarkBlue}{rgb}{0.1,0.1,0.5}

\definecolor{colorOne}{HTML}{E69F00}   
\definecolor{colorTwo}{HTML}{56B4E9}   
\definecolor{colorThree}{HTML}{009E73} 
\definecolor{colorFour}{HTML}{F0E442}  
\definecolor{colorFive}{HTML}{0072B2}  
\definecolor{colorSix}{HTML}{D55E00}   
\definecolor{colorSeven}{HTML}{CC79A7} 
\definecolor{colorEight}{HTML}{000000} 
\definecolor{colorNine}{HTML}{999999}  

\usepackage{pifont}

\usepackage{bm}
\usepackage{url}
\usepackage{xspace}
\usepackage[mathscr]{euscript}

\usepackage{environ}  
\usepackage{chngcntr} 

\usepackage{booktabs}

\usepackage[T1]{fontenc}
\usepackage{lmodern}

\usepackage{tikz}
\usetikzlibrary{arrows}
\usetikzlibrary{arrows.meta}
\usetikzlibrary{shapes}
\usetikzlibrary{backgrounds}
\usetikzlibrary{positioning}
\usetikzlibrary{decorations.markings}
\usetikzlibrary{decorations.pathmorphing}
\usetikzlibrary{patterns}
\usetikzlibrary{calc}
\usetikzlibrary{fit}
\usetikzlibrary{decorations.pathreplacing}

\tikzset{vertexA/.style={circle,fill=black,minimum size=5pt,inner sep=0pt, font=\small}}
\tikzstyle{vertexB}=[circle,fill=black,minimum size=3pt,inner sep=0pt, font=\small]
\tikzstyle{selected vertex} = [vertex, fill=red!24]
\tikzstyle{edge} = [draw, black!50, -]
\tikzstyle{vertexRed}=[circle,fill=red,draw,minimum size=10pt,inner sep=0pt, font=\small, line width=0.5pt]
\tikzstyle{vertexBlack}=[circle,fill=black,draw,minimum size=10pt,inner sep=0pt, font=\small, line width=0.5pt]

\usepackage{mdframed}

\usepackage[noend]{algpseudocode}
\makeatletter
\def\BState{\State\hskip-\ALG@thistlm}
\makeatother

\usepackage{cite}
\usepackage{enumitem}

\usepackage[margin=1in]{geometry}

\usepackage{float} 

\newtheorem{theorem}{Theorem}
\newtheorem{lemma}{Lemma}[section]

\newtheorem{proposition}[lemma]{Proposition}
\newtheorem{corollary}[lemma]{Corollary}

\newtheorem*{claim*}{Claim}
\newtheorem*{proposition*}{Proposition}
\newtheorem*{lemma*}{Lemma}
\newtheorem*{problem*}{Problem}

\usepackage{nameref}
\definecolor{ForestGreen}{rgb}{0.1333,0.5451,0.1333}
\definecolor{Red}{rgb}{0.9,0,0}
\usepackage[linktocpage=true,
	pagebackref=true,colorlinks,
		urlcolor=black,
	linkcolor=DarkRed,citecolor=ForestGreen,
	bookmarks,bookmarksopen,bookmarksnumbered]
	{hyperref}
\usepackage[noabbrev,nameinlink]{cleveref}
\crefname{property}{property}{Property}
\creflabelformat{property}{(#1)#2#3}
\crefname{equation}{eq}{Eq}
\creflabelformat{equation}{(#1)#2#3}

\crefname{lemma}{Lemma}{Lemmas}
\crefname{claim}{Claim}{Claims}

\newtheorem{mdresult}{Result}

\newtheorem*{mdresult*}{Main Result}

\newtheorem{definition}[lemma]{Definition}
\theoremstyle{definition}
\newtheorem{remark}[lemma]{Remark}
\newtheorem{observation}[lemma]{Observation}

\crefname{assumption}{assumption}{assumptions}
\Crefname{assumption}{Assumption}{Assumptions}

\newtheorem{mdinvariant}[lemma]{Lemma}

\theoremstyle{definition}
\newtheorem{mdalg}{Algorithm}
\newenvironment{Algorithm}{\begin{tbox}\begin{mdalg}}{\end{mdalg}\end{tbox}}

\allowdisplaybreaks



\renewcommand{\qed}{\nobreak \ifvmode \relax \else
      \ifdim\lastskip<1.5em \hskip-\lastskip
      \hskip1.5em plus0em minus0.5em \fi \nobreak
      \vrule height0.75em width0.5em depth0.25em\fi}

\renewcommand{\leq}{\leqslant}
\renewcommand{\geq}{\geqslant}

\newcommand{\Ot}{\ensuremath{\widetilde{O}}}

\newcommand{\eps}{\ensuremath{\varepsilon}}

\newcommand{\bracket}[1]{\left[#1\right]}
\newcommand{\paren}[1]{\ensuremath{\left(#1\right)}\xspace}

\newcommand{\expect}[1]{\Exp\bracket{#1}}

\newcommand{\ALG}{\ensuremath{\mbox{\sc alg}}\xspace}

\DeclareMathOperator*{\Exp}{\ensuremath{{\mathbb{E}}}}
\DeclareMathOperator*{\Prob}{\ensuremath{\textnormal{Pr}}}
\renewcommand{\Pr}{\Prob}

\newenvironment{tbox}{\begin{tcolorbox}[
		enlarge top by=5pt,
		enlarge bottom by=5pt,
		 breakable,
		 boxsep=2pt,
                  left=5pt,
                  right=7pt,
                  top=10pt,
                  arc=0pt,
                  boxrule=1pt,toprule=1pt,
                  colback=white
                  ]
	}
{\end{tcolorbox}}

\newcommand{\II}{\ensuremath{\mathbb{I}}}

\newcommand{\mireal}[1][]{
  \ifx\relax#1\relax%
    \II(\mione \,; \mitwo)%
  \else%
    \II(\mione \,; \mitwo\mid #1)%
  \fi
}

\newcommand{\E}{\mathop{\mathbb{E}}}

\renewcommand{\leq}{\leqslant}
\renewcommand{\le}{\leq}
\renewcommand{\geq}{\geqslant}
\renewcommand{\ge}{\geq}

\newcommand{\myqed}[1]{\let\qed\relax \hspace*{\fill} #1 \ensuremath{\square}}

\NewEnviron{constraint}{
	\begin{tcolorbox}
		\begin{equation}
			\BODY
		\end{equation}
	\end{tcolorbox}
}

\newcommand{\dem}{\textsf{dem}}
\newcommand{\yes}{\textsf{C}}
\newcommand{\no}{\overline{\textsf{C}}}

\newcommand{\sagnik}[1]{}

\title{Bipartite Matching with Few Rounds of Demand Queries}

\author{Shravan Mehra\thanks{University of Birmingham, United Kingdom, \texttt{shravansinghmehra@gmail.com}}
\and
Sagnik Mukhopadhyay\thanks{University of Birmingham, United Kingdom, \texttt{schwagznikst@gmail.com}}
\and
Vihan Shah\thanks{University of Birmingham, United Kingdom, \texttt{vihanshah98@gmail.com}}
}

\date{}

\begin{document}
\maketitle

\thispagestyle{empty}
\pagenumbering{roman}

\begin{abstract}

We study the randomised and deterministic demand-query complexity of
approximate bipartite matching when only a few rounds of queries are
allowed. The demand-query model, introduced in [Blumrosen and Nisan SICOMP 2009; Nisan SODA 2021], is inspired by combinatorial auctions and welfare maximisation, where an oracle returns the first neighbouring vertex \textit{(item)} of a vertex $u$ \textit{(agent)} according to a given permutation $\pi$ \textit{(preference ordering of the agent over the items)}.

\noindent\textbf{Randomised demand query.}
Our main randomised result gives a trade-off between the number of
rounds \(k\), the per-round query budget \(q=n^{1+\delta}\), and the
achievable approximation ratio. Let $1\leq k\leq \log n$, and $-1\leq \delta\leq \frac1k$.
Every \(k\)-round randomised algorithm using \(O(q)\) demand queries per
round has a worst-case expected-size approximation factor of
$\Omega\!\left(
        \frac{n^{\frac{1}{2k}-\frac{\delta}{2}}}{k}\right)$.

We also prove a
high-probability lower bound that, for \(\delta\geq 0\), loses only a
factor of \(\sqrt{\log n}\). Conversely, we give a randomised algorithm achieving approximation factor
$O\!\left(
\frac{n^{(1-k\delta)/(k+1)}}{k}
\right)$
in expectation. For $q\ge n$, the same approximation factor is achieved
with probability at least $2/3$.
This implies:

\begin{itemize}
    \item For the non-adaptive case, these bounds match, giving the tight trade-off \(\Theta(n^{(1-\delta)/2})\).

    \item For \(q=\Theta(n)\), a constant-factor approximation has round complexity \(\Theta(\log n/\log\log n)\).
\end{itemize}

We note that in the stronger number-in-hand multiparty shared-blackboard communication model, a corresponding lower bound for constant-factor approximation was established by [Braverman, Oshman FOCS 2017]. We recover the implication of their result for the demand-query model as a special case of our result, with a significantly simpler proof.

\noindent\textbf{Deterministic demand query.}
For deterministic algorithms, for \(kq\geq n\), we show that the optimal deterministic approximation factor is
$\widetilde{\Theta}\!\left(
        \frac{n^2}{k^2q}\right)$,
thereby improving the previous lower bound of [Nisan SODA 2021] by a polynomial factor over a wide range of parameters.
In particular, for $q=\Theta(n)$, we improve the deterministic
round-complexity lower bound for constant-factor approximation
from $\widetilde{\Omega}(n^{1/3})$ to
$\widetilde{\Omega}(\sqrt{n})$.
The round complexity of our deterministic algorithm also suffices for constructing a maximal matching.
Our algorithms also return matching and vertex-cover
certificates, and the deterministic trade-off extends to minimum vertex
cover and maximum fractional matching.

\noindent\textbf{Demand-OR query.}
We extend the randomised upper bounds to the weaker Demand-OR model [Nisan 2021].
We prove a non-adaptive simulation of uniform random-neighbour queries using known techniques, avoiding the $O(\log n)$-round penalty of the trivial binary-search simulation.
Consequently, our upper bounds yield the same $k$-round approximation guarantees using $O(q\log^2 n)$ Demand-OR queries per round.
They also yield a $(2+\varepsilon)$-approximation with high probability in $O(\log n/\log\log n)$ rounds using $O(n\log^3 n)$ Demand-OR queries per round.

\end{abstract}

\clearpage
\bigskip
\setcounter{tocdepth}{2}
\tableofcontents

\clearpage

\pagenumbering{arabic}
\setcounter{page}{1}

\section{Introduction}

Bipartite matching is one of the most fundamental and extensively studied problems in graph algorithms, with a long history spanning combinatorial optimisation, theoretical computer science, and operations research. It serves as a canonical model for allocation problems, with applications ranging from job assignment to resource allocation \cite{DGS86, gale1962college, AdWords}.

A closely related problem arises in combinatorial auctions, where each agent seeks a bundle of items that maximises their utility. In the unit-demand setting, this problem is equivalent to bipartite matching: agents correspond to vertices on one side (say, the left side), items to vertices on the other (right side), and feasible allocations correspond to matchings in the underlying bipartite graph. This connection is well studied in the literature; see, for example, \cite{AlonNRW15, Nisan21, NisanS05, Dob2005}.

\smallskip\noindent
\textbf{The Demand Query Model.} Motivated by this connection, we study bipartite matching in the demand query model \cite{Nisan21, NisanS05}, where 
a query specifies a vertex (agent) $u$ on the left together with an ordering over the vertices (items) on the right side, and returns the first neighbour of $u$ according to this ordering. This model captures the behaviour of unit-demand agents responding to prices. The measure of complexity for this model is the number of such queries and, crucially for this work, the number of rounds of adaptivity. As noted in \cite{Nisan21}, this query model sits between standard adjacency matrix query, which demand query can simulate, and communication complexity which can simulate demand queries with logarithmic overhead. However, the complexity of these two models is significantly different: \cite{BlikstadBEMN22} showed that the deterministic communication complexity of exact bipartite matching is $\Ot(n)$, but one needs $\Omega(n^2)$ adjacency matrix queries to output a constant approximation of maximum matching, even in the randomised setting \cite{AssadiCK19, BhattacharyaKS23}.

\smallskip\noindent
\textbf{Approximate Bipartite Matching.} 
We study the problem of computing an approximate maximum matching in the demand query model. Let $\mu$ denote the size of a maximum matching. For $\alpha \ge 1$, we say an algorithm achieves an $\alpha$-approximation if it outputs a matching of size at least $\mu / \alpha$.

We briefly mention the relevant past research here; we explain more in the next section. Combinatorial auctions, especially social welfare maximising allocation for unit-demand bidders (in other words, approximate bipartite matching), have long been studied through the lens of communication protocols, specifically using simultaneous message passing (SMP) protocols for one-round auctions \cite{DobzinskiNO19} and using multiparty number-in-hand shared blackboard communication protocols when more rounds are involved \cite{AlonNRW15, BravermanO17}.
Motivated by these communication protocols, \cite{Nisan21} introduced the demand query model for bipartite matching, capturing a natural and concrete form of limited interaction between the agents and the allocator. Our work continues this line of research by studying the power of demand queries under a limited number of adaptive rounds, with a particular emphasis on lower bound techniques.

In terms of proving lower bounds for demand queries, there are two main techniques: multiparty communication complexity and Demand-OR query. We address both of them in this work. Firstly,
\cite{DobzinskiNO19} established the current best randomised multiparty number-in-hand shared blackboard communication protocol which is an $O(\log n)$-round protocol for constant-factor approximation, a result subsequently translated into the demand query framework by Nisan \cite{Nisan21}. This communication model is indeed more powerful than the demand query model and, as a result, any lower bound obtained here holds true for the demand query. In this communication model, \cite{DobzinskiNO19, HuangRVZ15} also showed a single-round communication lower bound for matching. Subsequently, \cite{AlonNRW15} and \cite{BravermanO17} improved this lower bound for multiple rounds, and the current state-of-the-art is an $\Omega(\log n/\log \log n)$-round lower bound for achieving constant approximation of the maximum matching when $O(n)$ communication is allowed per round \cite{BravermanO17}. We expound on these results in the next section.

The second technique is via a weaker query model, namely the Demand-OR query model, which only achieves parity with the standard demand query when $\Omega(\log n)$ rounds are available for binary search. The existing deterministic lower bounds for demand queries provided by \cite{Nisan21} rely on the Demand-OR query model, where it is shown that a $k$-round Demand-OR query algorithm must issue $\widetilde{\Omega}(n^2/k^3)$ queries\footnote{The notations $\widetilde{O}$, $\widetilde{\Omega}$, and $\widetilde{\Theta}$ hide polylogarithmic factors in $n$.} per round in order to achieve constant-approximation.

Consequently, to the best of our knowledge, there are no existing lower bound techniques specifically tailored to the demand query model. 
Existing approaches proceed through stronger or less precise models, and therefore do not yield a complete picture of demand-query complexity.
On one hand, even though communication-based lower bounds have proven effective for constant-factor approximations, if we aim for the exact version of the problem, it can be solved with $\Ot(n)$ total communication deterministically (even in two-party communication model) \cite{BlikstadBEMN22} given enough rounds. Achieving a comparable result with demand queries remains well beyond current algorithmic reach, and it might be the case that communication protocols are simply more powerful than demand query. On the other hand, although working with demand-OR query gives us an opportunity to develop new query-specific lower bound techniques, the $\log n$-round penalty might prove too expensive to show optimal bounds for fewer rounds. In this work, we develop lower bound techniques exclusive for a few rounds of demand queries.


\subsection{Our results}


\paragraph*{Randomised demand query.} We first consider completely non-adaptive algorithms. We know from \cite{shah2026sublinear} that we can get a $\sqrt{n}$-approximate maximum matching using only $\Ot(n)$ random neighbour queries which can be simulated by demand queries. Towards a lower bound, \cite{DobzinskiNO19} showed that any randomised simultaneous message passing (SMP) protocol where each player sends at most $n^\delta$ ($\delta < 1$) bits cannot approximate maximum matching better than a factor of $O(n^{1/2 - 2\delta})$. In the non-adaptive setting, the SMP model is stronger than the general demand query model and hence, putting $\delta = 0$, we recover a $\sqrt n$-type lower bound when each player issues $O(1)$ demand queries.

We recover and extend this bound using only demand queries, thus bypassing the requirement of using a (possibly) stronger model to prove lower bounds. More generally, when $O(n^{1+\delta})$ demand queries are allowed, we obtain a tight $\Theta(n^{1/2-\delta/2})$ expected-size approximation tradeoff. The same lower bound holds for algorithms required to achieve their approximation with constant probability, and for $\delta\ge 0$ the upper bound also holds with constant probability. 
For example, when $\delta=1/4$, \cite{DobzinskiNO19} rules out constant-factor approximation, whereas we rule out any $o(n^{3/8})$-approximation. We additionally establish a high-probability lower bound that loses only a factor of $\sqrt{\log n}$.

\begin{theorem}[Non-adaptive/ One round]
Consider non-adaptive demand query algorithms for maximum matching. Then, for every $\delta\in[-1,1]$,
\begin{itemize}
        \item any randomised algorithm that uses $O(n^{1+\delta})$ demand queries has worst-case expected-size approximation factor $\Omega(n^{1/2-\delta/2})$; the same asymptotic lower bound holds for algorithms that achieve their approximation with probability at least $2/3$;
        \item there is a randomised algorithm that uses $O(n^{1+\delta})$ demand queries and has expected-size approximation factor $O(n^{1/2-\delta/2})$. For $\delta\ge 0$, the same asymptotic approximation holds with probability at least $2/3$.
\end{itemize}
\end{theorem}

Setting $\delta = 0$ gives a tight $\Theta(\sqrt n)$ expected-size approximation bound with $O(n)$ queries, recovering the corresponding one-round results of \cite{DobzinskiNO19, Nisan21}. However, our main result is when the number of rounds of adaptivity is increased to more than one, for example, what happens when there are $k > 1$ rounds and we can issue $n^{1+\delta}$ demand queries in each round. 
\cite{DobzinskiNO19} left as an open question what meaningful lower bound we can prove in this limited adaptivity regime. However, as mentioned before, the model of computation for \cite{DobzinskiNO19} is multiparty number-in-hand shared blackboard model of communication. Any demand query protocol can be simulated in this model with only a logarithmic overhead: Each player simulates the demand query protocol (using shared randomness) and when it is their turn to answer a query, they simply communicate the answer (which is an edge) using $O(\log n)$ bits. Indeed, for this model, \cite{AlonNRW15} improved the results of \cite{DobzinskiNO19} to multiple rounds where they prove an $\Omega(n^{1/5^k})$-approximation lower bound for $k$-round randomised communication protocols.
Subsequently \cite{BravermanO17} showed that any $k$-round communication protocol that outputs a constant-approximation matching must use $\Omega(n^{1+1/k}/k^4)$ bits of communication.

A couple of points to note: First, these communication lower bounds exploit an array of sophisticated round elimination and direct sum techniques and, as such, are complicated. For example, \cite{BravermanO17} employs a reduction from multiparty set disjointness and even proving a multiparty set disjointness lower bound requires careful use of information complexity and direct sum arguments. This is not unusual as round elimination arguments even in two-party communication complexity, let alone multiparty, require care because of the powerful nature of the model. In contrast, for demand queries, we fully recover these bounds by a significantly simpler analysis. To our knowledge, no such lower bound technique for randomised demand query protocols was known before our work. Secondly, even though \cite{BravermanO17} showed a lower bound for constant-approximation, for superconstant approximation factors, the only lower bound is from \cite{AlonNRW15} which we improve exponentially\footnote{We believe that, even though not stated, the result of \cite{BravermanO17} can be made to work for superconstant approximation regime with significant technical modification. In contrast, our lower bound technique gives this lower bound in a much simpler fashion.}. Here is our main result:

\begin{theorem}[Few rounds, $O(n)$ queries.] \label{thm:few-rounds}
Let $k \leq \log n$ be the number of rounds. Then,
\begin{itemize}
        \item any randomised algorithm that uses $O(n)$ demand queries per round has worst-case expected-size approximation factor
        $
        \Omega\!\left(\frac{n^{1/(2k)}}{k}\right).
        $
        The same asymptotic lower bound holds for algorithms that achieve their approximation with probability at least $2/3$;
        \item there is a randomised algorithm that uses $O(n)$ demand queries per round and achieves approximation factor
        $
        O\!\left(\frac{n^{1/(k+1)}}{k}\right)
        $
        in expectation and with probability at least $2/3$.
\end{itemize}
\end{theorem}

\Cref{thm:few-rounds} follows from \Cref{thm:random-lb-expectation} and \Cref{ub:thm:main}. The upper-bound algorithm is the same independent-neighbour greedy algorithm underlying \cite{DobzinskiNO19}, but our analysis improves its dependence on $k$. 
We also show that this improvement is essentially tight for this particular algorithm: in the natural low-round regime, there are graphs on which its approximation factor is $\Omega(n^{1/(k+1)}/k)$, matching our upper bound up to constant factors.
This algorithm-specific optimality does not close the remaining gap for unrestricted demand-query algorithms. A general demand query reveals not only a returned neighbour but also certifies every vertex preceding it in the query order to be a non-neighbour. An arbitrary algorithm can correlate these prefixes across many left vertices and use this negative information in later rounds, whereas the independent-neighbour greedy algorithm uses demand queries only as random-neighbour samples. Thus closing the gap requires either a stronger lower bound that remains hard against such correlated-prefix strategies or a new algorithm that exploits these additional certificates.

For low-adaptivity demand query algorithms with $O(n)$ queries per round, the remaining gap is therefore between an $\Omega(n^{1/(2k)}/k)$ lower bound and an $O(n^{1/(k+1)}/k)$ upper bound. More generally, with $O(n^{1+\delta})$ queries per round, our expected-size lower and upper bounds are, respectively,
\[
\Omega\!\left(\frac{n^{\frac{1-k \delta}{2k}}}{k}\right)
\qquad\text{and}\qquad
O\!\left(\frac{n^{\frac{1-k\delta}{k+1}}}{k}\right).
\]

The improvement in the upper bound becomes especially significant when the number of rounds grows. Combined with our lower bound, it resolves the round complexity of constant-factor approximation up to constant factors.

\begin{theorem}[Constant-approximation] \label{thm:const-approx}
With $O(n)$ demand queries per round:
\begin{itemize}
    \item any randomised algorithm that outputs an $O(1)$-approximate maximum matching, even with probability $1/n$, needs $\Omega(\frac{\log n}{\log \log n})$ rounds;
    \item for every fixed $\varepsilon>0$, there is a randomised algorithm that outputs a $(2+\varepsilon)$-approximate maximum matching with high probability in $O(\frac{\log n}{\log\log n})$ rounds.
\end{itemize}
\end{theorem}

As mentioned, \cite{BravermanO17} also achieves this round complexity lower bound but for the communication model. 

The same randomised algorithm also yields a $(2+\varepsilon)$-approximate
minimum vertex cover with high probability in
$O(\log n/\log\log n)$ rounds and $O(n)$ demand queries per round,
using one additional round to identify unmatched left vertices with
unmatched neighbours (\Cref{ub:cor:vertex-cover}).  The corresponding
Demand-OR result incurs only polylogarithmic query overhead
(\Cref{cor:demand-or-vertex-cover}).

\paragraph*{Deterministic demand query.} We also obtain near-tight bounds for deterministic algorithms. Here the dependence on the number of rounds can be characterised up to logarithmic factors.

\begin{theorem}[Deterministic algorithms] \label{thm:intro-deterministic}
Let $q$ be the number of demand queries allowed per round. If $kq<n$, no deterministic $k$-round algorithm guarantees any finite multiplicative approximation on every graph. If $kq\ge n$, the optimal deterministic approximation factor is
$
\widetilde{\Theta}\!\left(\frac{n^2}{k^2q}\right).
$
More precisely, there is an upper bound
$
O\!\left(\frac{n^2\log^2 n}{k^2q}\right),
$
and every deterministic algorithm has approximation factor at least
$
\Omega\!\left(\frac{n^2}{k^2q\log(2+kq/n)}\right).
$
\end{theorem}

In particular, with $\Theta(n)$ demand queries per round, deterministic constant approximation requires and can be achieved in $\widetilde{\Theta}(\sqrt n)$ rounds. This gives a polynomial separation from the randomised round complexity above. For comparison, Nisan's deterministic demand-query
lower bound for constant approximation gives the requirement
$q=\Omega(n^2/(k^3\log^4 n))$~\cite[Corollary~2]{Nisan21}.
Our lower bound improves the polynomial
round dependence to $k^2$ for the unrestricted approximation problem.  By K\"onig's theorem, the same
near-tight deterministic tradeoff holds for constructing an approximate
minimum vertex cover (\Cref{cor:deterministic-certificates}).  Combined
with the randomised upper bound above, this gives a polynomial
separation in the round complexity of constant-approximation vertex
cover (\Cref{cor:deterministic-randomised-separation}).

\paragraph*{Demand-OR query.} Finally, we note that \cite{Nisan21} utilises a slightly weaker variant of the demand query, the demand-OR query, to establish lower bounds against \textit{deterministic} algorithms. A demand-OR query consists of a left vertex $u \in L$ and a subset $S \subseteq R$, and the oracle returns $1$ if there exists a vertex $v \in S$ such that $uv$ is an edge, and $0$ otherwise. While a demand query can simulate a demand-OR query in a single step (by placing $S$ at the prefix of the permutation), the converse simulation in \cite{Nisan21} relies on sequential binary search, which incurs an $O(\log n)$ blow-up in round complexity.

An interesting question is whether this round blow-up is truly necessary. Since we are concerned with the limited-round regime ($k \ll \log n$), such a blow-up is significant. While our lower bounds for demand queries naturally extend to the weaker demand-OR model, a naive simulation of our upper bound would suffer this $\log n$ round penalty. We demonstrate that this penalty is not inherent: using non-adaptive geometric subsampling together with an exact singleton-recovery primitive, we can trade a polylogarithmic increase in the number of queries per round for no increase in round complexity. This keeps the round complexity identical to the stronger demand query model.

\begin{theorem} \label{thm:intro-dor-query}
Let $q=n^{1+\delta}$. Using $k$ rounds of adaptivity and $\widetilde{O}(q)$ Demand-OR queries per round, one can obtain expected-size approximation factor
$
O\!\left(\frac{n^{\frac{1-k\delta}{k+1}}}{k}\right).
$
For $\delta\ge0$, the same asymptotic approximation holds with probability at least $2/3$. Moreover, for every fixed $\varepsilon>0$, a $(2+\varepsilon)$-approximation can be obtained with high probability in $O(\log n/\log\log n)$ rounds using $O(n\log^3 n)$ Demand-OR queries per round.
\end{theorem}

For deterministic Demand-OR queries, when $q\ge n$ and $k\ge2$ we obtain the same $\widetilde{\Theta}(n^2/(k^2q))$ approximation tradeoff as for deterministic demand queries. In contrast, in one round exactly $n^2$ deterministic Demand-OR queries are necessary and sufficient to guarantee any finite multiplicative approximation on every graph.  For minimum vertex cover, however, one round and
$q\ge n$ Demand-OR queries already suffice for an $O(n^2/q)$-approximation
(\Cref{prop:deterministic-or-vertex-cover}).

\newcommand{\oldbound}[1]{\textcolor{red!90!black}{#1}}
\newcommand{\newbound}[1]{\textcolor{blue!90!black}{#1}}

\begin{table*}[!t]
\centering
\footnotesize
\setlength{\tabcolsep}{3.5pt}
\renewcommand{\arraystretch}{1.25}

\begin{tabularx}{\textwidth}{@{}
    >{\raggedright\arraybackslash}p{0.14\textwidth}
    >{\raggedright\arraybackslash}p{0.39\textwidth}
    >{\raggedright\arraybackslash}p{0.20\textwidth}
    >{\raggedright\arraybackslash}X
@{}}

\toprule
\textbf{Setting} &
\textbf{Previous bounds} &
\textbf{This work} &
\textbf{Comments} \\
\midrule


\multicolumn{4}{@{}l}{
    \textbf{Randomised demand queries}
} \\
\midrule


$k=1$, $q=n$
&
$O(\sqrt n)$
\cite{dobzinski2014economic,Nisan21,shah2026sublinear}
\newline
$\Omega\bigl(n^{\oldbound{1/2-o(1)}}\bigr)$
\cite{dobzinski2014economic}
&
$\displaystyle
\newbound{\Theta}(\sqrt n)
$
&
Recovers the one-round result.
\\
\midrule


$k=1$, $q=n^{1+\delta}$
&
No general $\delta$-dependent tradeoff known.
&
$\displaystyle
\newbound{\Theta}\!\left(
n^{(1-\newbound{\delta})/2}
\right)
$
&
Tight $\delta$-dependence,
$-1\le\delta\le1$.
\\
\midrule


$k\ge2$, $q=n$
&
$\displaystyle
\Omega\!\left(
n^{1/\oldbound{5^{k+1}}}
\right)$
\cite{AlonNRW15}
\newline
$\displaystyle
O\!\left(n^{1/(k+1)}\right)$
\cite{dobzinski2014economic}
&
$\displaystyle
\Omega\!\left(
\frac{n^{1/\newbound{2k}}}{k}
\right)$
\newline
$\displaystyle
O\!\left(
\frac{n^{1/(k+1)}}{\newbound{k}}
\right)$
&
Stronger lower-bound exponent;
sharper greedy analysis.
\\
\midrule


$k\ge2$, $q=n^{1+\delta}$
&
No general round--query--approximation
tradeoff known.
&
$\displaystyle
\Omega\!\left(
\frac{n^{1/(2k)\newbound{-\delta/2}}}{k}
\right)$
\newline
$\displaystyle
O\!\left(
\frac{n^{(1-\newbound{k\delta})/(k+1)}}
{k}
\right)$
&
New $\delta$-dependence,
$-1\le\delta\le1/k$.
\\
\midrule


$\alpha=O(1)$, $k\ge1$
&
$\displaystyle
\Omega\!\left(
\frac{n^{1+1/k}}{\oldbound{k^4}}
\right)$ bits
\cite{BravermanO17}
\newline
Implies
$\displaystyle
q=\Omega\!\left(
\frac{n^{1+1/k}}{\oldbound{k^5}\log n}
\right)$
&
$\displaystyle
q=\Omega\!\left(
\frac{n^{1+1/k}}{\newbound{k^2}}
\right)$
\newline
$\displaystyle
q=O\!\left(
\frac{n^{1+1/k}}{k}
\right)$
&
Near-tight query budget.
\\
\midrule


$\alpha=O(1)$, $q=n$
&
$k=\Omega(\log\log n)$
\cite{AlonNRW15}
\newline
$k=\Omega(\log n/\log\log n)$
for near-optimal approximation
\cite{BravermanO17}
\newline
$k=O(\log n)$
\cite{dobzinski2014economic,Nisan21}
&
$\displaystyle
k=\Theta\!\left(
\frac{\log n}{\newbound{\log\log n}}
\right)$
&
Optimal round complexity.
\\


\midrule
\multicolumn{4}{@{}l}{
    \textbf{Deterministic queries}
} \\
\midrule


Demand, $\alpha=O(1)$
&

$\displaystyle
q=\Omega\!\left(
\frac{n^2}{{\oldbound{k^3}}\log^4 n}
\right)$
\cite{Nisan21}
&
$\displaystyle
q=\widetilde{\Theta}\!\left(
\frac{n^2}{{\newbound{k^2}}}
\right)$
&
Near-tight tradeoff for all
approximation factors;
requires $kq\ge n$.
\\
\midrule


Demand,
$q=n$,
$\alpha=O(1)$
&
$\displaystyle
k=\Omega\!\left(
\frac{n^{\oldbound{1/3}}}{\log^{4/3}n}
\right)$
\cite{Nisan21}
&
$\displaystyle
k=\widetilde{\Theta}\!\left(
n^{\newbound{1/2}}
\right)$
&
Polynomial separation
from randomized rounds.
\\
\midrule


Demand-OR,
$k\ge2$,
$q\ge n$,
$\alpha=O(1)$
&
$\displaystyle
q=\Omega\!\left(
\frac{n^2}{\oldbound{k^3}\log n}
\right)$
\cite{Nisan21}
&
$\displaystyle
q=\widetilde{\Theta}\!\left(
\frac{n^2}{\newbound{k^2}}
\right)$
&
Matching deterministic
tradeoff, up to polylogs.
\\

\bottomrule
\end{tabularx}

\caption{
Comparison with previous results.
The table highlights improvements over previous work rather than listing all results of this paper.
Here $k$ denotes the number of rounds,
$q$ the number of queries per round,
and $\alpha$ the approximation factor.
Randomised approximation factors are stated in expectation.
The one-round randomised result holds for
$-1\le\delta\le1$,
and the general $k$-round result for
$-1\le\delta\le1/k$.
All approximation factors are understood to be at least one.
The deterministic demand-query tradeoff assumes
$kq\ge n$; otherwise no finite approximation is possible.
\oldbound{Red} highlights terms in previous bounds that have been improved,
while \newbound{blue} highlights the corresponding improvements in this work.
}
\label{tab:results}
\end{table*}

\subsection{Related Work}

Maximum matching and minimum vertex cover have been extensively studied in
restricted-access graph models. In particular, a large body of work studies
estimating the maximum matching size (equivalently vertex cover size in bipartite graphs) using adjacency-list and adjacency-matrix
queries \cite{chen2020sublinear,kapralov2020space,behnezhad2021optimal,
behnezhad2023beating,behnezhad2023sublinear,bhattacharya2023dynamic,
BhattacharyaKS23,behnezhad2024approximating,azarmehr2025tight}.
In these models, constructing an approximate matching admits simple and
essentially tight query bounds, motivating the extensive study of the more
challenging problem of estimating the matching size.
Related problems have also been studied in one-way and simultaneous
communication models \cite{goelkk12,assadikly16}, where the input is
distributed among multiple parties.
Matching has also been studied in distributed sketching and massively parallel
computation models \cite{assadi2020lower,assadi2025rounds,czumaj2018round,ghaffari2018improved}.

Matching and vertex cover have also been extensively studied in graph
streams. The recent breakthrough lower bound for matching in
insertion-only streams implies that the simple greedy algorithm is optimal \cite{assadisemi2026,assadi2026greedy}.
For dynamic streams, where both edge insertions and deletions are allowed,
\cite{assadikly16,darkk20,assadi2022asymptotically,naidus22space}
establish asymptotically optimal space bounds for approximating
maximum matching and minimum vertex cover.
Multi-pass streaming algorithms for matching have also been extensively studied, with \cite{fischer2022deterministic,assadi2023hidden} establishing important upper and lower bounds for insertion-only streams, and \cite{assadi2026settling} settling the pass complexity of approximate matching in dynamic streams.
Unlike our setting, the complexity measure is space rather than
queries.


\subsection{Organisation} 
We provide a technical overview covering \Cref{thm:few-rounds},  \Cref{thm:intro-deterministic} and \Cref{thm:intro-dor-query} in \Cref{sec:technical-overview}. 
In \Cref{sec:prelims}, we define the required notation, tools, and the model of demand queries. Subsequently, in \Cref{sec:k-round-lb}, we prove our main result, which is the randomised lower bound. In \Cref{ub:sec:upper-bound}, we give the randomised upper bounds and their improved analysis compared to previous work. In \Cref{sec:demand-or}, we give our algorithms and bounds for Demand-OR queries.
Finally, in \Cref{sec:deterministic}, we prove the deterministic round--query tradeoff.
\section{Technical overview}
\label{sec:technical-overview}

We explain the main ideas behind our bounds. For simplicity, we consider
$O(n)$ demand queries per round, corresponding to $\delta=0$, and suppress
logarithmic factors in the query and approximation bounds. Throughout,
the graph has $n$ vertices on each side, and $\mu$ denotes its maximum
matching size.

\subsection{Randomised lower bound}

We construct a random bipartite graph $G=(L\cup R,E)$ with $n$ vertices
on each side. We may assume $kn^{-\frac{1}{2k}}\le 1/2$, since otherwise the
claimed approximation lower bound is constant, which is a triviality. First, we choose a random
partition $R=T_0\sqcup\cdots\sqcup T_k$. Each right vertex $v$ chooses
exactly one layer, independently of the other right vertices, with
\[
    \Pr[v\in T_i]=
    \begin{cases}
        n^{-\frac{1}{2k}}, & 0\le i<k,\\
        1-kn^{-\frac{1}{2k}}, & i=k.
    \end{cases}
\]
Conditional on this partition, we include each pair $uv$, with $u\in L$
and $v\in T_i$, independently with probability $n^{-\frac{i}{k}}$. Thus every
vertex in $T_0$ is adjacent to all of $L$, and the expected degree of a
vertex in $T_i$ is $n^{1-\frac{i}{k}}$, decreasing by a factor of $n^{\frac{1}{k}}$
between successive layers. The algorithm is given the labeled vertex
sets, but the layer assignments are hidden.

Each of the first $k$ layers has expected size $n^{1-\frac{1}{2k}}$, so their
union has expected size $kn^{1-\frac{1}{2k}}$. These layers provide many
easily discovered edges, but any matching can use each of their right
vertices only once. Meanwhile, $T_k$ has linear size with high
probability. Although its vertices have expected degree one, the graph
between $L$ and $T_k$ contains a matching of size $\Omega(n)$ with high
probability. Finding a matching substantially larger than the total size
of the denser layers therefore requires discovering edges into this
sparse layer.

\begin{figure}[h]
    \centering
    \scalebox{0.85}{

\begin{tikzpicture}[
    x=1cm, y=1cm, font=\small,
    layer/.style={draw, rounded corners=2pt, minimum width=1.85cm,
                  minimum height=0.72cm, line width=0.65pt, align=center},
    smalllabel/.style={font=\footnotesize, text=colorEight},
    vertex/.style={circle, fill=colorEight, inner sep=1.35pt},
    matchingvertex/.style={circle, fill=colorFive, inner sep=1.35pt},
    stub/.style={draw=colorNine, line width=0.9pt},
    matching/.style={draw=colorFive, line width=1.15pt}
]

\node[draw=colorFive!80!black, fill=colorTwo!13,
      line width=0.7pt, rounded corners=2pt,
      minimum width=1.3cm, minimum height=7.5cm] (L) at (0,0) {};
\node at (0,3.25) {$L$};
\node[smalllabel, anchor=south] at (0,3.87) {$|L|=n$};

\node[smalllabel, anchor=south] at (6.15,3.87) {$|R|=n$};
\node[layer, fill=colorOne!23, draw=colorOne!80!black] (T0) at (6.15,3.17) {$T_0$};
\node[layer, fill=colorSix!13, draw=colorSix!76!black] (T1) at (6.15,1.87) {$T_1$};
\node[layer, fill=colorFour!23, draw=colorNine] (Ti) at (6.15,0.57) {$T_i$};
\node[layer, fill=colorThree!15, draw=colorThree!77!black] (Tk1) at (6.15,-0.73) {$T_{k-1}$};
\node[layer, fill=colorTwo!22, draw=colorFive!80!black,
      minimum height=1.56cm] (Tk) at (6.15,-2.43) {$T_k$};

\foreach \N in {T0,T1,Tk1}{
    \node[smalllabel, anchor=west] at ($(\N.east)+(0.2,0)$) {$n^{1-\frac{1}{2k}}$};
}
\node[smalllabel, anchor=west] at ($(Tk.east)+(0.2,0)$) {$\Theta(n)$};

\node[smalllabel, anchor=east] at (4.3,3.77) {expected degree};
\foreach \yy/\degree in {3.17/n,1.87/{n^{1-1/k}},0.57/{n^{1-i/k}},-0.73/{n^{1/k}}}{
    \node[smalllabel, anchor=east] at (4.35,\yy) {$\degree$};
    \node[vertex] at (5.18,\yy) {};
    \draw[stub] (5.18,\yy) -- ++(-0.42,0.22);
    \draw[stub] (5.18,\yy) -- ++(-0.42,-0.22);
}
\node[smalllabel, anchor=east] at (4.35,-1.79) {$1$};
\node[vertex] at (5.18,-1.79) {};
\draw[stub] (5.18,-1.79) -- ++(-0.42,0.18);
\draw[stub] (5.18,-1.79) -- ++(-0.42,-0.18);

\node[smalllabel, anchor=south] at (9.4,3.84)
    {degree drops by factor $n^{1/k}$};
\draw[-{Latex[length=1.8mm]}, color=colorNine, thick]
    (9.4,3.57) -- (9.4,-1.00);
\draw[decorate, decoration={brace,amplitude=4.5pt}, color=colorNine, thick]
    (11.0,3.55) -- (11.0,-1.12);
\node[smalllabel, anchor=west, align=left] at (11.38,0.96)
    {$k$ dense layers\\total size $k n^{1-\frac{1}{2k}}$};

\foreach \yy in {-2.03,-2.26,-2.49,-2.72,-2.95}{
    \pgfmathsetmacro{\rightoffset}{\yy+2.43}
    \draw[matching] ($(L.east)+(0,\yy)$) -- ($(Tk.west)+(0,\rightoffset)$);
    \node[matchingvertex] at ($(L.east)+(0,\yy)$) {};
    \node[matchingvertex] at ($(Tk.west)+(0,\rightoffset)$) {};
}
\node[smalllabel, anchor=west, align=left] at (9.36,-2.45)
    {matching of size\\$\Omega(n)$ w.h.p.};

\node[text=colorNine, anchor=west] at (6,0) {$\vdots$};

\node[text=colorNine, anchor=west] at (6,1.3) {$\vdots$};

\end{tikzpicture}
    }
    \caption{The lower bound construction.}

    \label{fig:rand-dem-query-lb}
\end{figure}
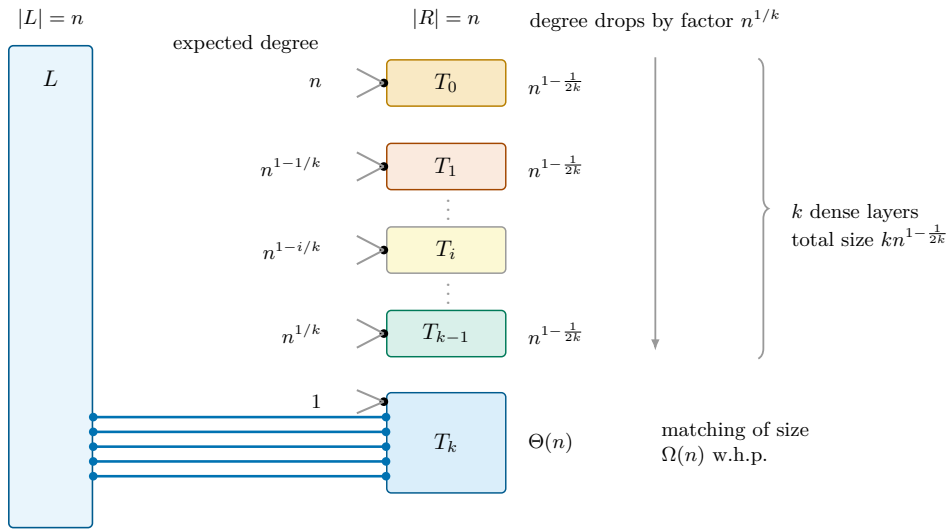

{The main difficulty is that a demand query reveals not only an edge
but also certifies every preceding pair to be a non-edge. We therefore
track both positive and negative information obtained by the algorithm.
At the end of round $r$, we discard any right vertex that is incident to a certified edge or to more than $n^{\frac{r}{k}}$ certified non-edges. Let $R_r$ denote the set of remaining right vertices. Thus, every vertex in $R_r$ has no certified incident edge and at most $n^{\frac{r}{k}}$ certified incident non-edges. We will keep track of the number of discarded vertices as it will help us bound the size of matching found.}

{Fix a transcript of the first $r$ rounds and let $v\in R_r$. Since $v$ has at most $n^{\frac{r}{k}}$
certified non-edges incident to it, this information
reduces the probability that $v$ belongs to $T_r$ by at most
a constant factor. Consequently,
\[
    \Pr[v\in T_r\mid\text{transcript}]
    =\Omega(n^{-\frac{1}{2k}}).
\] }

Since a vertex in $T_r$ is adjacent to each left vertex with
probability $n^{-\frac{r}{k}}$, every uncertified pair $(u,v)$ with $v\in R_r$ is an edge with conditional probability at least
$\Omega(n^{-\frac{r}{k}-\frac{1}{2k}})$.
It follows that a demand query certifies at most
$O(n^{\frac{r}{k}+\frac{1}{2k}})$ new pairs involving vertices in
$R_r$ in expectation. With $O(n)$ queries in round $r+1$, the expected
number of newly certified non-edges involving these vertices is
therefore $O(n^{1+\frac{r}{k}+\frac{1}{2k}})$.

{A vertex of $R_r$ can exceed the next
threshold $n^{\frac{r+1}{k}}$ only after acquiring
$\Omega(n^{\frac{r+1}{k}})$ new certified non-edges. Thus the expected
number of vertices discarded for this reason at the end of round $r+1$ is
\[
    O\left(
        \frac{n^{1+\frac{r}{k}+\frac{1}{2k}}}{n^{\frac{r+1}{k}}}
    \right)
    =O(n^{1-\frac{1}{2k}}).
\]}

{A vertex of $R_r$ can also be discarded because it acquired a certified edge adjacent to it. We will only count edges in incident to $T_k$, accounting for dense layers separately. Since a pair incident to a vertex
in $T_k$ is an edge with probability $1/n$, each uncertified pair has
probability at most $1/n$ of being an edge into $T_k$. Consequently,
among the pairs certified in round $r+1$, the expected number of
edges into $T_k$ is at most
\[
    O\left(
        \frac{n^{1+\frac{r}{k}+\frac{1}{2k}}}{n}
    \right)
    = O(n^{1-\frac{1}{2k}}),
\]
where we use $r\le k-1$. This also bounds the expected number of
vertices in $R_r\cap T_k$ discarded because of certified edges.
}

{We therefore discard $O(n^{1-\frac{1}{2k}})$ vertices per round in expectation, and hence $O(kn^{1-\frac{1}{2k}})$ over all rounds. The dense layers together contain $O(kn^{1-\frac{1}{2k}})$ vertices in expectation. Thus, the total
number of right vertices discarded because of certified edges or the
non-edge thresholds is $O(kn^{1-\frac{1}{2k}})$ in expectation.}

{The set of discarded vertices, together with all vertices in the denser
layers, forms a vertex cover of the certified edges: every certified
edge either has its right endpoint in a denser layer or has its right
endpoint discarded. Therefore, any matching consisting only of
certified edges has expected size
$O(kn^{1-\frac{1}{2k}})$. Note that a similar bound also holds true if the algorithm is allowed to output uncertified edges as part of the matching.}

{Finally, since the graph has a matching of size $\Omega(n)$ with high
probability, comparing this with the expected matching size obtainable
by the algorithm yields an approximation lower bound of
$
    \Omega\left(\frac{n^{\frac{1}{2k}}}{k}\right).
$
The same lower bound holds for algorithms that succeed with constant
probability.}

\subsection{Randomised upper bound}
We use the independent-neighbour greedy algorithm underlying
\cite{dobzinski2014economic}: In each round, every unmatched left vertex
independently samples a uniformly random neighbour, when one exists,
among the right vertices unmatched at the start of that round.
A demand query implements this by placing those right vertices first
in a uniformly random order.
We process the returned edges greedily, without further queries. Our
analysis improves the dependence on $k$ by retaining the progress made
across all rounds. This is the primary difference from the analysis of \cite{dobzinski2014economic} which lets us match the lower bound for constant approximation with $q = O(n)$.

Let $X$ be the final matching size. At least $\mu-2X$ edges of a fixed
maximum matching have both endpoints still unmatched. Consider one such edge $uv$ whose endpoints remain unmatched
throughout the execution. Let $R_r$ denote the unmatched right
vertices at the beginning of round $r$, and write
$d_r=|N(u)\cap R_r|$. Since $u$ remains unmatched, it participates
in every round. In round $r$, some neighbours of $u$ may be matched
to other left vertices before $u$ is processed. Let $f_r$ be the
fraction of its $d_r$ neighbours that are still available when
$u$ is processed. Observe that every neighbour available at the
beginning of round $r+1$ must also have been available when $u$
was processed in round $r$. Therefore,
\[
    f_r \ge \frac{d_{r+1}}{d_r}.
\]
Multiplying these inequalities over all $k$ rounds gives
\[
    \prod_{r=1}^{k} f_r
    \ge \prod_{r=1}^{k}\frac{d_{r+1}}{d_r}
    =\frac{d_{k+1}}{d_1}.
\]
Now, at most $X$ right vertices are matched during the entire
execution, so $d_1-d_{k+1}\le X$. Moreover, $d_{k+1}\ge 1$,
since $v$ remains unmatched and is a neighbour of $u$.
Consequently,
\[
    \prod_{r=1}^{k} f_r
    \ge \frac{d_{k+1}}{d_{k+1}+X}
    \ge \frac{1}{X+1}.
\]
Finally, by the arithmetic--geometric mean inequality,
\[
    \sum_{r=1}^{k}f_r
    \ge k\left(\prod_{r=1}^{k}f_r\right)^{1/k}
    \ge \frac{k}{(X+1)^{1/k}}.
\]
The significance of this inequality comes from understanding what $f_r$ represents.
Although all samples in a round are drawn simultaneously, for the sake of analysis, we
expose them in the fixed greedy processing order. Since the
samples of different left vertices are independent, conditional
on the previously exposed samples, the sample of a vertex $u$
remains uniform among its neighbours unmatched at the beginning
of the round. Some of these neighbours may already have been
matched by the time $u$ is processed. Consequently, the
conditional probability that $u$ is matched is precisely
$f_{r}$, the fraction of its round-start neighbours that
remain unmatched when $u$ is processed. By linearity of
expectation, the expected sum of these conditional probabilities
over all vertices and rounds equals the expected size of the
matching produced by the algorithm. Thus, if many edges of the
optimum matching remain unmatched, the preceding inequality
guarantees substantial accumulated progress across the $k$ rounds,
which in turn forces the expected output matching to be large.

More concretely, let $f_{r,u}$ denote the conditional
probability that $u$ gets matched when processed in round $r$,
taking $f_{r,u}=0$ if $u$ is already matched.
Since at least $(\mu-2X)$ edges of a fixed maximum
matching survive, the preceding argument gives
\[
    A=\sum_{u\in L}\sum_{r=1}^{k}f_{r,u}
    \ge \frac{k(\mu-2X)}{(X+1)^{1/k}}.
\]
Even though $A$ is not the same random variable as $X$
\footnote{Consider a one-round execution on $K_{2,2}$.
$A$ always takes value $3/2$, whereas $X$ takes values
$1$ and $2$, each with probability $1/2$.},
the crucial observation is that
$\mathbb E[A]=\mathbb E[X]$. Call this common expectation $m$.
If $m\ge\mu/8$, we already have a constant approximation.
Otherwise, Markov's inequality gives
$\Pr[X\le 2m]\ge 1/2$. On this event,
$\mu-2X\ge\mu/2$, and hence $A\ge\frac{k\mu}{2(2m+1)^{1/k}}$.
Taking expectations, we obtain
\[
    m=\mathbb E[A]
    \ge\frac{k\mu}{4(2m+1)^{1/k}}.
\]
Since $m\ge 1$ whenever $\mu>0$, this implies
$m=\Omega((k\mu)^{k/(k+1)})$.
Combining the two cases, and using $\mu\le n$,
the resulting expected-size approximation factor is $O(\frac{n^{1/(k+1)}}{k})$.



\subsection{Deterministic lower bound}
Our proof builds on the adversarial approach of \cite{Nisan21}:
answer the queries while keeping two graphs consistent with every
answer, one with a large matching and another whose edges meet a
small vertex cover. Nisan preserves a perfect matching in the first
graph by ensuring that no vertex loses too many possible neighbours.
His adversary uses a small hitting set to answer large queries,
while the remaining small queries are handled by declaring the
corresponding pairs absent. However, if too many such pairs are
incident on a vertex, declaring them absent would violate this
invariant. Nisan handles these ``heavy'' vertices by instead fixing
all their remaining undecided incident pairs present. This preserves
the perfect matching, but each heavy vertex now contributes to
the vertex cover of the second graph. The resulting lower bound
is therefore limited by the need to balance the size of the
hitting sets against the number of heavy vertices.

More concretely, let $\theta$ be a threshold separating large
and small queries in each round. Nisan answers all large queries
using a hitting set of size $O(n\log n/\theta)$. The remaining
small queries involve at most $q\theta$ undecided pairs.
To ensure that no vertex loses more than $n/(2k)$ possible
neighbours in a round, any vertex incident on more than
$n/(2k)$ of these pairs is declared heavy. There can be at
most $O(kq\theta/n)$ such vertices, and all their remaining
incident pairs are fixed present. Consequently, the vertex
cover of the small graph has size at most
\[
    O\left(
        k\left(\frac{n\log n}{\theta}
        +\frac{kq\theta}{n}\right)
    \right).
\]
Balancing these two terms gives
$\theta=\Theta(n\sqrt{\log n/(kq)})$, resulting in a cover
of size $\widetilde O(k^{3/2}\sqrt q)$. For $q=\Theta(n)$,
this yields an approximation lower bound of only
$\widetilde\Omega(\sqrt n/k^{3/2})$.

Our key change is to preserve only a matching of size $\Omega(n)$,
rather than a perfect matching. This allows us to retire a small
number of left vertices whose neighbourhoods become too restricted:
we fix all their remaining undecided pairs absent and answer future
queries from their fixed neighbourhoods. Retirement adds no edges,
so it does not enlarge the cover of the small graph. Instead of
protecting every vertex separately, we bound the total number of
pairs fixed absent by queries to active vertices, and use this to
bound how many vertices must be retired.

For $O(n)$ queries per round, choose a prefix threshold
$\theta=cn/k$, where $c>0$ is sufficiently small. At the start of
each round, consider the undecided pairs preceding each query's
first known edge, or all undecided pairs if no edge is known.
For every query with at least $\theta$ such pairs, take the right
endpoints of the first $\theta$. A set $T_r$ of
$\widetilde O(k)$ right vertices intersects all but at most
$n/(8k)$ of these prefixes. We retire the left vertices with a
missed prefix, and fix every undecided pair between the remaining
active left vertices and $T_r$ present before answering any query.
These commitments ensure that each query to an active vertex
fixes at most $\theta$ new pairs absent. At the end of the round,
we also retire any active left vertex that has accumulated at
least $n/2$ absent pairs.

Over all $k$ rounds, queries to active vertices fix at most
$O(kn\theta)$ pairs absent. Thus only $O(k\theta)=O(cn)$ vertices
are retired by the absent-pair threshold, in addition to at most
$n/8$ retired because of missed prefixes. More than $n/2$ left
vertices therefore remain active, each with more than $n/2$
possible neighbours. Meanwhile, every pair fixed present has its
right endpoint in
\[
    T=\bigcup_{r=1}^k T_r,
    \qquad |T|=\widetilde O(k^2).
\]
Making all undecided pairs absent gives a graph whose matching
size is at most $|T|$. Making them present gives a graph with a
matching of size at least $\lfloor n/2\rfloor$. Both completions
preserve every answer, including all non-edges certified by its
prefix, so a deterministic algorithm must return the same
matching on both.

This gives an approximation lower bound of
$\widetilde\Omega(n/k^2)$, compared with
$\widetilde\Omega(\sqrt n/k^{3/2})$ from \cite{Nisan21}.
The polynomial improvement comes from eliminating the extra cover
vertices needed to protect heavy vertices, allowing the larger
threshold $\theta=\Theta(n/k)$. Working directly with demand-query
prefixes also avoids the logarithmic round overhead in working with Demand-OR queries (as Nisan did) instead of demand queries.



\subsection{Deterministic upper bound}
A sequential greedy algorithm processes the left vertices one at
a time, matching each vertex to a currently unmatched neighbour
whenever possible. This produces a maximal matching and hence a
$2$-approximation, but its direct implementation requires up to
$n$ rounds: each demand query depends on which right vertices
were matched by earlier queries. Our approach is to recover enough
neighbours in advance to process several left vertices without
making further queries. We first explain how to execute such a
batch, and then how to trade approximation quality for fewer rounds.

More concretely, consider a batch $B$ of at most $b$ left vertices,
and let $S$ be the right vertices currently unmatched. For each
$u\in B$, recover $b$ distinct neighbours in $S$, or all its
neighbours in $S$ if fewer than $b$ exist. We can then process
the vertices of $B$ greedily using only these lists. If a vertex
has a list of size $b$, at most $b-1$ of its listed neighbours
can have been used by earlier vertices in the batch, so at least
one remains available. If its list has fewer than $b$ neighbours,
the list is complete, and we can determine exactly whether an
available neighbour remains. Thus the recovered lists suffice
to execute a valid sequential greedy process within the batch,
without any further queries.


To recover these lists, we use a binary-tree search procedure
that, given a vertex $u$ and a set $S$ of available right vertices,
finds $b$ distinct neighbours of $u$ in $S$, or all such neighbours
if fewer than $b$ exist. This takes $O(\log n)$ rounds and
$O(b)$ demand queries per round. Running the procedure in
parallel for all vertices of $B$ therefore allows us to process
an entire batch in $O(\log n)$ rounds, using $O(b^2)$ queries
per round.

To obtain the approximation--round tradeoff, partition $L$ into
$g$ balanced groups $L_1,\ldots,L_g$. For each group $L_i$, run
a separate greedy execution on $G[L_i,R]$, processing its vertices
in batches of size $b=\Theta\left(\sqrt{n/g}\right)$.

Each group independently computes a maximal matching in its
subgraph $G[L_i,R]$, without coordinating with the other groups.
Thus, the same right vertex may be matched in different groups,
but not more than once within any individual group.
We eventually return the largest of these matchings. Their
total query cost per round is $O(gb^2)=O(n)$.

Each group contains $O(n/g)$ left vertices and therefore needs
$O(n/(gb)+1)$ batches. Consequently, all executions finish in
\[
    O\left(\left(\frac{n}{gb}+1\right)\log n\right)
    =O\left(\sqrt{\frac{n}{g}}\log n\right)
\]
rounds.

Let $M_i$ be the matching produced by group $L_i$. We return
the largest of these matchings, rather than their union, since
different groups may have used the same right vertices.
To see why this gives a $2g$-approximation, fix a maximum
matching $M^\star$ of $G$. Each $M_i$ is maximal in $G[L_i,R]$,
so every edge of $M^\star$ whose left endpoint belongs to $L_i$
meets an edge of $M_i$. Each edge of $M_i$ can meet at most
two edges of $M^\star$. Summing over the groups gives
\[
    \sum_{i=1}^g |M_i|\ge \frac{\mu}{2},
    \qquad
    \max_{1\le i\le g}|M_i|
    \ge \frac{\mu}{2g}.
\]
Increasing the number of groups thus reduces the number of
batches processed sequentially, at the cost of a proportionate
loss in approximation.

For $k$ at least a sufficiently large multiple of $\log n$,
choosing
\[
    g=\Theta\left(\max\left\{
        1,\frac{n\log^2 n}{k^2}
    \right\}\right)
\]
fits the executions within $k$ rounds. This gives approximation
factor
\[
    O\left(\max\left\{1,\frac{n\log^2 n}{k^2}\right\}\right)
    =\widetilde O\left(\max\left\{1,\frac{n}{k^2}\right\}\right),
\]
matching our deterministic lower bound up to logarithmic factors.

\subsection{Demand-OR queries}
Our randomised demand query algorithm uses demand queries only to sample a uniform neighbour inside a specified set of right vertices. Thus, to obtain
the same algorithms using Demand-OR queries, it suffices to simulate
the following primitive: given a left vertex $u\in L$ and a set
$S\subseteq R$, return a uniformly random neighbour of $u$ in $S$,
if one exists, and otherwise report that no neighbour exists.

We prove a simulation lemma showing that this primitive can be implemented using
$O(\log^2 n)$ Demand-OR queries in a single nonadaptive round.
The algorithm detects an empty neighbourhood with certainty.
Otherwise, it returns a genuine neighbour with constant probability,
and its output is exactly uniform conditional on success.
This avoids the $O(\log n)$ adaptive rounds needed by the trivial binary-search
simulation, at the cost of an $O(\log^2 n)$ blowup in the number of queries.

We now explain the main idea behind this simulation, which uses well-known
techniques from the $L_0$-sampling literature~\cite{monemizadeh20101,jowhari2011tight,AhnGM12-soda,ghosh2024new}.
The idea is to guess the number $d$ of neighbours of $u$ in $S$
in powers of two, using guesses $1,2,4,\ldots$.
For each guess $g$, sample the vertices of $S$ independently with
probability $1/g$. One of these guesses satisfies $d\le g<2d$,
so its sampling rate is $\Theta(1/d)$ and isolates exactly one
neighbour with constant probability. To detect and recover this
neighbour, label the right vertices in binary. Query the sampled
set and, for each bit position, its two subsets consisting of
vertices with bit zero or bit one.
A unique neighbour gives exactly one positive answer at every position.
Two distinct neighbours give two positive answers at some position.
The test therefore verifies uniqueness and recovers a genuine neighbour.
All sampling scales and all bit queries are fixed before any answers
are received, so the entire procedure takes one round.

By symmetry, the first successful trial in a fixed ordering of independent
trials returns an exactly uniform neighbour, conditional on success.
A constant number of independent repetitions per guess ensures high constant success probability. 
The $O(\log n)$ guesses and $O(\log n)$ queries per trial give
$O(\log^2 n)$ queries per simulated call. Consequently, applying
our demand-query upper bound gives the same asymptotic approximation
using $k$ rounds and $O(q\log^2 n)$ Demand-OR queries per round.
Additional parallel repetitions yield the high-probability
$(2+\eps)$-approximation in $O(\log n/\log\log n)$ rounds
using $\widetilde O(n)$ queries per round.

\section{Prelimilaries} \label{sec:prelims}


\paragraph*{Notation.}
For a bipartite graph $G=(V,E)$ where $V=L \sqcup R$. For each vertex $v\in V$, we denote by $N(v)=\{u\in V : (u,v)\in E\}$ its neighbourhood and by $\deg(v)=|N(v)|$ its degree. For a set of vertices $S$ we define $N_S(v) = N(v) \cap S$ and $\deg_S(v)=|N_S(v)|$.

Let $\mu(G)$ be the size of the maximum matching of the graph $G$. 
For a subset $F$ of edges in $E$, we use $V(F)$ to denote the vertices incident on $F$; similarly, for a set $U$ of vertices, $E(U)$ denotes the edges incident on $U$. 
For sets $A,B$ of vertices, $E(A,B)$ denotes the edges between $A$ and $B$. 
We further use $G[U]$ for any set $U$ of vertices to denote the induced subgraph of $G$ on $U$. 

\paragraph*{General Inequalities} We use the following inequality a number of times in this work.
\begin{proposition}\label{lem: 1/4 1/e}
    $\frac{1}{4} \leq (1-\frac{1}{x})^x\leq \frac{1}{e}$ for all $x \geq 2$.
\end{proposition}

\paragraph*{The Demand Query Model} The demand query model is defined over a bipartite graph $G = (L \cup R, E)$ with $|L| = |R| = n$. 
A single demand query $q$ is specified by an ordered pair $q = (u_q, \sigma_q)$ where $u_q\in L$ and $\sigma_q$ is a permutation of the vertex set $R$. Upon issuing this demand query, the oracle returns 
$$\dem(q) = \arg\min_{v\in R} \{\sigma_q^{-1}(v): u_qv\in E\}.$$ 
Thus, this demand query certifies that there is an edge between $u_q$ and $\dem(q)$ in $G$.
 It also certifies that $u_qv$ is not an edge in $G$ for all $\{v: \sigma_q(v) < \sigma_q(\dem(q))\}$. If $N(u) = \emptyset$, the oracle returns $\perp$. This interface generalises the standard neighbour query by allowing the algorithm to specify an arbitrary preference hierarchy over the set of potential edges incident to a vertex.
%
Note that this query does not give any other information about $u_qv$ edges $\forall v:\sigma_q(v) > \sigma_q(\dem(q))$. Therefore, given query set $Q$, we can certify $(u,v)\in L\times R$ by either edges, non-edges, as in  the following definition, or uncertified edges.
\begin{definition} \label{def:yes-no-uncert}
    Given query set $Q$ and their responses, we certify the edges of $G$ in the following way:
    \begin{itemize}
    \item \textbf{Graph edges:} Certify $(u,v)\in \yes(Q)$ if there exist $q\in Q$ such that $u = u_q$ and $v= \dem(q)$.
    \item \textbf{Graph non-edges:} Certify $(u,v) \in \no(Q)$ if there exists $q \in Q$ such that $u = u_q$ and $\sigma_q(v) < \sigma_q(\dem(q))$.
\end{itemize}
\end{definition}

For the remaining pairs $(u,v)$, we are uncertain whether it belongs to the graph or not. We call them \textbf{uncertified pairs}. We will also use the following definition in the proof.
%
%
\begin{definition}[$\yes(v, Q), \no(v, Q)$] \label{def:yes-no-edge-v}
    We define $\yes(v, Q)$ to be the set of all edges in $\yes(Q)$ incident to $v$ and $\no(v, Q)$ to be the set of all edges in $\no(Q)$ incident to $v$.
\end{definition}

Although \Cref{def:yes-no-edge-v} is defined for any vertex $v \in L \cup R$, we will primarily use this definition where $v \in R$.


\paragraph*{Probabilistic Tools} 
 
 We use the following standard concentration inequalities. The first is a standard form of Chernoff bounds. 
 
 \begin{proposition}[Chernoff bound; c.f.~\cite{DubhashiP09}]\label{prop:chernoff}
 	Suppose $X_1,\ldots,X_m$ are $m$ independent random variables with range $[0,1]$ each. Let $X := \sum_{i=1}^m X_i$ and $\mu_L \leq \expect{X} \leq \mu_H$. Then, for any $\eps > 0$, 
 	\[
 	\Pr\paren{X >  (1+\eps) \cdot \mu_H} \leq 
    \exp\paren{-\frac{\eps^2 \cdot \mu_H}{3+\eps}} \quad \textnormal{and} \quad 
    \Pr\paren{X <  (1-\eps) \cdot \mu_L} \leq \exp\paren{-\frac{\eps^2 \cdot \mu_L}{2+\eps}}.
 	\]
 \end{proposition}

We use the term ``with high probability'' to mean with probability at least $1-n^{-c}$ for some large constant $c > 0$, which can be made arbitrarily large by increasing the queries of our algorithms by a constant factor.

\newcommand{\YES}{\operatorname{YES}}
\newcommand{\NO}{\operatorname{NO}}
\newcommand{\New}{\operatorname{New}}
\newcommand{\Exist}{\operatorname{Exist}}
\newcommand{\Ber}{\operatorname{Ber}}
\newcommand{\Bin}{\operatorname{Bin}}
\section[A Randomised k-round Lower Bound]{A Randomised $k$-round Lower Bound}
\label{sec:k-round-lb}

In this section, we prove our main lower bound for randomised algorithms
with a bounded number of rounds of demand queries.  Throughout the section,
$G=(L\sqcup R,E)$ is a bipartite graph with $|L|=|R|=n$, and $\mu(G)$
denotes its maximum matching size.  All logarithms are natural unless a
base is specified.  If an algorithm outputs a set of pairs $M$, let
$W=|M|$ when $M$ is a matching in $G$, and let $W=0$ otherwise.  Thus an
output-size bound applies to every valid output matching, including one
that contains edges not certified by the queries.

\begin{theorem}\label{thm:k-round-lb}
Let $1\le k\le\lfloor\log_2 n\rfloor$, $-1\le\delta\le 1/k$, and
$\beta:=1/(2k)-\delta/2$.  Any randomised algorithm that uses $k$ rounds
and $O(n^{1+\delta})$ demand queries per round has an input graph $G$ with
$\mu(G)=\Omega(n)$ on which, with probability at least
$1-\widetilde O(n^{-2})$ over the algorithm's randomness,
\[
    W=O\!\left(k n^{1-\beta}\sqrt{\log n}+k\log n\right).
\]
Consequently, its worst-case approximation factor is
\[
    \Omega\!\left(
      \max\!\left\{1,
        \min\!\left\{
          \frac{n^\beta}{k\sqrt{\log n}},
          \frac{n}{k\log n}
        \right\}
      \right\}
    \right).
\]
In particular, for $\delta\ge0$, this gives the approximation lower bound
\[
    \Omega\!\left(
      \max\!\left\{1,\frac{n^\beta}{k\sqrt{\log n}}\right\}
    \right).
\]
\end{theorem}

\begin{remark}
If $k n^{-\beta}\sqrt{\log n}>1/2$, then
$k n^{1-\beta}\sqrt{\log n}>n/2$.  Since every matching has size at most
$n$, the stated output-size bound follows immediately on any graph with
a perfect matching.  We therefore assume for the rest of the proof that
\begin{equation}\label{eq:mass-range}
    k n^{-\beta}\sqrt{\log n}\le\frac12.
\end{equation}
All estimates below are for sufficiently large $n$; their constants may
depend on the constant in the query bound.
\end{remark}

In Section~\ref{sec:hard-distribution}, we define the hard distribution.
Section~\ref{sec:proof-overview} describes the proof strategy, and
Sections~\ref{sec:distribution-properties} and~\ref{sec:helper-lemmas}
establish the properties needed to analyse each round.  We complete the
proof in Section~\ref{sec:proof-main}.

\subsection{The Hard Distribution}
\label{sec:hard-distribution}

\begin{definition}[$\mu_{k,\delta}$]\label{def:hard-distribution}
For parameters satisfying~\eqref{eq:mass-range}, let $\mu_{k,\delta}$ be
the following distribution on bipartite graphs with vertex set
$L\sqcup R$, where $|L|=|R|=n$.
\begin{itemize}
    \item \textbf{$R$ partitioning.}
    Independently for each $v\in R$, assign $v$ to exactly one of
    $T_0,T_1,\ldots,T_k$.  Each set $T_j$, for $j<k$, is chosen with
    probability $n^{-\beta}\sqrt{\log n}$, and $T_k$ is chosen with
    the remaining probability $1-k n^{-\beta}\sqrt{\log n}$.

    \item \textbf{Edges.}
    Conditional on this partition, include every pair $(u,v)\in L\times R$
    independently, with probability $p_j:=n^{-j/k}$ when $v\in T_j$.
\end{itemize}
The algorithm is given the labelled vertex sets $L$ and $R$, but not the
partition into the sets $T_j$.
\end{definition}

The vertices in $T_0$ have the highest degree, and the expected degree
decreases geometrically across $T_1,\ldots,T_k$.  In particular, a vertex
in $T_k$ has edge probability $1/n$, and hence expected degree one.  Each
of the other layers has expected size $n^{1-\beta}\sqrt{\log n}$.  We
will use the notation
\[
    D:=\bigcup_{j=0}^{k-1}T_j
\]
for the union of these dense layers.

\begin{observation}\label{obs:dense-size}
With probability at least $1-n^{-4}$,
\[
    |D|=O\!\left(k n^{1-\beta}\sqrt{\log n}+\log n\right).
\]
Moreover, $\Pr[|T_k|<n/3]=\exp(-\Omega(n))$.
\end{observation}

\begin{proof}
For $v\in R$, the indicators of the events $v\in D$ are independent
Bernoulli random variables with parameter $k n^{-\beta}\sqrt{\log n}$.
Therefore,
\[
    |D|\sim\Bin\!\left(n,k n^{-\beta}\sqrt{\log n}\right),
    \qquad
    \E|D|=k n^{1-\beta}\sqrt{\log n}.
\]
The first assertion follows from the additive Chernoff bound.  For the
second assertion, $|T_k|$ is binomial with mean at least $n/2$
by~\eqref{eq:mass-range}, so the lower-tail Chernoff bound applies.
\end{proof}

\begin{lemma}\label{lem:large-matching}
A graph $G$ sampled from $\mu_{k,\delta}$ satisfies
$\mu(G)\ge n/10$ with probability $1-\exp(-\Omega(n))$.
\end{lemma}

\begin{proof}
Condition on any layer assignment with $|T_k|\ge n/3$.  Process
$\lfloor n/3\rfloor$ vertices of $T_k$ in a fixed order, and greedily
match each to an unmatched neighbour in $L$ whenever possible.  As long
as the current matching has size less than $n/10$, at least $9n/10$
vertices of $L$ are unmatched.  The edges from the current vertex to
these vertices are unexposed, independent Bernoulli $1/n$ variables.
Thus the conditional probability of adding a matching edge is at least
\[
    1-(1-1/n)^{9n/10}\ge1-e^{-0.9}>\frac12.
\]
Stop when the matching first has size at least $n/10$, and replace any subsequent
success indicators by independent fair coins.  Every indicator in this
sequence has conditional success probability at least $1/2$.  Coupling
successive indicators using independent uniform random variables shows
that their sum stochastically dominates
$\Bin(\lfloor n/3\rfloor,1/2)$.  If the matching never reaches size
$n/10$, this sequence agrees with the original greedy process.  The
lower-tail Chernoff bound therefore makes this event
$\exp(-\Omega(n))$.  Finally, remove the conditioning using
Observation~\ref{obs:dense-size}.
\end{proof}

\subsection{Proof Overview}
\label{sec:proof-overview}

To prove Theorem~\ref{thm:k-round-lb}, we first fix an arbitrary
deterministic algorithm $\ALG$ and analyse its execution on
$G\sim\mu_{k,\delta}$.  Let $c_0\ge1$ be a constant such that $\ALG$
asks at most $c_0n^{1+\delta}$ demand queries per round.  Let
$Q_1,Q_2,\ldots,Q_k$ be the query sets in its successive rounds, and let
$Q_{\le r}:=\bigcup_{i=1}^r Q_i$.  All queries in a round are chosen
before any responses from that round are received.

We recall the information certified by a demand query.  A query
$q=(u_q,\sigma_q)$ specifies a left vertex and an ordering
$\sigma_q(1),\ldots,\sigma_q(n)$ of $R$.  Its response $\dem(q)$ is the
first neighbour of $u_q$ in this ordering, or $\bot$ if there is no
neighbour.  When $\dem(q)\ne\bot$, define
\[
\begin{split}
    \YES(q)&:=\{(u_q,\dem(q))\},\\
    \NO(q)&:=\{(u_q,\sigma_q(i)):
                     1\le i<\sigma_q^{-1}(\dem(q))\}.
\end{split}
\]
When $\dem(q)=\bot$, let $\YES(q)=\varnothing$ and
$\NO(q)=\{u_q\}\times R$.  For a set $Q$ of answered queries, write
\[
    \YES(Q):=\bigcup_{q\in Q}\YES(q),
    \qquad
    \NO(Q):=\bigcup_{q\in Q}\NO(q).
\]
These are the pairs certified to be edges and non-edges, respectively.
A pair in neither set is \emph{uncertified}.  We write $\YES(v,Q)$ and
$\NO(v,Q)$ for the pairs in these sets incident on $v$.

At the end of each round, we categorise the vertices of $R$ according to
the information obtained so far.  Define nested sets
$R=R_0\supseteq R_1\supseteq\cdots\supseteq R_k$ by
\begin{equation}\label{eq:residual-sets}
    R_{r+1}:=
    \left\{v\in R_r:
      |\YES(v,Q_{\le r+1})|=0
      \text{ and }
      |\NO(v,Q_{\le r+1})|\le n^{(r+1)/k}
    \right\}.
\end{equation}
Thus a vertex leaves $R_r$ \sagnik{A vertex of $R_r$ is not included in $R_{r+1}$} if an incident edge is certified, or if the
number of its incident certified non-edges exceeds the next threshold.
For $0\le r<k$, let
\begin{align*}
    \NO_r&:=\{v\in R_r:
                  |\NO(v,Q_{\le r+1})|>n^{(r+1)/k}\},\\
    \YES_r&:=\{v\in R_r:|\YES(v,Q_{r+1})|\ge1\}.
\end{align*}
Then $R_r\setminus R_{r+1}=\NO_r\cup\YES_r$.  Notice that a vertex
in $R_r$ has no previously certified incident edge, so the definition
of $\YES_r$ could equivalently use $Q_{\le r+1}$.

Every edge certified by the algorithm has its right endpoint in
$R\setminus R_k$.  Therefore, $R\setminus R_k$ is a vertex cover for
the certified edges.  To bound its size, we count the vertices in $D$
once over the entire execution, and analyse the remaining removals
round by round.  In particular,
\begin{equation}\label{eq:cover-bound}
    |R\setminus R_k|
    \le |D|+\sum_{r=0}^{k-1}|\NO_r|
              +\sum_{r=0}^{k-1}|\YES_r\cap T_k|.
\end{equation}
The first term is controlled by Observation~\ref{obs:dense-size}.
For the other two terms, we will show, with high probability, that
\begin{align*}
    \sum_{r=0}^{k-1}|\NO_r|
        &=O\!\left(k n^{1-\beta}\sqrt{\log n}\right),\\
    \sum_{r=0}^{k-1}|\YES_r\cap T_k|
        &=O\!\left(n^{1-\beta}\sqrt{\log n}+k\log n\right).
\end{align*}
We then account separately for uncertified output pairs into $T_k$.
Any output pair into $D$ is included in the same single charge for $D$.

To analyse round $r+1$, fix an arbitrary positive-probability transcript
$\tau$ of the first $r$ rounds.  This fixes the queries and responses
in these rounds, the set $R_r$, and the queries in $Q_{r+1}$, but not
their responses.  All conditional estimates below hold uniformly over
such transcripts.  We suppress the superscript $\tau$ on these fixed
objects to keep the notation simple.

\subsection{Properties of the Graph Distribution}
\label{sec:distribution-properties}

The proof uses the following independence properties of
$\mu_{k,\delta}$.

\begin{observation}\label{obs:column-independence}
The layer assignments and incident edge variables associated with
distinct vertices of $R$ are independent.  Conditional on $v\in T_i$,
the pairs $(u,v)$, $u\in L$, are independent Bernoulli $n^{-i/k}$
variables and are independent of all random data at other right
vertices.
\end{observation}

This follows directly from Definition~\ref{def:hard-distribution}.
The conditioning on $v\in T_i$ in the second assertion is necessary:
without it, finding an edge incident on $v$ can increase the probability
that $v$ belongs to a denser layer, and hence change the probability
that another incident pair is an edge.  We next show how independence
across right vertices is preserved after fixing a transcript.

\begin{definition}[$\mathcal E_\tau$, $\mathcal E_\tau(v)$]
\label{def:transcript-events}
Let $\mathcal E_\tau$ be the event that every pair in $\YES(Q_{\le r})$
is an edge and every pair in $\NO(Q_{\le r})$ is a non-edge.
Similarly, $\mathcal E_\tau(v)$ is the event that the pairs in
$\YES(v,Q_{\le r})$ are edges and the pairs in $\NO(v,Q_{\le r})$
are non-edges.
\end{definition}

Although $\mathcal E_\tau$ holds trivially because of the transcript $\tau$, we define it anyway, as it clarifies the probability space used in the proof. We use $\Pr[A]$ for probability under the initial distribution
$\mu_{k,\delta}$, and $\Pr[A\mid\mathcal E_\tau]$ for probability
after the transcript is fixed.

\begin{lemma}\label{lem:transcript-independence}
The event $\mathcal E_\tau$ is exactly the event that $\ALG$ produces
transcript $\tau$, and
\[
    \mathcal E_\tau=\bigcap_{v\in R}\mathcal E_\tau(v).
\]
Conditional on $\mathcal E_\tau$, the random data \sagnik{Maybe there is better terminology?} at distinct right
vertices remain independent.  Conditional further on $v\in T_i$,
the uncertified pairs incident on $v$ remain independent Bernoulli
$n^{-i/k}$ variables.
\end{lemma}

\begin{proof}
A demand-query response specifies exactly that the returned pair is
present and all earlier pairs in the query ordering are absent; a
response $\bot$ specifies that all pairs in that row are absent.
These are precisely the constraints in $\mathcal E_\tau$.  Because the
algorithm is deterministic, fixing earlier responses also fixes its
later queries.  Thus these constraints characterise the entire
transcript, including its adaptive choices.

Each event $\mathcal E_\tau(v)$ depends only on the layer assignment
of $v$ and its incident edge variables.  These collections of random
data are independent across $v$ by
Observation~\ref{obs:column-independence}.  Conditioning on their
product event therefore preserves independence across right vertices.
Once a layer is fixed, the incident edge bits were independent, and
the transcript only fixes some of their coordinates.  The remaining
coordinates keep their original independent distributions.
\end{proof}

\subsection{Helper Lemmas}
\label{sec:helper-lemmas}

In this section, we analyse the demand queries in round $r+1$
conditioned on a transcript $\tau$ through round $r$.

\paragraph{The probability of belonging to a layer.}
For a vertex in $R_r$, no incident edge has been certified, and at most
$n^{r/k}$ incident non-edges have been certified.  We show that this information is insignificant for altering the probability of a vertex $v \in T_i$, i.e., it has not reduced the probability of belonging to
$T_i$, for $i\ge r$, by more than a constant factor.

\begin{lemma}\label{lem:posterior-mass}
For every $v\in R_r$ and $i\ge r$,
\[
    \Pr[v\in T_i\mid\mathcal E_\tau]
       \ge\frac14\Pr[v\in T_i].
\]
\end{lemma}

\begin{proof}
For $r=0$, the transcript is empty and the assertion is immediate.
Suppose $r\ge1$.  By the definition of $R_r$, the event
$\mathcal E_\tau(v)$ prescribes only non-edges, and there are at most
$n^{r/k}$ of them.  Conditional independence within a fixed layer gives
\begin{align*}
    \Pr[\mathcal E_\tau(v)\mid v\in T_i]
      &=(1-n^{-i/k})^{|\NO(v,Q_{\le r})|}\\
      &\ge(1-n^{-i/k})^{n^{r/k}}\\
      &\ge(1-n^{-r/k})^{n^{r/k}}
       \ge\frac14.
\end{align*}
For the last inequality, use $(1-1/x)^x\ge1/4$ for $x\ge2$ and
$n^{r/k}\ge n^{1/k}\ge2$.  By
Lemma~\ref{lem:transcript-independence} and Bayes' rule,
\begin{align*}
    \Pr[v\in T_i\mid\mathcal E_\tau]
      &=\Pr[v\in T_i\mid\mathcal E_\tau(v)]\\
      &=\frac{\Pr[\mathcal E_\tau(v)\mid v\in T_i]
                  \Pr[v\in T_i]}
              {\Pr[\mathcal E_\tau(v)]}\\
      &\ge\frac14\Pr[v\in T_i],
\end{align*}
as required.
\end{proof}

\begin{lemma}\label{lem:edge-probability}
For every $r<k$ and every uncertified pair $(u,v)$ with $v\in R_r$,
\[
    \Pr[(u,v)\in E\mid\mathcal E_\tau]
        \ge\frac14 n^{-r/k-\beta}\sqrt{\log n}.
\]
\end{lemma}

\begin{proof}
It is enough to count the possibility that $v\in T_r$.  By
Lemmas~\ref{lem:transcript-independence} and~\ref{lem:posterior-mass},
\begin{align*}
    \Pr[(u,v)\in E\mid\mathcal E_\tau]
      &\ge\Pr[(u,v)\in E\mid v\in T_r,\mathcal E_\tau]
                   \Pr[v\in T_r\mid\mathcal E_\tau]\\
      &\ge n^{-r/k}\cdot\frac14\Pr[v\in T_r]\\
      &=\frac14 n^{-r/k-\beta}\sqrt{\log n}.
\end{align*}
\end{proof}

\paragraph{Bounding the number of newly certified pairs.}
Set
\begin{equation}\label{eq:prefix-length}
    t_r:=\left\lceil32 n^{r/k+\beta}\sqrt{\log n}\right\rceil.
\end{equation}
For a query $q=(u_q,\sigma_q)\in Q_{r+1}$, let $F_q$ be the first
$t_r$ pairs $(u_q,v)$ with $v\in R_r$ that are uncertified after round
$r$, in the order induced by $\sigma_q$.  If there are fewer than
$t_r$ such pairs, take all of them.  Both $t_r$ and $F_q$ are fixed by
the transcript before round $r+1$.

\begin{lemma}\label{lem:query-prefix}
For any query $q\in Q_{r+1}$,
\[
    \Pr\!\left[
      |(\NO(q)\setminus\NO(Q_{\le r}))\cap(L\times R_r)|
          \ge t_r
      \,\middle|\,\mathcal E_\tau
    \right]\le n^{-8}.
\]
Moreover, except with conditional probability at most $n^{-8}$, every
newly certified pair in $L\times R_r$ supplied by $q$ belongs to $F_q$.
\end{lemma}

\begin{proof}
Let $v_1,v_2,\ldots$ list the right endpoints of the previously
uncertified pairs into $R_r$ in the query ordering.  If the query
certifies at least $t_r$ of these pairs to be non-edges, its first
$t_r$ pairs in this list must all be absent.  The converse need not
hold, since the query may stop earlier at a vertex outside $R_r$.
If the list has fewer than $t_r$ pairs, the claimed event is impossible.
Otherwise, since the vertices $v_i$ are distinct,
Lemmas~\ref{lem:transcript-independence} and~\ref{lem:edge-probability}
give
\begin{align*}
    \Pr\!\left[
       \bigcap_{i=1}^{t_r}\{(u_q,v_i)\notin E\}
          \,\middle|\,\mathcal E_\tau
    \right]
      &=\prod_{i=1}^{t_r}
          \Pr[(u_q,v_i)\notin E\mid\mathcal E_\tau]\\
      &\le\left(1-\frac14 n^{-r/k-\beta}\sqrt{\log n}\right)^{t_r}\\
      &\le\exp(-8\log n)=n^{-8}.
\end{align*}
If any newly certified pair into $R_r$ lies beyond $F_q$, then the
query must also have passed all of the first $t_r$ pairs without
finding an edge.  The same bound therefore proves the second assertion.
\end{proof}

\begin{lemma}\label{lem:new-pairs}
Let $\New_r$ be the set of pairs in $L\times R_r$ that were
uncertified after round $r$ and are certified by the end of round
$r+1$.  Then, conditional on any transcript $\tau$ through round $r$,
\[
    |\New_r|
       =O\!\left(n^{1+\delta+r/k+\beta}\sqrt{\log n}\right)
\]
with probability at least $1-O(n^{-6})$.
\end{lemma}

\begin{proof}
Let $\mathcal B_r$ be the event that all newly certified pairs in
$L\times R_r$ supplied by each query $q\in Q_{r+1}$ lie in $F_q$.
Since $\delta\le1/k\le1$, the number of queries in the round is at
most $c_0n^2$.  A union bound using Lemma~\ref{lem:query-prefix} gives
\begin{equation}\label{eq:prefix-success}
    \Pr[\mathcal B_r^c\mid\mathcal E_\tau]\le c_0n^{-6}.
\end{equation}
Put $F_r:=\bigcup_{q\in Q_{r+1}}F_q$, which is fixed by $\tau$.
On $\mathcal B_r$,
\begin{equation}\label{eq:candidate-budget}
    |\New_r|\le |F_r|\le c_0n^{1+\delta}t_r
       =O\!\left(n^{1+\delta+r/k+\beta}\sqrt{\log n}\right).
\end{equation}
The last equality absorbs the ceiling in~\eqref{eq:prefix-length},
using $\beta\ge0$.  Notice that $F_r$ is defined and fixed by $\tau$
whether or not $\mathcal B_r$ occurs.
\end{proof}

\subsection{Proof of Theorem~\ref{thm:k-round-lb}}
\label{sec:proof-main}

We now bound the three terms in~\eqref{eq:cover-bound}, beginning with
the vertices that leave $R_r$ because of certified non-edges.

\paragraph{Counting vertices adjacent to certified non-edges.}

\begin{lemma}\label{lem:negative-removals}
For every $0\le r<k$, conditional on any transcript through round $r$,
\[
    |\NO_r|=O\!\left(n^{1-\beta}\sqrt{\log n}\right)
\]
with probability at least $1-O(n^{-6})$.  Consequently,
\[
    \sum_{r=0}^{k-1}|\NO_r|
       =O\!\left(k n^{1-\beta}\sqrt{\log n}\right)
\]
with probability at least $1-O(kn^{-6})$.
\end{lemma}

\begin{proof}
By Lemma~\ref{lem:new-pairs}, at most
$O(n^{1+\delta+r/k+\beta}\sqrt{\log n})$ previously uncertified pairs
in $L\times R_r$ are certified in round $r+1$, except with conditional
probability $O(n^{-6})$.  A vertex $v\in R_r$ already has at most
$n^{r/k}$ certified non-edges.  This also holds for $r=0$, when the
actual number is zero.  Thus, if $v\in\NO_r$, it must acquire more than
\[
    n^{(r+1)/k}-n^{r/k}
       =n^{(r+1)/k}(1-n^{-1/k})
       \ge\frac12 n^{(r+1)/k}
\]
new certified non-edges.  Each new pair has only one right endpoint,
so dividing the bound on $|\New_r|$ by this threshold gives
\[
    |\NO_r|
      =O\!\left(n^{1+\delta+\beta-1/k}\sqrt{\log n}\right)
      =O\!\left(n^{1-\beta}\sqrt{\log n}\right).
\]
Here and below, we use the identity
\begin{equation}\label{eq:balance-identity}
    1+\delta+\beta-\frac1k=1-\beta.
\end{equation}
The conditional bound is uniform over the preceding transcript.
Averaging over that transcript and taking a union bound over the
$k$ rounds proves the second assertion.
\end{proof}

\paragraph{Counting vertices adjacent to certified edges.}
We account for all vertices in $D$ using
Observation~\ref{obs:dense-size}, regardless of the rounds in which
their incident edges are discovered.  It remains to bound the vertices
in $\YES_r\cap T_k$.

\begin{lemma}\label{lem:sparse-removals}
For every $0\le r<k$, conditional on any transcript through round $r$,
\[
    |\YES_r\cap T_k|
       =O\!\left(n^{\delta+r/k+\beta}\sqrt{\log n}+\log n\right)
\]
with probability at least $1-O(n^{-4})$.  Consequently,
\[
    \sum_{r=0}^{k-1}|\YES_r\cap T_k|
       =O\!\left(n^{1-\beta}\sqrt{\log n}+k\log n\right)
\]
with probability at least $1-O(kn^{-4})$.
\end{lemma}

\begin{proof}
Fix the preceding transcript $\tau$, and recall the sets $F_q$ and
$F_r=\bigcup_{q\in Q_{r+1}}F_q$.  For $(u,v)\in F_r$, let
$\Exist(u,v)$ be the indicator of the event
\[
    \{(u,v)\in E\}\cap\{v\in T_k\},
\]
and define
\[
    Z_r:=\sum_{(u,v)\in F_r}\Exist(u,v).
\]
On the event $\mathcal B_r$ from~\eqref{eq:prefix-success}, every
edge certified in this round with right endpoint in $R_r$ belongs
to $F_r$.  Each vertex of $\YES_r\cap T_k$ is incident on at least
one such edge.  Therefore, on $\mathcal B_r$,
\[
    |\YES_r\cap T_k|\le Z_r.
\]

We next bound $Z_r$ conditional on $\mathcal E_\tau$.  The set $F_r$
is fixed by this conditioning.  Condition further on any complete
layer assignment having positive probability.  Every pair in $F_r$
is uncertified after round $r$.  By
Lemma~\ref{lem:transcript-independence}, the edge bits of these pairs
whose right endpoints lie in $T_k$ are mutually independent
Bernoulli $1/n$ variables.  For all other pairs, $\Exist(u,v)=0$.
Consequently, $Z_r$ is stochastically dominated by
$\Bin(|F_r|,1/n)$: the zero indicators can be replaced by fresh
independent Bernoulli $1/n$ variables.  This dominating distribution
is the same for every layer assignment, so the domination remains
valid after averaging over those assignments.

An additive binomial tail bound now gives
\begin{equation}\label{eq:sparse-tail}
    \Pr\!\left[
        Z_r>\frac{2|F_r|}{n}+10\log n
        \,\middle|\,\mathcal E_\tau
    \right]\le n^{-4}.
\end{equation}
For example, for $X\sim\Bin(m,1/n)$, the estimate
$\E[2^X]\le e^{m/n}$ and Markov's inequality imply this bound.
We combine~\eqref{eq:sparse-tail} with the prefix event by writing
\begin{align*}
    &\Pr\!\left[
        |\YES_r\cap T_k|>\frac{2|F_r|}{n}+10\log n
        \,\middle|\,\mathcal E_\tau
    \right]\\
    &\qquad\le
       \Pr[\mathcal B_r^c\mid\mathcal E_\tau]
       +\Pr\!\left[
          Z_r>\frac{2|F_r|}{n}+10\log n
          \,\middle|\,\mathcal E_\tau
       \right]
       \le c_0n^{-6}+n^{-4}.
\end{align*}
In particular, the estimate for $Z_r$ does not condition on
$\mathcal B_r$.  By~\eqref{eq:candidate-budget},
\[
    \frac{|F_r|}{n}
       =O\!\left(n^{\delta+r/k+\beta}\sqrt{\log n}\right),
\]
which proves the bound for one round.

For the sum over rounds, use $n^{1/k}\ge2$ to obtain
\[
    \sum_{r=0}^{k-1}n^{r/k}
       =\frac{n-1}{n^{1/k}-1}
       \le 2n^{1-1/k}.
\]
The bounds above are uniform over each preceding transcript, so
averaging and a union bound over the rounds give
\begin{align*}
    \sum_{r=0}^{k-1}|\YES_r\cap T_k|
      &=O\!\left(
           n^{\delta+\beta}\sqrt{\log n}
              \sum_{r=0}^{k-1}n^{r/k}+k\log n
         \right)\\
      &=O\!\left(
           n^{1+\delta+\beta-1/k}\sqrt{\log n}+k\log n
         \right)\\
      &=O\!\left(n^{1-\beta}\sqrt{\log n}+k\log n\right)
\end{align*}
with probability at least $1-O(kn^{-4})$.
\end{proof}

\paragraph{Accounting for uncertified output edges.}
We have bounded the right endpoints needed for certified output edges.
The algorithm may also output pairs that were not certified by a
query.  We now bound the number of these pairs that are true edges
into $T_k$.

\begin{lemma}\label{lem:guessed-edges}
The number of true uncertified output edges with right endpoint in
$T_k$ is $O(\log n)$ with probability at least $1-n^{-4}$, whenever
the output pairs have distinct endpoints.
\end{lemma}

\begin{proof}
Fix the final transcript.  The deterministic algorithm's output is now
fixed.  If its pairs do not have distinct endpoints, then $W=0$, so
there is no valid output matching to bound.  Otherwise, there are at
most $n$ output pairs, with distinct right endpoints.

For an uncertified output pair $(u,v)$ with
$\Pr[v\in T_k\mid\mathcal E_\tau]>0$,
Lemma~\ref{lem:transcript-independence} gives
\begin{align*}
    \Pr[v\in T_k,\ (u,v)\in E\mid\mathcal E_\tau]
      &=\Pr[v\in T_k\mid\mathcal E_\tau]
          \Pr[(u,v)\in E\mid v\in T_k,\mathcal E_\tau]\\
      &=\frac1n\Pr[v\in T_k\mid\mathcal E_\tau]
       \le\frac1n.
\end{align*}
If $\Pr[v\in T_k\mid\mathcal E_\tau]=0$, the same bound holds
because the joint event is impossible.  The indicators of these joint
events are independent across distinct right endpoints by
Lemma~\ref{lem:transcript-independence}.  Their sum is therefore
stochastically dominated by $\Bin(n,1/n)$, and is $O(\log n)$ except
with conditional probability at most $n^{-4}$.  Averaging over the
final transcript proves the assertion.
\end{proof}

\paragraph{Completing the proof.}
First consider a deterministic algorithm, and suppose its output $M$
is a valid matching.  Its edges with right endpoint in $D$ contribute
at most $|D|$, whether or not those edges were certified.  Consider a
certified output edge whose right endpoint $v$ belongs to $T_k$.
At the round when $v$ first left the nested sets $R_r$, it was counted
either by $\NO_r$ or by $\YES_r\cap T_k$.  Thus all such certified
output edges are covered by the two sums below.  Writing $U$ for the
number of true uncertified output edges into $T_k$, we obtain
\begin{equation}\label{eq:final-charge}
    W\le |D|+\sum_{r=0}^{k-1}|\NO_r|
                  +\sum_{r=0}^{k-1}|\YES_r\cap T_k|+U.
\end{equation}
An output pair certified to be a non-edge cannot occur in a valid
matching.  If the output is invalid, $W=0$ and the desired upper bound
holds automatically.

Apply Observation~\ref{obs:dense-size} as a single event under the
original graph distribution, and combine it with
Lemmas~\ref{lem:negative-removals}, \ref{lem:sparse-removals}, and
\ref{lem:guessed-edges}.  No transcript estimate was conditioned on the
size of $D$.  A union bound and~\eqref{eq:final-charge} give
\begin{equation}\label{eq:distributional-conclusion}
    W=O\!\left(k n^{1-\beta}\sqrt{\log n}+k\log n\right)
\end{equation}
with probability at least $1-O(kn^{-4})$ over
$G\sim\mu_{k,\delta}$.  By Lemma~\ref{lem:large-matching},
$\mu(G)\ge n/10$ except with probability $\exp(-\Omega(n))$.

Finally, let the algorithm be randomised.  Fixing its entire random
tape gives a deterministic algorithm to which the same distributional
bound applies.  The distribution $\mu_{k,\delta}$ and the constants
in~\eqref{eq:distributional-conclusion} do not depend on the tape.
Let
\[
    B:=C\!\left(k n^{1-\beta}\sqrt{\log n}+k\log n\right)
\]
for a sufficiently large constant $C$.  Averaging over the random tape
$\omega$ yields
\[
    \Pr_{G,\omega}[W>B]=O(kn^{-4}).
\]
Let $\mathcal H$ be the event $\mu(G)\ge n/10$.  Averaging now over
graphs in this event gives
\[
    \E_{G\mid\mathcal H}\!
       \left[\Pr_\omega[W>B\mid G]\right]
    \le\frac{O(kn^{-4})}{\Pr_G[\mathcal H]}
    =O(kn^{-4}).
\]
Therefore there is a fixed graph $G$ with $\mu(G)\ge n/10$ for which
$\Pr_\omega[W>B\mid G]=O(kn^{-4})$.  This is in particular
$\widetilde O(n^{-2})$, as asserted in
Theorem~\ref{thm:k-round-lb}.

On this graph, dividing the linear optimum by the output-size bound
gives
\[
    \frac{n}{k n^{1-\beta}\sqrt{\log n}+k\log n}
       =\Theta\!\left(
          \min\!\left\{
             \frac{n^\beta}{k\sqrt{\log n}},
             \frac{n}{k\log n}
          \right\}
        \right).
\]
Combining this with the universal lower bound of one on an
approximation factor proves the stated general bound.  For
$\delta\ge0$, we have $\beta\le1/(2k)\le1/2$, so
$n^{1-\beta}\sqrt{\log n}$ dominates $\log n$.  This gives the final
assertion of the theorem.\hfill$\square$

\subsection{An Expected-Size Lower Bound}
\label{sec:expected-size-lower-bound}

In this section, we prove a corresponding lower bound on the expected
size of the output matching. We use the same layered construction,
with a different choice of layer probabilities. We then bound the
contribution of each query directly in expectation.

\begin{theorem}[Expected-size and constant-error lower bounds]
\label{thm:random-lb-expectation}
Let $1\le k\le\lfloor\log_2 n\rfloor$, $-1\le\delta\le1/k$, and
$\beta:=1/(2k)-\delta/2$. Any randomised algorithm that uses $k$ rounds
and $O(n^{1+\delta})$ demand queries per round has an input graph $G$
with $\mu(G)\ge n/10$ on which
\begin{equation}\label{eq:random-lb-expectation}
    \E W=O\!\left(k n^{1-\beta}+1\right),
\end{equation}
where the expectation is over the algorithm's randomness.
Consequently, its worst-case expected-size approximation factor is
\[
    \Omega\!\left(
        \max\left\{1,\frac{n^\beta}{k}\right\}
    \right).
\]
The same asymptotic lower bound holds for algorithms required to
return a valid approximation with probability at least $2/3$.
\end{theorem}

The remaining budget range $\delta>1/k$ gives only the trivial
approximation lower bound of one: in this range the displayed
output-size bound already exceeds $n$.

We first consider the case $k n^{-\beta}\le1/2$. Let
$\mu^{\mathrm{exp}}_{k,\delta}$ denote the distribution in
Definition~\ref{def:hard-distribution}, with the layer probabilities
replaced by
\[
    \Pr[v\in T_j]=n^{-\beta}\quad (0\le j<k),
    \qquad
    \Pr[v\in T_k]=1-k n^{-\beta}.
\]
The edge probabilities remain $p_j=n^{-j/k}$. The layer assignments
are independent across right vertices, as before.

Fix a deterministic algorithm $\ALG$. We use the same sets
$R_r$, $\NO_r$, $\YES_r$ and
$D=\bigcup_{j<k}T_j$ as in the preceding proof. The proofs of
Lemmas~\ref{lem:transcript-independence}
and~\ref{lem:posterior-mass} apply unchanged to this distribution.
In particular, after fixing a transcript $\tau$ through round $r$,
every uncertified pair $(u,v)$ with $v\in R_r$ satisfies
\begin{equation}\label{eq:expected-edge-probability}
    \Pr[(u,v)\in E\mid\mathcal E_\tau]
        \ge \frac14 n^{-r/k-\beta}.
\end{equation}

We start by bounding the contribution of one demand query.

\begin{lemma}\label{lem:query-expectation}
Fix a transcript $\tau$ through round $r$, and let
$q=(u_q,\sigma_q)\in Q_{r+1}$. Let $H_q$ be the number of previously
uncertified pairs in $L\times R_r$ that $q$ certifies to be non-edges.
Under $\mu^{\mathrm{exp}}_{k,\delta}$,
\[
    \E[H_q\mid\mathcal E_\tau]\le4n^{r/k+\beta},
\]
and
\[
    \Pr[\dem(q)\in R_r\cap T_k\mid\mathcal E_\tau]
        \le4n^{r/k+\beta-1}.
\]
\end{lemma}

\begin{proof}
List the previously uncertified pairs into $R_r$ in the order
induced by $\sigma_q$ as
\[
    (u_q,v_1),(u_q,v_2),\ldots.
\]
This list is fixed by the preceding transcript, and its right
endpoints are distinct.

If $H_q\ge t$, then the first $t$ pairs in this list must all be
absent. Using independence across right vertices
and~\eqref{eq:expected-edge-probability}, we obtain
\[
    \Pr[H_q\ge t\mid\mathcal E_\tau]
        \le
        \left(1-\frac14 n^{-r/k-\beta}\right)^t.
\]
If the list contains fewer than $t$ pairs, the event is impossible.
Summing the tail gives
\[
    \E[H_q\mid\mathcal E_\tau]
        =\sum_{t\ge1}\Pr[H_q\ge t\mid\mathcal E_\tau]
        \le4n^{r/k+\beta}.
\]

For the second assertion, suppose the query returns $v_i\in T_k$.
Then $(u_q,v_i)$ is an edge, and every preceding pair in the list
is a non-edge. Since $(u_q,v_i)$ is uncertified after round $r$,
\[
    \Pr[v_i\in T_k,\ (u_q,v_i)\in E\mid\mathcal E_\tau]
        =\frac1n\Pr[v_i\in T_k\mid\mathcal E_\tau]
        \le\frac1n.
\]
This joint event is independent of the absence of the preceding
pairs, because they have different right endpoints. Therefore,
\begin{align*}
    \Pr[\dem(q)\in R_r\cap T_k\mid\mathcal E_\tau]
      &\le
        \frac1n\sum_{i\ge1}
        \left(1-\frac14 n^{-r/k-\beta}\right)^{i-1}\\
      &=4n^{r/k+\beta-1}.
\end{align*}
The query may stop earlier at a vertex outside $R_r$, which only
decreases the probability being bounded.
\end{proof}

\begin{proof}[Proof of Theorem~\ref{thm:random-lb-expectation}]
If $k n^{-\beta}>1/2$, then $k n^{1-\beta}>n/2$, so the claimed
output-size bound follows immediately on any graph with a perfect
matching. We therefore assume $k n^{-\beta}\le1/2$ and use
$\mu^{\mathrm{exp}}_{k,\delta}$.

First fix a deterministic algorithm $\ALG$, and let $c_0\ge1$ be
a constant such that it asks at most $c_0n^{1+\delta}$ demand queries
per round. We bound its expected output using the same accounting
as in~\eqref{eq:final-charge}.

\paragraph{Counting vertices adjacent to certified non-edges.}
Fix a transcript $\tau$ through round $r$. As in the proof of
Lemma~\ref{lem:negative-removals}, a vertex in $\NO_r$ must acquire
more than
\[
    n^{(r+1)/k}-n^{r/k}
        \ge\frac12 n^{(r+1)/k}
\]
new certified non-edges. The total number of these new non-edges
is at most $\sum_{q\in Q_{r+1}}H_q$. Thus
Lemma~\ref{lem:query-expectation} gives
\begin{align*}
    \E[|\NO_r|\mid\mathcal E_\tau]
      &\le
        \frac{2}{n^{(r+1)/k}}
        \sum_{q\in Q_{r+1}}\E[H_q\mid\mathcal E_\tau]\\
      &=O\!\left(n^{1+\delta+\beta-1/k}\right)\\
      &=O\!\left(n^{1-\beta}\right).
\end{align*}
The last equality uses
$1+\delta+\beta-1/k=1-\beta$.
Averaging over the transcript and summing over the rounds yields
\[
    \E\!\left[\sum_{r=0}^{k-1}|\NO_r|\right]
        =O\!\left(k n^{1-\beta}\right).
\]

\paragraph{Counting vertices adjacent to certified edges.}
Every vertex in $\YES_r\cap T_k$ is returned by at least one query
in round $r+1$. Applying the second assertion of
Lemma~\ref{lem:query-expectation} separately to each query gives
\begin{align*}
    \E[|\YES_r\cap T_k|\mid\mathcal E_\tau]
      &\le
        \sum_{q\in Q_{r+1}}
        \Pr[\dem(q)\in R_r\cap T_k\mid\mathcal E_\tau]\\
      &=O\!\left(n^{\delta+r/k+\beta}\right).
\end{align*}
Neither this calculation nor the preceding calculation requires
independence between different queries.

Since
\[
    \sum_{r=0}^{k-1}n^{r/k}
        =\frac{n-1}{n^{1/k}-1}
        \le2n^{1-1/k},
\]
we obtain
\begin{align*}
    \E\!\left[\sum_{r=0}^{k-1}|\YES_r\cap T_k|\right]
      &=O\!\left(
          n^{\delta+\beta}\sum_{r=0}^{k-1}n^{r/k}
        \right)\\
      &=O\!\left(n^{1+\delta+\beta-1/k}\right)\\
      &=O\!\left(n^{1-\beta}\right).
\end{align*}
The dense vertices are again counted once over the entire execution.
Their expected number is
\[
    \E|D|=k n^{1-\beta}.
\]

\paragraph{Accounting for uncertified output edges.}
Fix the final transcript. As in the proof of
Lemma~\ref{lem:guessed-edges}, every uncertified output pair $(u,v)$
has probability at most $1/n$ of being a true edge with $v\in T_k$.
An output whose pairs have distinct endpoints contains at most $n$
such pairs, so their expected contribution is at most one.
If the output pairs share an endpoint, then $W=0$.
All output edges into $D$ are already covered by the charge to $|D|$.

Combining these estimates with~\eqref{eq:final-charge} gives
\begin{equation}\label{eq:expected-charge-general}
    \E_{G\sim\mu^{\mathrm{exp}}_{k,\delta}}W
        =O\!\left(k n^{1-\beta}+1\right).
\end{equation}

The proof of Lemma~\ref{lem:large-matching} also applies to
$\mu^{\mathrm{exp}}_{k,\delta}$: the sparse layer has probability
at least $1/2$, and its incident edges have probability $1/n$.
Consequently,
\[
    \Pr[\mu(G)\ge n/10]=1-\exp(-\Omega(n)).
\]

Now let the algorithm be randomised, with random tape $\omega$.
The bound in~\eqref{eq:expected-charge-general} holds uniformly
for every fixed tape, so it also holds after averaging over $\omega$.
Writing $\mathcal H$ for the event $\mu(G)\ge n/10$, we have
\[
    \E_{G\mid\mathcal H}
       \!\left[\E_\omega[W\mid G]\right]
    \le
    \frac{\E_{G,\omega}W}{\Pr[\mathcal H]}
    =O\!\left(k n^{1-\beta}+1\right).
\]
There is therefore a fixed graph $G$ with $\mu(G)\ge n/10$ on which
the expected output satisfies~\eqref{eq:random-lb-expectation}.

Since $\beta\le1$, we have $k n^{1-\beta}\ge1$. Dividing the linear
optimum by the expected output bound, and using the lower bound
of one on an approximation factor, gives
\[
    \Omega\!\left(
        \max\left\{1,\frac{n^\beta}{k}\right\}
    \right).
\]
Finally, if an algorithm returns a valid $\alpha$-approximation
with probability at least $2/3$ on every input, then on this graph
\[
    \E W\ge\frac23\frac{\mu(G)}{\alpha}.
\]
Combining this with~\eqref{eq:random-lb-expectation} gives the same
lower bound on $\alpha$, up to a constant factor.
\end{proof}

We next record the consequences for one round and for constant
approximation. For the latter, we use the following observation
about the failure probability in the preceding high-probability
proof.

\begin{remark}
For every fixed $c>0$, Theorem~\ref{thm:k-round-lb} also holds with
failure probability at most $n^{-c}$, allowing the constant in its
output-size bound to depend on $c$.

To see this, replace the candidate-prefix length by
\[
    t_r=
    \left\lceil
        4(c+6)n^{r/k+\beta}\sqrt{\log n}
    \right\rceil.
\]
The probability that a query passes this prefix without finding
an edge is then at most $n^{-c-6}$. Taking the constants in the
additive binomial bounds sufficiently large makes their failure
probabilities at most $n^{-c-6}$ as well. The same union bounds
over at most $c_0kn^2$ queries and over the rounds, followed by
averaging over graphs with linear optimum, give failure probability
at most $n^{-c}$ for sufficiently large $n$. Only the constants
in the output-size bound change.
\end{remark}

\begin{corollary}[One round and constant approximation]
\label{cor:random-lb-consequences}
For $k=1$, any randomised algorithm using $O(n^{1+\delta})$ demand
queries has expected-size and constant-error approximation lower
bound
\[
    \Omega\!\left(n^{(1-\delta)/2}\right).
\]
Moreover, an algorithm using $O(n)$ demand queries per round that
returns a constant approximation with probability at least $n^{-d}$
on every input, for any fixed $d>0$, requires
\[
    \Omega\!\left(\frac{\log n}{\log\log n}\right)
\]
rounds.
\end{corollary}

\begin{proof}
For $k=1$, we have $\beta=(1-\delta)/2$, so the first assertion
follows directly from Theorem~\ref{thm:random-lb-expectation}.

For the second assertion, it is enough to rule out
\[
    k\le\frac14\frac{\log n}{\log\log n}.
\]
In this range,
\[
    n^{1/(2k)}\ge(\log n)^2.
\]
Apply the strengthened high-probability bound above with
$\delta=0$ and failure exponent $c=d+1$. Its output-size threshold,
divided by $n$, is
\[
    O_d\!\left(
        \frac{k\sqrt{\log n}}{n^{1/(2k)}}
        +\frac{k\log n}{n}
    \right)
    =o(1).
\]
The selected input has maximum matching size at least $n/10$.
For any fixed approximation factor, a valid approximation on this
input must therefore exceed the output-size threshold once $n$
is sufficiently large. Its probability is at most
$n^{-d-1}<n^{-d}$, a contradiction.
\end{proof}

\subsection{The Two-Round Hard Distribution}
\label{sec:two-round-layer-recovery}

We now consider the expected-size hard distribution when $k=2$
and $\delta=0$. We show that two rounds suffice to find a matching
of size $\Omega(n^{3/4}/\log n)$ using at most $n$ demand queries
per round. Thus the exponent of the lower bound obtained from
this distribution is tight, up to logarithmic factors.

The same algorithm works when the probabilities of the two dense
layers are multiplied by any fixed positive constant.

\begin{proposition}\label{prop:two-round-layer-recovery}
Fix a constant $a_0>0$. Independently assign every vertex of $R$
to the layers $T_0,T_1,T_2$ with probabilities
\[
    \Pr[v\in T_0]=\Pr[v\in T_1]=a_0n^{-1/4},
    \qquad
    \Pr[v\in T_2]=1-2a_0n^{-1/4}.
\]
Conditional on the layers, include each pair $(u,v)\in L\times R$
independently, with probability $1$, $n^{-1/2}$, or $1/n$
according as $v$ belongs to $T_0$, $T_1$, or $T_2$.

For sufficiently large $n$, a two-round randomised algorithm
using at most $n$ demand queries per round finds a matching
of size
\[
    \Omega\!\left(\frac{n^{3/4}}{\log n}\right)
\]
with high probability over the graph and the algorithm's
randomness. Consequently, this distribution admits an
$O(n^{1/4}\log n)$ approximation in two rounds.
\end{proposition}

\begin{proof}
The first round will identify a large set of vertices contained
in $T_2$. The second round will find a matching into this set.
To keep the edges used in the second round independent of the
first-round information, partition $L$ into two fixed sets
$L_A,L_B$, with
\[
    |L_A|=\lfloor n/2\rfloor,
    \qquad
    |L_B|=\lceil n/2\rceil.
\]
The first round uses only vertices in $L_A$, and the second round
uses only vertices in $L_B$.

\paragraph{The first round.}
Let
\[
    t:=\lceil C\sqrt n\log n\rceil,
    \qquad
    g:=\left\lfloor\frac{|L_A|}{t}\right\rfloor,
\]
where $C>0$ is a sufficiently large constant. Partition $gt$
vertices of $L_A$ into $g$ groups of size $t$, leaving the
remaining vertices unused.

For each group $U$, independently choose a uniformly random set
$S_U\subseteq R$ of size
\[
    \ell:=\lfloor c_1n^{1/4}\rfloor,
\]
where $c_1>0$ is a fixed constant. Every vertex $u\in U$ asks
a demand query that places $S_U$ first in its ordering.

We accept $S_U$ if every response from a vertex in $U$ lies
outside $S_U$ or is $\bot$. This happens exactly when
\[
    E(U,S_U)=\varnothing.
\]
Let $S$ be the union of all accepted sets. The first round uses
$gt\le|L_A|\le n/2$ queries.

\paragraph{Properties of the graph distribution.}
For $U\subseteq L$, write $N(U)$ for the set of its neighbours
in $R$. We first establish three properties of the layers and
the edges incident on $L_A$.

First, every vertex of $T_1$ has a neighbour in every group
with high probability. For a fixed $v\in T_1$ and a fixed group
$U$, the probability of having no such neighbour is
\[
    (1-n^{-1/2})^t
        \le \exp(-t/\sqrt n)
        \le n^{-C}.
\]
A union bound over the vertices and groups proves the assertion.
Every vertex of $T_0$ has a neighbour in every group because
all its incident pairs are edges. Therefore, on this event,
every accepted set is contained in $T_2$.

Second, with high probability,
\[
    |N(U)|=O(n^{3/4})
\]
for every group $U$. Indeed, a binomial tail bound gives
$|T_0\cup T_1|=O(n^{3/4})$. Conditional on the layers, the number
of edges from a group $U$ to $T_2$ is binomial with expectation
at most $t$. Another binomial tail bound, followed by a union
bound over the groups, makes this number $O(t)$ simultaneously
for all groups. Since
\[
    t=O(\sqrt n\log n)=o(n^{3/4}),
\]
the claimed bound on $|N(U)|$ follows.

Third, define
\[
    Z:=\{v\in T_2:N(v)\cap L_A=\varnothing\}.
\]
Then $|Z|=\Omega(n)$ with high probability. To see this,
condition on the layers. We have $|T_2|=\Omega(n)$ with high
probability, and each vertex of $T_2$ belongs to $Z$ independently
with probability
\[
    (1-1/n)^{|L_A|}=\Omega(1).
\]
The assertion follows from the lower-tail Chernoff bound.

These three properties hold simultaneously with high probability.
Their failure probability can be made smaller than any fixed
inverse polynomial by taking $C$ sufficiently large.

\paragraph{Bounding the size of the accepted set.}
Condition on a layer assignment and on all edges incident on
$L_A$ for which the three properties above hold. The random
choices of the sets $S_U$ are still independent across groups.

Fix $z\in Z$ and a group $U$. Since $z\notin N(U)$, the probability
that $S_U$ contains $z$ and is accepted is exactly
\[
    \Pr[z\in S_U,\ S_U\text{ is accepted}]
      =
      \frac{\ell}{n}
      \frac{\binom{n-1-|N(U)|}{\ell-1}}
           {\binom{n-1}{\ell-1}}.
\]
The ratio of binomial coefficients can be written as
\[
    \prod_{i=0}^{\ell-2}
       \left(1-\frac{|N(U)|}{n-1-i}\right).
\]
Since $|N(U)|=O(n^{3/4})$ and $\ell=\Theta(n^{1/4})$, this product
is bounded below by a positive constant. More explicitly, for
sufficiently large $n$, its logarithm is at least
\[
    -O\!\left(\frac{|N(U)|\ell}{n}\right)=-O(1).
\]
Thus there is a constant $c_2>0$ such that
\[
    \Pr[z\in S_U,\ S_U\text{ is accepted}]
        \ge c_2\frac{\ell}{n}
\]
for every $z\in Z$ and every group $U$.

By independence across groups,
\[
    \Pr[z\in S]
        \ge1-\left(1-c_2\frac{\ell}{n}\right)^g.
\]
As $g\ell=o(n)$ and $|Z|=\Omega(n)$, summing over $z\in Z$ gives
\[
    \E|S|
      \ge
      |Z|\left(
          1-\left(1-c_2\frac{\ell}{n}\right)^g
      \right)
      =\Omega(g\ell).
\]

To show concentration, regard $|S|$ as a function of the $g$
independent choices $S_U$. Changing one choice can remove at
most $\ell$ vertices from the union and add at most $\ell$
vertices. Thus its effect on $|S|$ is at most $2\ell$.
The bounded-differences inequality gives
\[
    |S|=\Omega(g\ell)
        =\Omega\!\left(\frac{n^{3/4}}{\log n}\right)
\]
except with conditional probability $\exp(-\Omega(g))$.
Also, deterministically,
\[
    |S|\le g\ell=o(n).
\]

\paragraph{The second round.}
Each vertex $u\in L_B$ asks a demand query that places an
independently uniform random ordering of $S$ first.
We retain the responses that lie in $S$. For each distinct
returned right vertex, we keep one of its returned edges.

We show that this produces $\Omega(|S|)$ matching edges.
Condition further on all the first-round random choices, so
that $S$ is fixed. On the event considered above, $S\subseteq T_2$.
The pairs in $L_B\times S$ are still independent Bernoulli $1/n$
variables: all the graph information conditioned on concerns
the layers and the edges incident on $L_A$.

Write $m:=|S|$. For each $u\in L_B$, the probability that its
response belongs to $S$ is
\[
    1-(1-1/n)^m=\Theta(m/n).
\]
These events are independent across the vertices of $L_B$.
Since $|L_B|\ge n/2$ and $m=o(n)$, the expected number of such
responses is at least $m/4$ for sufficiently large $n$.
A Chernoff bound therefore gives at least $m/8$ responses in
$S$, except with probability $\exp(-\Omega(m))$.

Conditional on which rows return a vertex in $S$, their returned
vertices are independent and uniform in $S$. This follows from
symmetry and from the independent random ordering used in each
row. Consider the first $\lfloor m/8\rfloor$ such responses.
The expected number of distinct right vertices among them is
\[
    m\left(
        1-\left(1-\frac1m\right)^{\lfloor m/8\rfloor}
    \right)
    =\Omega(m).
\]
Changing one response changes this number by at most one.
Another application of bounded differences shows that the
number of distinct returned vertices is $\Omega(m)$ except
with probability $\exp(-\Omega(m))$.

The retained edges form a matching: only one edge is kept for
each right endpoint, and each left vertex supplied at most
one response. Its size is therefore
\[
    \Omega(m)
      =\Omega\!\left(\frac{n^{3/4}}{\log n}\right).
\]
The second round uses $|L_B|\le n$ queries. Combining the
failure probabilities proves the claimed success probability
over the graph and the algorithm's randomness.

Finally, $\mu(G)\le n$, so a matching of this size gives an
$O(n^{1/4}\log n)$ approximation.
\end{proof}

\paragraph{Scope of the exponent gap.}
For $O(n)$ demand queries per round,
Theorem~\ref{thm:random-lb-expectation} gives approximation
lower bound
\[
    \Omega\!\left(
        \max\left\{1,\frac{n^{1/(2k)}}{k}\right\}
    \right).
\]
The general upper bound in Theorem~\ref{ub:thm:main}
has factor
\[
    O\!\left(
        \frac{n^{1/(k+1)}}{k^{k/(k+1)}}
    \right)
\]
up to logarithmic factors. For fixed $k>1$, the exponents
$1/(2k)$ and $1/(k+1)$ remain different.

Proposition~\ref{prop:two-round-layer-recovery} shows that the
particular two-round distribution used here admits an
$O(n^{1/4}\log n)$ approximation. Its guarantee is over a graph
sampled from that distribution and the algorithm's randomness.
It does not give the same guarantee on arbitrary graphs. These results leave the approximation
gap for arbitrary demand-query algorithms open.


%

\section[A Randomised k-round Upper Bound]{A Randomised $k$-round Upper Bound}
\label{ub:sec:upper-bound}

In this section, we prove an upper bound for finding an approximate
maximum matching using $k$ rounds of demand queries. We use the
independent-neighbour greedy algorithm of \cite{DobzinskiNO19}
with a modification when fewer than $n$ queries are available per round.
The analysis keeps the contribution of every round. This gives the
following tradeoff, including the dependence on $k$.

Throughout this section, $G=(L\sqcup R,E)$ is a fixed bipartite graph
with $|L|=|R|=n$, and $\mu:=\mu(G)$. All probabilities and expectations
are over the algorithm's randomness. An expected-size approximation
factor $\alpha$ means that the expected size of the returned matching
is at least $\mu/\alpha$.

\begin{theorem}[$k$-round upper bound]
\label{ub:thm:main}
Let $k\ge1$ and $1\le q\le n^2$ be integers. There is a randomised
algorithm using at most $k$ rounds and at most $q$ demand queries per
round whose output $M$ is always a matching. For $\mu>0$, it satisfies
\begin{equation}\label{ub:eq:expected-size}
    \mathbb E|M|
    \ge c\min\left\{
        \mu,\ \frac{kq\mu}{n},\
        \left(\frac{kq\mu}{n}\right)^{k/(k+1)}
    \right\},
\end{equation}
where $c>0$ is an absolute constant. In particular, its worst-case
expected-size approximation factor is
\begin{equation}\label{ub:eq:approximation}
    O\!\left(
        \max\left\{
            1,\frac{n^{1/(k+1)}}{(kq/n)^{k/(k+1)}}
        \right\}
    \right).
\end{equation}
When $q\ge n$, the same asymptotic approximation factor is achieved
with probability at least $2/3$, with the same round and query budgets.
\end{theorem}

For $q=\Theta(n^{1+\delta})$, the approximation factor
in~\eqref{ub:eq:approximation} becomes
\begin{equation}\label{ub:eq:delta-tradeoff}
    O\!\left(
        \max\left\{
            1,\frac{n^{(1-k\delta)/(k+1)}}{k^{k/(k+1)}}
        \right\}
    \right).
\end{equation}
The maximum with one records that an approximation factor cannot be
smaller than one. In particular, taking $q=n$ gives factor
$O(\max\{1,n^{1/(k+1)}/k^{k/(k+1)}\})$.
The probability guarantee in the theorem applies whenever $q\ge n$;
for smaller budgets, the theorem asserts the expected-size guarantee.

In Section~\ref{ub:sec:algorithm}, we describe the algorithm and check
its query budget. We give the proof overview in
Section~\ref{ub:sec:overview}, followed by the helper lemmas and the
proof of Theorem~\ref{ub:thm:main}. In
Section~\ref{ub:sec:constant-approximation}, we show that the same
algorithm obtains a $(2+\varepsilon)$-approximation with high probability
in $O(\log n/\log\log n)$ rounds.

\subsection{The Algorithm}
\label{ub:sec:algorithm}

We first recall how to obtain a uniformly random neighbour of a left
vertex $u$ inside a set $S\subseteq R$. Ask a demand query that places
an independently uniform ordering of $S$ first. If the response lies
in $S$, it is exactly uniform in $N(u)\cap S$. Otherwise, including
when the response is $\bot$, the query certifies that
$N(u)\cap S=\varnothing$, and we discard the response.

The algorithm maintains a matching $M$. In this section, $L_r$ and
$R_r$ denote the unmatched left and right vertices at the beginning
of round $r$, respectively.

\paragraph{Independent-neighbour greedy.}
\begin{enumerate}
    \item Initially, set $M:=\varnothing$, $L_1:=L$, and $R_1:=R$.
    If $q\ge n$, set $t:=\lfloor q/n\rfloor$ and fix an ordering of
    $L$. If $q<n$, fix a balanced partition of $L$ into $q$ nonempty
    groups and fix an ordering of these groups.

    \item In each round $r\in\{1,\ldots,k\}$, issue the following
    queries before processing any of their answers.

    \emph{When $q\ge n$:} every $u\in L_r$ asks $t$ independent
    random-neighbour queries into $R_r$. These responses form an
    ordered list of samples for $u$.

    \emph{When $q<n$:} independently select one vertex uniformly
    from each group. If the selected vertex belongs to $L_r$, ask
    one random-neighbour query for it into $R_r$. A group whose
    selected vertex is already matched asks no query.

    \item Process the answers greedily.

    \emph{When $q\ge n$:} process $L_r$ in the fixed left-vertex
    order. Match $u$ to the first still-unmatched neighbour in its
    sampled list, if such a neighbour exists.

    \emph{When $q<n$:} process the groups in their fixed order.
    Add the selected vertex's returned edge if both endpoints
    are still unmatched.

    Add these edges to $M$, and let $L_{r+1}$ and $R_{r+1}$ be the
    vertices that remain unmatched.

    \item After round $k$, return $M$.
\end{enumerate}

Every sampled edge is a true edge, and an edge is added only when both
endpoints are unmatched. Thus the output is always a matching.
All selections and query orderings in a round are fixed before any
answer from that round is used. The greedy processing does not make
additional queries and never examines samples belonging to an
unprocessed vertex or group.

When $q\ge n$, each round uses at most
$nt=n\lfloor q/n\rfloor\le q$ queries. When $q<n$, each of the $q$
groups asks at most one query. Therefore the query budget holds in
every execution, rather than only in expectation.

\subsection{Proof Overview}
\label{ub:sec:overview}

Fix a maximum matching $M^\star$ of $G$, and write $X$ for the size of
the algorithm's final matching. At most $2X$ edges of $M^\star$ have an
endpoint matched by the algorithm. Consider any of the remaining
edges, say $uv$. Both endpoints of this edge remain unmatched
throughout the execution.

For each round $r$, consider the fraction of $u$'s round-start
neighbours that are still available when $u$, or its group in the
low-budget case, is processed. If this fraction is small, then many
of its neighbours must have been matched in that round. However, the algorithm matches only $X$ right vertices
in total, while $v$ remains an available neighbour of $u$ at the end.
These two facts imply that the product of the $k$ availability
fractions is at least $1/(X+1)$. Their sum is consequently at least
$k/(X+1)^{1/k}$.

We relate these fractions to the conditional probability of adding
an edge at each processing step. Let $A$ be the sum of these conditional
probabilities over the whole execution. We will show that surviving
edges of $M^\star$ force $A$ to be large. On the other hand,
$\mathbb E A=\mathbb E X$, and an exponential bound shows that $A$
is unlikely to be much larger than $X$. These statements give both
the expected-size and the constant-error guarantees.

The inequality obtained from the surviving edges holds for every
completed execution. We do not condition the distribution of a query
on the event that its left vertex remains unmatched at termination.

\subsection{Helper Lemmas}
\label{ub:sec:helpers}

For the analysis, put $s:=q/n$. Expose the random choices in each
round in the order in which its vertices or groups are processed.
This is only an ordering of the analysis: all queries in the round
were issued together. At a processing step, the unexposed choices are
still independent of the information exposed so far. In the
low-budget case, expose a group's selected vertex and its sample
together, when that group is processed.

Index the processing steps over all rounds by $j$. Let
$\mathcal F_{j-1}$ be the information exposed before step $j$, and
let $Y_j\in\{0,1\}$ indicate whether this step adds an edge to $M$.
Set
\begin{equation}\label{ub:eq:conditional-probabilities}
    p_j:=\mathbb E[Y_j\mid\mathcal F_{j-1}],
    \qquad
    X:=\sum_jY_j=|M|,
    \qquad
    A:=\sum_jp_j.
\end{equation}
We may include zero-success steps for inactive vertices, so that the
number of steps is fixed. The tower property gives
\begin{equation}\label{ub:eq:equal-expectations}
    \mathbb E A=\mathbb E X.
\end{equation}

For a left vertex $u$ active in round $r$, define
\[
    d_r(u):=|N(u)\cap R_r|.
\]
Let $f_{r,u}$ be the fraction of these neighbours that are still
unmatched immediately before $u$ is processed. When $q<n$, use the
moment immediately before $u$'s group is processed, before its selected
vertex is exposed. Set $f_{r,u}:=0$ for an inactive vertex or when
$d_r(u)=0$.

\paragraph{Surviving edges force progress.}
We first make the argument from the proof overview precise. For a real
number $z$, write $(z)_+:=\max\{z,0\}$.

\begin{lemma}\label{ub:lem:progress}
For every execution of the algorithm,
\begin{equation}\label{ub:eq:progress}
    A\ge c_0(\mu-2X)_+
        \min\left\{1,\frac{ks}{(X+1)^{1/k}}\right\},
\end{equation}
where $c_0>0$ is an absolute constant. When $q=n$, the stronger bound
\begin{equation}\label{ub:eq:one-sample-progress}
    A\ge\frac{k(\mu-2X)_+}{(X+1)^{1/k}}
\end{equation}
holds. When $q<n$, we also have
\begin{equation}\label{ub:eq:small-budget-progress}
    A\ge\frac{ks(\mu-2X)_+}{2(X+1)^{1/k}}.
\end{equation}
\end{lemma}

\begin{proof}
Fix a maximum matching $M^\star$. At most $2X$ of its edges have an
endpoint matched by the algorithm. Consider an edge $uv\in M^\star$
whose endpoints are both unmatched at termination. The vertex $u$
is active in every round, and $d_{k+1}(u)\ge1$ because $v\in R_{k+1}$.
Since the sets of unmatched right vertices only shrink,
\begin{equation}\label{ub:eq:availability-ratio}
    f_{r,u}\ge\frac{d_{r+1}(u)}{d_r(u)}
    \qquad (1\le r\le k).
\end{equation}
Also, exactly $X$ right vertices are matched over the entire execution,
so
\[
    d_1(u)-d_{k+1}(u)\le X.
\]
Multiplying~\eqref{ub:eq:availability-ratio} over the rounds gives
\[
    \prod_{r=1}^k f_{r,u}
    \ge\frac{d_{k+1}(u)}{d_1(u)}
    \ge\frac{d_{k+1}(u)}{d_{k+1}(u)+X}
    \ge\frac1{X+1}.
\]
The arithmetic--geometric mean inequality therefore gives
\begin{equation}\label{ub:eq:sum-availability}
    \sum_{r=1}^k f_{r,u}
        \ge\frac{k}{(X+1)^{1/k}}.
\end{equation}
We now relate these fractions to the conditional probabilities
in~\eqref{ub:eq:conditional-probabilities}.

\paragraph{When $q\ge n$.}
Immediately before an active vertex $u$ is processed, its $t$ samples
are independent and uniform in its round-start neighbourhood. Thus,
when that neighbourhood is nonempty, its conditional probability of
being matched at this step is
\[
    1-(1-f_{r,u})^t
    \ge(1-e^{-1})\min\{1,tf_{r,u}\}
    \ge c_0\min\{1,sf_{r,u}\}.
\]
The last inequality uses $t=\lfloor s\rfloor\ge s/2$. The lower bound
also holds when the neighbourhood is empty, because then $f_{r,u}=0$.
For any nonnegative numbers $z_1,\ldots,z_k$,
\[
    \sum_{r=1}^k\min\{1,z_r\}
        \ge\min\left\{1,\sum_{r=1}^kz_r\right\}.
\]
Consequently, by~\eqref{ub:eq:sum-availability}, the left endpoint of
each surviving edge of $M^\star$ contributes at least
\[
    c_0\min\left\{1,\frac{ks}{(X+1)^{1/k}}\right\}
\]
to $A$. Distinct matching edges have distinct left endpoints. Summing
over the at least $(\mu-2X)_+$ surviving edges
proves~\eqref{ub:eq:progress} in this case.

If $q=n$, there is exactly one sample per active left vertex. Its
conditional probability of being matched is exactly $f_{r,u}$.
We can therefore sum~\eqref{ub:eq:sum-availability} directly, without
truncation or a constant loss. This
proves~\eqref{ub:eq:one-sample-progress}.

\paragraph{When $q<n$.}
Consider a group $B$ in round $r$. Conditional on the information
exposed before this group is processed, its selected vertex is uniform
in $B$. Conditional on selecting an active vertex $u$, the probability
that its sample is still available is $f_{r,u}$. Hence this group's
conditional probability of adding an edge is exactly
\[
    \frac1{|B|}\sum_{u\in B\cap L_r}f_{r,u}.
\]
The partition is balanced, and $q<n$, so
\[
    |B|\le\left\lceil\frac nq\right\rceil\le\frac{2n}{q}.
\]
The group's conditional success probability is therefore at least
$(s/2)\sum_{u\in B\cap L_r}f_{r,u}$. Sum over the groups and rounds,
and keep only the contribution from the surviving edges of
$M^\star$. Equation~\eqref{ub:eq:sum-availability} gives
\[
    A\ge\frac{s}{2}(\mu-2X)_+
             \frac{k}{(X+1)^{1/k}},
\]
which proves~\eqref{ub:eq:small-budget-progress}.
It also implies~\eqref{ub:eq:progress}, after decreasing $c_0$ if
necessary.
\end{proof}

\paragraph{Relating conditional progress to the output.}
Equation~\eqref{ub:eq:equal-expectations} is sufficient for the
expected-size bound. For the probability guarantees, we use the
following exponential bound.

\begin{lemma}\label{ub:lem:exponential-bound}
For every $\theta>0$,
\begin{equation}\label{ub:eq:exponential-bound}
    \mathbb E\!\left[
        \exp\!\left(-\theta X+(1-e^{-\theta})A\right)
    \right]\le1.
\end{equation}
\end{lemma}

\begin{proof}
At a single processing step,
\begin{align*}
    \mathbb E[e^{-\theta Y_j}\mid\mathcal F_{j-1}]
      &=1-(1-e^{-\theta})p_j\\
      &\le\exp\!\left(-(1-e^{-\theta})p_j\right).
\end{align*}
Thus
\[
    \mathbb E\!\left[
        e^{-\theta Y_j+(1-e^{-\theta})p_j}
        \mid\mathcal F_{j-1}
    \right]\le1.
\]
Multiplying the factors and applying conditional expectation one step
at a time proves~\eqref{ub:eq:exponential-bound}.
\end{proof}

\subsection{Proof of Theorem~\ref{ub:thm:main}}
\label{ub:sec:main-proof}

We have already checked that the output is a matching and that the
algorithm respects the query budget. It remains to prove the two
approximation guarantees.

\paragraph{The expected size of the matching.}
Write $m:=\mathbb E X$ and $z:=ks$. If $m\ge\mu/8$, then
\eqref{ub:eq:expected-size} follows immediately. We therefore assume
$m<\mu/8$.

First, $m>0$. Otherwise, $X=0$ almost surely.
Lemma~\ref{ub:lem:progress} would then give
\[
    A\ge c_0\mu\min\{1,z\}>0
\]
almost surely, contradicting $\mathbb E A=\mathbb E X=0$.

By Markov's inequality, $\Pr[X\le2m]\ge1/2$. On this event,
$\mu-2X\ge\mu/2$ and $X+1\le2m+1$. Taking expectations
in~\eqref{ub:eq:progress} and
using~\eqref{ub:eq:equal-expectations}, we obtain
\begin{equation}\label{ub:eq:expected-progress}
    m\ge c_1\mu
        \min\left\{1,\frac{z}{(2m+1)^{1/k}}\right\},
\end{equation}
where $c_1>0$ is an absolute constant.

If $m<1$, then $(2m+1)^{1/k}\le3$, and
\eqref{ub:eq:expected-progress} gives
\[
    m\ge\frac{c_1}{3}\mu\min\{1,z\}.
\]
This is sufficient, since
$\min\{\mu,\mu z,(\mu z)^{k/(k+1)}\}\le\mu\min\{1,z\}$.

Now suppose $m\ge1$. If $z\ge(2m+1)^{1/k}$, then
\eqref{ub:eq:expected-progress} gives $m\ge c_1\mu$, which again
suffices. Otherwise, using $2m+1\le3m$, we get
\[
    m^{1+1/k}\ge\frac{c_1}{3}\mu z.
\]
Therefore
\[
    m=\Omega\!\left((\mu z)^{k/(k+1)}\right),
\]
with an absolute implicit constant independent of $k$. Combining the
cases proves~\eqref{ub:eq:expected-size}.

Dividing this expected-size bound into $\mu$ gives approximation
factor
\[
    O\!\left(
        \max\left\{
            1,\frac1{ks},
            \frac{\mu^{1/(k+1)}}{(ks)^{k/(k+1)}}
        \right\}
    \right).
\]
Since $\mu\le n$ and $nks=kq\ge1$, we have
\[
    \frac1{ks}
        \le\frac{n^{1/(k+1)}}{(ks)^{k/(k+1)}}.
\]
Thus the factor is bounded by~\eqref{ub:eq:approximation}.

\paragraph{The constant-error guarantee.}
Assume $q\ge n$ and $\mu>0$. Set
\[
    B:=\min\left\{\mu,(\mu ks)^{k/(k+1)}\right\},
    \qquad x:=c_3B,
\]
where $c_3>0$ is a sufficiently small absolute constant. We show that
$\Pr[X<x]\le1/3$.

First suppose $x\ge1$. Choose $c_3\le1/4$, so that $x\le\mu/4$.
On the event $X<x$, we have $\mu-2X\ge\mu/2$ and $X+1\le2x$.
Lemma~\ref{ub:lem:progress} therefore gives
\[
    A\ge\frac{c_0\mu}{4}
           \min\left\{1,\frac{ks}{x^{1/k}}\right\}.
\]
The definition of $B$ gives $B\le\mu$ and
$B^{1+1/k}\le\mu ks$. Since $x\le B$, it follows that
\[
    \mu\min\left\{1,\frac{ks}{x^{1/k}}\right\}\ge B.
\]
Taking $c_3\le c_0/32$ yields $A\ge8x$ on the event $X<x$.
Apply Lemma~\ref{ub:lem:exponential-bound} with $\theta=1$. On this
event, the exponent in~\eqref{ub:eq:exponential-bound} is at least
$[8(1-e^{-1})-1]x$. Hence
\[
    \Pr[X<x]
        \le\exp\!\left(-[8(1-e^{-1})-1]x\right)
        \le e^{-4x}<\frac13.
\]

If $x<1$, the event $X<x$ is impossible. On a nonempty graph, the
first round with at least one sample per left vertex always adds
an edge. Thus $X\ge1$ deterministically.

We have proved that $X\ge c_3B$ with probability at least $2/3$.
Because $ks\ge1$, this is the same output-size bound
as~\eqref{ub:eq:expected-size}, up to an absolute constant. It gives
the claimed approximation guarantee and completes the proof of
Theorem~\ref{ub:thm:main}.

\paragraph{An explicit bound with one sample per vertex.}
The stronger inequality~\eqref{ub:eq:one-sample-progress} also gives
the following useful form of the expected-size guarantee when $q=n$:
\begin{equation}\label{ub:eq:explicit-one-sample}
    \mathbb E X
        \ge\frac{\mu}{2+(\mu+1)^{1/k}/k}.
\end{equation}
Indeed, $X\le\mu$, so~\eqref{ub:eq:one-sample-progress} implies
\[
    A\ge\frac{k(\mu-2X)}{(\mu+1)^{1/k}}.
\]
This inequality also holds when $\mu-2X<0$, because $A\ge0$.
Taking expectations, using $\mathbb E A=\mathbb E X$, and
rearranging proves~\eqref{ub:eq:explicit-one-sample}.

\subsection{Optimal-Order Rounds for Constant Approximation}
\label{ub:sec:constant-approximation}

We now use the same one-sample algorithm to obtain a constant
approximation with high probability. Summing the progress over all
rounds is important here: it gives the required guarantee after
$O(\log n/\log\log n)$ rounds rather than $O(\log n)$ rounds.
In the following theorem, $\log$ denotes the natural logarithm.

\begin{theorem}\label{ub:thm:constant-approximation}
Fix $\varepsilon>0$. For sufficiently large $n$, the
independent-neighbour greedy algorithm, with one random-neighbour
demand query per active left vertex per round, returns a
$(2+\varepsilon)$-approximate maximum matching in
\[
    k:=\left\lceil\frac{2\log n}{\log\log n}\right\rceil
\]
rounds with high probability. It uses at most $n$ demand queries per
round, and its failure probability is at most
\[
    \exp\!\left(-\Omega_\varepsilon(k^2)\right)
        =n^{-\omega(1)}.
\]
\end{theorem}

\begin{proof}
The assertion is immediate when $\mu=0$, so assume $\mu>0$.
Let
\[
    \mathcal B:=\left\{X<\frac{\mu}{2+\varepsilon}\right\}
\]
be the failure event. We first show that, on $\mathcal B$, the total
conditional matching probability $A$ is large.

\paragraph{Every round adds an edge on the failure event.}
Whenever the current matching has size less than $\mu/2$, some edge
of a fixed maximum matching has both endpoints unmatched. The
residual graph is therefore nonempty. Since every active left vertex
asks a random-neighbour query, the next round adds at least one edge.

On $\mathcal B$, the matching size stays below
$\mu/(2+\varepsilon)<\mu/2$ throughout the execution. Hence every
round adds an edge, and
\begin{equation}\label{ub:eq:failure-implies-large-optimum}
    X\ge k,
    \qquad
    \mu>(2+\varepsilon)k.
\end{equation}

\paragraph{The conditional progress is much larger than the output.}
By~\eqref{ub:eq:one-sample-progress}, on $\mathcal B$ we have
\begin{equation}\label{ub:eq:constant-approximation-progress}
    A\ge\frac{\varepsilon}{2+\varepsilon}
           \frac{k\mu}{(\mu+1)^{1/k}}.
\end{equation}
The function $y/(y+1)^{1/k}$ is increasing for $y>0$. Together
with~\eqref{ub:eq:failure-implies-large-optimum}, this gives
\begin{equation}\label{ub:eq:quadratic-progress}
    A=\Omega_\varepsilon(k^2)
    \qquad\text{on }\mathcal B.
\end{equation}
Also, dividing~\eqref{ub:eq:constant-approximation-progress} by
$X<\mu/(2+\varepsilon)$ and using $\mu\le n$, we obtain
\[
    \frac AX
        \ge\frac{\varepsilon k}{(n+1)^{1/k}}.
\]
For the chosen number of rounds,
\[
    (n+1)^{1/k}=O\!\left(\sqrt{\log n}\right),
    \qquad
    \frac{k}{(n+1)^{1/k}}\longrightarrow\infty.
\]
Consequently $A\ge4X$ on $\mathcal B$ for sufficiently large $n$.

\paragraph{Bounding the failure probability.}
Apply Lemma~\ref{ub:lem:exponential-bound} with $\theta=1$.
On $\mathcal B$, the exponent satisfies
\[
    -X+(1-e^{-1})A
        \ge\left(1-e^{-1}-\frac14\right)A
        =\Omega_\varepsilon(k^2),
\]
where the last equality uses~\eqref{ub:eq:quadratic-progress}.
Markov's inequality therefore gives
\[
    \Pr[\mathcal B]
        \le\exp\!\left(-\Omega_\varepsilon(k^2)\right).
\]
Finally, $k^2/\log n\longrightarrow\infty$, so this probability is
smaller than $n^{-c}$ for every fixed $c>0$ once $n$ is sufficiently
large. This proves the theorem.
\end{proof}

Combined with the round lower bound from the preceding section,
Theorem~\ref{ub:thm:constant-approximation} shows that the number of
rounds required for constant approximation with $O(n)$ demand queries
per round is
\[
    \Theta\!\left(\frac{\log n}{\log\log n}\right)
\]
up to constant factors. The lower bound already holds for constant
success probability, while the upper bound above holds with high
probability.

\begin{corollary}[Randomised minimum vertex cover]
\label{ub:cor:vertex-cover}
For every fixed $\varepsilon>0$, there is a randomised algorithm that
outputs a $(2+\varepsilon)$-approximate minimum vertex cover with high
probability in $O(\log n/\log\log n)$ rounds, using at most $n$ demand
queries per round.
\end{corollary}
\begin{proof}
Run the one-sample algorithm for
$k=\lceil 2\log n/\log\log n\rceil$ rounds.  Let $M$ be its matching,
$X=|M|$, and $A$ the sum of conditional matching probabilities defined
in~\eqref{ub:eq:conditional-probabilities}.  In one additional round,
query every unmatched left vertex, placing all unmatched right vertices
first in its demand ordering.  Let $U$ be the unmatched left vertices
for which the response lies among the unmatched right vertices, and
output
\[
    C:=V(M)\cup U.
\]
Every edge has an endpoint in $C$: this is immediate for an edge
incident to a matched vertex, while the left endpoint of an edge with
both endpoints unmatched belongs to $U$.  Thus $C$ is a vertex cover.

For every $u\in U$, a neighbour of $u$ remains unmatched throughout
the execution.  The argument proving~\eqref{ub:eq:sum-availability}
therefore applies to $u$, whether or not it is the endpoint of a
surviving edge of the fixed optimum matching.  With one sample per
active left vertex, its contribution to $A$ is at least
$k/(X+1)^{1/k}$.  Since distinct vertices of $U$ contribute distinct
processing steps, we have, for every execution,
\begin{equation}\label{ub:eq:cover-progress}
    A\ge\frac{k|U|}{(X+1)^{1/k}}.
\end{equation}

Fix $c>0$.  Lemma~\ref{ub:lem:exponential-bound} with $\theta=1$
implies
\begin{equation}\label{ub:eq:cover-compensator-tail}
    \Pr\!\left[A>4X+4(c+2)\log n\right]\le n^{-c-2}.
\end{equation}
Indeed, on this event the exponent $-X+(1-e^{-1})A$ is at least
$(c+2)\log n$, so Markov's inequality applies.

If $U=\varnothing$, then $|C|=2X\le2\mu$.  Otherwise an edge remains
between unmatched vertices after the $k$ matching rounds.  Such an edge
also exists at the beginning of every earlier round.  As all unmatched
left vertices sample a neighbour in each round, each round adds an edge.
Therefore $X\ge k$ and $\mu\ge k$.  On the event complementary to
\eqref{ub:eq:cover-compensator-tail}, using
\eqref{ub:eq:cover-progress} and $X\le\mu\le n$ gives
\[
    \frac{|U|}{\mu}
    \le\frac{(n+1)^{1/k}}{k}
            \left(4+\frac{4(c+2)\log n}{\mu}\right)
    \le\frac{(n+1)^{1/k}}{k}
            \left(4+\frac{4(c+2)\log n}{k}\right)
    =o(1).
\]
The final equality follows from
$(n+1)^{1/k}=O(\sqrt{\log n})$ and
$k=\Theta(\log n/\log\log n)$.  Hence
$|C|=2X+|U|\le(2+\varepsilon)\mu$ for all sufficiently large $n$.
By K\"onig's theorem, $\mu$ is also the minimum vertex-cover size.
Since $c$ was arbitrary, the result holds with high probability.
\end{proof}

\begin{proposition}[The single-edge obstruction]
\label{prop:single-edge-obstruction}
On the distribution consisting of one uniformly random edge in
$L\times R$, any randomised algorithm making at most $Q$ demand
queries in total has probability at most
$\min\{1,Q/n+1/n^2\}$ of returning a valid nonempty matching.
For an algorithm that always returns a valid matching on every graph,
its expected output size on this distribution is at most
$\min\{1,Q/n\}$.
\end{proposition}
\begin{proof}
Fix the random tape.  Until the unique edge is found, the responses
are exactly those on the empty graph.  Along that all-negative
execution, at most $Q$ distinct left vertices are queried.  The
probability that the unique edge's left endpoint is among them is at
most $Q/n$.  If no queried row contains it, a valid nonempty output
must be the single correct edge.  The all-negative transcript fixes
this possible guessed pair, whose unconditional probability is at
most $1/n^2$.  An always-valid algorithm must instead output the empty
matching on this transcript, because it is also the transcript of the
empty graph.  Average over the random tape.
\end{proof}

\subsection{The Long-Round, Small-Budget Regime}
\label{ub:sec:long-rounds}

\begin{theorem}[Tight expected-size tradeoff for many rounds]
\label{ub:thm:long-rounds}
Suppose $1\le q\le n$ and $k\ge2\log(n+1)$.  The optimal worst-case
expected-size approximation factor among always-valid randomised
algorithms, with at most $q$ demand queries per round, is
\[
 \Theta\!\left(\max\left\{1,\frac{n}{kq}\right\}\right).
\]
\end{theorem}
\begin{proof}
Equations~\eqref{ub:eq:small-budget-progress} and
\eqref{ub:eq:one-sample-progress} imply, in both cases $q\le n$,
\[
 A\ge\frac{ks}{2(n+1)^{1/k}}(\mu-2X).
\]
Take expectations, use \eqref{ub:eq:equal-expectations}, and rearrange:
\[
 \mathbb E X\ge
 \frac{\mu}{2+2(n+1)^{1/k}/(ks)}
 =\Omega\bigl(\mu\min\{1,ks\}\bigr).
\]
Here $(n+1)^{1/k}\le\sqrt e$.  Conversely,
Proposition~\ref{prop:single-edge-obstruction}, with $Q=kq$, supplies an input
with optimum one and expected output at most
$\min\{1,kq/n\}$.  The trivial lower bound of one on an approximation
factor covers $kq\ge n$.
\end{proof}

For $q=n^{1+\delta}\le n$, this is
$\Theta(\max\{1,(kn^\delta)^{-1}\})$.  In particular, the expression
saturates at a constant when $kn^\delta\ge1$; it must not be interpreted
as an approximation factor smaller than one.

\subsection{Tightness for Independent-Neighbour Greedy}
\label{ub:sec:greedy-tightness}

The next lower bound applies to the independent-neighbour greedy
algorithm of Section~\ref{ub:sec:algorithm} with $q=n$.  It proves that the
$n^{1/(k+1)}$ dependence in our upper-bound analysis is inherent to
this particular algorithm, rather than an artifact of the proof.

\paragraph{Why this does not close the upper--lower bound gap.}
It is important to distinguish an \emph{algorithm-specific} lower
bound from a lower bound in the demand-query model.  Let
$\alpha_k^\star(n)$ denote the best worst-case approximation factor
achievable by an arbitrary $k$-round randomized demand-query
algorithm using $O(n)$ queries per round.  Ignoring polylogarithmic
and polynomial-in-$k$ factors, our model-wide results currently give
\[
   n^{1/(2k)}
   \;\lesssim\;
   \alpha_k^\star(n)
   \;\lesssim\;
   n^{1/(k+1)}.
\]
The theorem below shows only that the specific independent-neighbour
greedy algorithm attains the right-hand exponent in the worst case:
its approximation factor is
$\Omega(n^{1/(k+1)}/k)$ on an explicit family of graphs.  Hence no
sharper analysis of this algorithm can improve its exponent to
$1/(2k)$.

This does \emph{not} imply the same lower bound for an arbitrary
demand-query algorithm.  A general demand query reveals more than a
single sampled neighbour: if a query returns $v$, then every vertex
placed before $v$ in the query permutation is simultaneously
certified to be a non-neighbour.  An unrestricted algorithm can
coordinate these prefixes across many left vertices and use the
resulting negative information in later rounds.  By contrast,
the algorithm in Section~\ref{ub:sec:algorithm} uses demand queries only as independent
random-neighbour samples and does not exploit such correlated
certificates.  Indeed, shared-prefix filtering can substantially
outperform independent-neighbour sampling on natural layered
instances.  Therefore, closing the remaining exponent gap requires
either a stronger lower bound that remains hard against these
correlated-prefix strategies, or a new algorithm that exploits the
additional information exposed by demand queries.

\begin{theorem}[Algorithm-specific lower bound]\label{ub:thm:greedy-tightness}
There is an absolute constant $C>0$ such that, whenever
\[
 n^{1/(k+1)}\ge C(k+\log n),
\]
there is a bipartite graph with $\mu(G)\ge n/2$ on which the one-sample
independent-neighbour greedy algorithm satisfies
\[
 \mathbb E|M|\le Ck n^{k/(k+1)}.
\]
Its worst-case expected-size approximation is therefore
$\Omega(n^{1/(k+1)}/k)$.  In this range, this matches
\eqref{ub:eq:approximation} up to absolute constants, since
$k^{1/(k+1)}$ is bounded.
\end{theorem}
\begin{proof}
Fix a sufficiently large absolute constant $B$, and put
\[
 a=\left(\frac{n}{B^{k-1}}\right)^{1/(k+1)},
 \qquad h=\lfloor n/a\rfloor,
 \qquad p_j=(Ba)^{-j}\quad(0\le j<k).
\]
Choosing the theorem's constant large enough ensures
$a\ge C_1\log n$ and $a\ge2k$, where $C_1$ is as large as required
below.  Partition $R$ into layers $T_0,\ldots,T_{k-1}$ of size $h$
and a set $P$ of size $n-kh\ge n/2$.  Independently insert every pair
from $L$ to $T_j$ with probability $p_j$.  Give each vertex of $P$
exactly one neighbour, with these neighbours distinct.  These private
edges form a matching of size $|P|$.  The useful parameter relations,
including the harmless rounding of $h$, are
\begin{equation}\label{ub:eq:greedy-hard-parameters}
 hp_{k-1}=\Theta(a),\qquad
 (Ba)^{k-1}=n/a^2=\Theta(h/a).
\end{equation}

\paragraph{A property uniform over remaining subsets.}
With positive probability, simultaneously for every layer $T_j$,
\begin{equation}\label{ub:eq:greedy-regularity}
 \deg_{T_j}(u)\le2hp_j\quad\text{for every }u\in L,
\end{equation}
and, for every $W\subseteq T_j$ with $|W|\ge3h/4$,
\begin{equation}\label{ub:eq:greedy-robustness}
 \bigl|\{u\in L:\deg_W(u)<hp_j/4\}\bigr|
 \le D/p_j
\end{equation}
for an absolute constant $D$.  The first assertion follows from
Chernoff bounds because $hp_j\ge c a\ge cC_1\log n$.
For the second, a violation gives a set $U$ of
$m=\lceil D/p_j\rceil$ vertices and a set $W$ with
$e(U,W)<mhp_j/4$, whereas its expectation is at least
$3mhp_j/4$.  For fixed $U,W$, this event has probability
$\exp(-\Omega(mhp_j))=\exp(-\Omega(Dh))$.
There are at most $2^h(en/m)^m$ choices.  Since
$hp_j\ge cC_1\log n$, choosing $C_1$ sufficiently large absorbs the
factor $(en/m)^m$, and choosing $D$ sufficiently large absorbs $2^h$.
A union bound over the layers proves the claim.  If $m>n$, the
assertion is automatic.  Fix a graph satisfying these properties.
The uniformity over $W$ permits its remaining subset to be chosen
adaptively by the algorithm.

\paragraph{Later layers survive until their designated rounds.}
Stop the execution if, at the beginning of some round $r$, fewer than
$3h/4$ vertices of $T_{r-1}$ remain unmatched.  Before stopping,
\eqref{ub:eq:greedy-robustness} implies that all but $D/p_{r-1}$ left
vertices have at least $hp_{r-1}/4$ available neighbours in
$T_{r-1}$.  For each nonexceptional active vertex, the probability of
sampling a neighbour in a later layer $T_j$, $j\ge r$, is at most
\[
 \frac{2hp_j}{hp_{r-1}/4}
 =8(Ba)^{-(j-r+1)}.
\]
For fixed $j\ge1$, the sum of the conditional expectations of samples
entering $T_j$ in rounds $1,\ldots,j$, before stopping, is at most
\begin{align*}
 8n\sum_{r=1}^j(Ba)^{-(j-r+1)}
        +D\sum_{r=1}^j(Ba)^{r-1}
 &\le \frac{16n}{Ba}+2D(Ba)^{j-1}\\
 &\le h/16.
\end{align*}
The final inequality follows by first taking $B$ sufficiently large
and then $C_1$ sufficiently large; for $j\le k-1$,
$(Ba)^{j-1}/h=O(1/(Ba^2))$.

For completeness, if adapted indicators $I_i$ have conditional means
$z_i$ with $\sum_i z_i\le h/16$ on every history, then
\[
 \mathbb E e^{\sum_i I_i}
 \le e^{(e-1)h/16},
 \qquad
 \Pr\!\left[\sum_i I_i\ge h/4\right]
 \le e^{-\Omega(h)}.
\]
This follows by conditioning one indicator at a time, using
$\mathbb E[e^{I_i}\mid\mathcal F_{i-1}]
 \le e^{(e-1)z_i}$, or equivalently the associated exponential
supermartingale.  Apply this to the stopped sample counts for each
$j$.  A right vertex can be matched only if it is sampled.  A union
bound thus shows that stopping occurs with probability at most
$k e^{-\Omega(h)}$.  Round one starts with its entire first layer, so
there is no initial failure event.

\paragraph{Few private edges are sampled.}
Before stopping, a nonexceptional vertex in round $r$ samples its
private edge with probability at most $4/(hp_{r-1})$; an exceptional
vertex contributes at most one.  The expected number of private-edge
samples in the stopped execution is at most
\[
 \sum_{r=1}^k\left(\frac{4n}{hp_{r-1}}+\frac{D}{p_{r-1}}\right)
 =O\!\left((a+D)(Ba)^{k-1}\right)=O(h),
\]
by \eqref{ub:eq:greedy-hard-parameters}.  If stopping occurs, the final
matching has at most $n$ edges, whose expected extra contribution is
at most $nk e^{-\Omega(h)}=o(h)$.  All nonprivate matching edges use
one of the $kh$ dense-layer vertices.  Therefore
\[
 \mathbb E|M|\le kh+O(h)=O(kn/a)=O(k n^{k/(k+1)}),
\]
since $B^{(k-1)/(k+1)}\le B$.  The private matching certifies
$\mu(G)\ge n/2$, completing the proof.
\end{proof}


\section{Algorithms using Demand-OR queries}\label{sec:demand-or}

A Demand-OR query specifies $u\in L$ and $S\subseteq R$ and returns
whether $N(u)\cap S$ is nonempty.  We show that the random-neighbour
queries used in \Cref{ub:sec:upper-bound} can be replaced by nonadaptive
Demand-OR procedures.  We simulate \emph{uniform random-neighbour
queries}, not demand queries with an arbitrary prescribed ordering.
This distinction is important when the number of adaptive rounds is
small.

The construction uses geometric subsampling, as in support-sampling
algorithms~\cite{jowhari2011tight}.  Its singleton test is deterministic.
Consequently a successful call always returns a genuine neighbour,
and its conditional output distribution is exactly uniform.  The
only error is a detectable failure to produce a neighbour.

\begin{theorem}[Demand-OR upper bounds]\label{thm:demand-or-upper}
Let $1\le q\le n^2$ and $k\ge1$.  Using $k$ rounds and
$O(q\log^2 n)$ Demand-OR queries per round, an always-valid randomised
algorithm obtains an
$O\!\left(\max\left\{1,\frac{n}{(kq)^{k/(k+1)}}\right\}\right)$-approximation
in expectation.  For $q\ge n$, the same asymptotic factor is obtained
with probability at least $2/3$.  Moreover, for every fixed
$\varepsilon>0$, a $(2+\varepsilon)$-approximation is obtained with high
probability in $O(\log n/\log\log n)$ rounds using
$O(n\log^3 n)$ Demand-OR queries per round.
\end{theorem}

The guarantee in expectation cannot be replaced by a uniform
high-probability guarantee over all small budgets: the
single-edge obstruction of \Cref{prop:single-edge-obstruction} already
applies to the stronger demand-query model.

\subsection{Exact singleton detection and recovery}
\label{sec:singleton-recovery}

\begin{lemma}[An exact nonadaptive singleton primitive]
\label{clm:demand-or-query-recover-one-neighbour}
\label{clm:demand-or-query-detect-single-neighbour}
\label{alg:recover-single-neighbour}
\label{alg:detect-single-neighbour}
Given $u\in L$ and $S\subseteq R$, $1+2\lceil\log_2 n\rceil$ 
deterministic nonadaptive Demand-OR queries distinguish the cases
$|N(u)\cap S|=0$, $|N(u)\cap S|=1$, and $|N(u)\cap S|\ge2$.
In the singleton case they also recover the unique neighbour.
\end{lemma}
\begin{proof}
Index the vertices of $R$ by $0,1,\ldots,n-1$, and use the binary
representation of each vertex's index.  Let $d=\lceil\log_2 n\rceil$.
Ask the query $(u,S)$ and, for every bit position
$i\in\{0,\ldots,d-1\}$, both queries
\[
 (u,S_{i,0}),\quad (u,S_{i,1}),
 \qquad
 S_{i,b}=\{v\in S:\text{bit }i\text{ of the index of }v\text{ is }b\}.
\]
If the first answer is zero, the neighbourhood is empty.  Otherwise,
a singleton neighbourhood makes exactly one answer in every bit pair
positive.  Conversely, two distinct neighbours differ at some bit,
making both answers of that pair positive.  Thus exactly one positive
answer in every pair is equivalent to a singleton.  The positive
sides identify all bits of its index and hence the neighbour itself.
All sets are fixed before any response is obtained.
\end{proof}

\paragraph{Why both bit sides are queried.}
Using only the bit-one sets recovers a label under a singleton
\emph{promise}, but does not verify the promise.  For example,
neighbours with binary indices $01$ and $10$ would produce the bit pattern
$11$, which need not be a neighbour.  The paired queries rule out
such false acceptances deterministically and remove the need to
condition a sampler's uniformity claim on a potentially biased
random testing event.

\subsection{Exact uniformity with detectable failure}

\begin{lemma}[Nonadaptive random-neighbour sampling]
\label{lemma:demand-or-query-rand-neighbour}
For $0<\delta<1$, there is a one-round randomised procedure using
\[
 O\bigl(\log^2 n\,(1+\log(1/\delta))\bigr)
\]
Demand-OR queries with the following properties on input $(u,S)$.
If $N(u)\cap S=\varnothing$, it returns \textsc{No neighbour}
with certainty.  Otherwise it either returns a valid neighbour or
returns \textsc{Fail}, with failure probability at most $\delta$.
Conditional on returning a neighbour, each element of $N(u)\cap S$
is returned with exactly the same probability.
\end{lemma}

\begin{Algorithm}\label{alg:rand-neighbour-demand-or}
An algorithm to sample a random neighbour using Demand-OR queries.

\medskip

\textbf{Input:} Demand-OR query access to a graph, a vertex $u\in L$, a set $S\subseteq R$, and a failure probability $\delta\in(0,1)$.

\medskip

\textbf{Output:} A uniformly random neighbour of $u$ in $S$, or \textsc{No neighbour} if no neighbour exists, or \textsc{Fail}.

\medskip

\textbf{Parameters:} $d=\lceil\log_2 n\rceil$ and $T=\max\{1,\lceil2e\log(1/\delta)\rceil\}$.

\smallskip

\textbf{Queries:}
\begin{enumerate}
    \item Query $(u,S)$.
    \item For every $j\in\{0,\ldots,d\}$ and every $t\in\{1,\ldots,T\}$:
    \begin{enumerate}
        \item Construct $S_{j,t}\subseteq S$ by including each vertex independently with probability $2^{-j}$.
        \item Run the exact singleton primitive on $S_{j,t}$.
    \end{enumerate}
\end{enumerate}

\textbf{Postprocessing:}
\begin{enumerate}
    \item If the query $(u,S)$ returns $0$, output \textsc{No neighbour}.
    \item Let $(j^*,t^*)$ be the lexicographically first pair for which the singleton primitive succeeds. Output the neighbour recovered by that copy of the singleton primitive.
    \item If no such pair exists, output \textsc{Fail}.
\end{enumerate}
\end{Algorithm}

\begin{proof}[Proof of \Cref{lemma:demand-or-query-rand-neighbour}]
Write $s=|N(u)\cap S|$.  Emptiness is detected exactly by $(u,S)$.
For $s\ge1$, let $K=2^{\lceil\log_2 s\rceil}$, i.e., the smallest power
of two that is at least $s$.  Then $s\le K<2s$.  For a trial at the
scale corresponding to $K$, exactly one neighbour is isolated with
probability
\[
 \frac sK(1-1/K)^{s-1}\ge\frac1{2e}.
\]
For $s=1$, the scale $K=1$ isolates that neighbour with certainty;
for $s\ge2$, use $s/K\ge1/2$ and
$(1-1/s)^{s-1}\ge e^{-1}$.  Independent repetitions at this scale
fail simultaneously with probability at most
$(1-1/(2e))^T\le\delta$.

There are $O(\log n)$ scales and $T$ trials per scale.  Each trial
uses $O(\log n)$ deterministic queries.  This proves the query count.
All random subsets and all their bit partitions are fixed before
responses are used, so there is no additional adaptive round.

Finally, for a fixed trial with sampling probability $p$ and a fixed
neighbour $v$, the probability that $v$ is the unique sampled
neighbour is
$p(1-p)^{s-1}$, independent of the identity of $v$.
The trials are independent.  Thus the probability that $v$ is
returned by the first successful trial is
\[
 \sum_{(j,t)}
 \Pr[\text{all earlier trials are nonsingletons}]
      \;2^{-j}(1-2^{-j})^{s-1},
\]
again independent of $v$.  The singleton test is exact, so this
calculation is unconditional apart from the fixed input.  Summing
over neighbours and conditioning on a return proves exact
uniformity.  Every returned pair is a true edge.
\end{proof}

\subsection{Transferring the matching guarantees}
\label{sec:demand-or-transfer}

\begin{proof}[Expected-size and constant-error parts of
\Cref{thm:demand-or-upper}]
Replace each random-neighbour call in the algorithm of
\Cref{ub:sec:algorithm} by an independent copy of the isolation
sampler with $\delta=1/16$.  Treat \textsc{Fail} as an empty sample.
Each round uses $O(q\log^2 n)$ queries, and the output is always a
matching.

Fix the input graph and the history before processing a vertex.  Its
round-start neighbourhood is fixed; let $\pi$ be the sampler's
probability of returning a neighbour from it.  If the neighbourhood
is nonempty, $\pi\ge15/16$, and exact conditional uniformity gives
probability $\pi f_{r,u}$ of a returned neighbour still being
available.  With $t$ independent calls the conditional matching
probability is
\[
 1-(1-\pi f_{r,u})^t
 \ge c\min\{1,(q/n)f_{r,u}\}.
\]
In the low-budget implementation, the conditional success probability
of a group $B$ is
\[
 \frac1{|B|}\sum_{u\in B\cap L_r}\pi_{r,u}f_{r,u}
 \ge\frac{15q}{32n}\sum_{u\in B\cap L_r}f_{r,u}.
\]
Here empty neighbourhoods contribute zero.  Thus the pathwise proof
of \Cref{ub:lem:progress} remains valid with a different
absolute constant.  The expectation argument gives the stated approximation guarantee,
including the optimum-sensitive bound
\eqref{ub:eq:expected-size}.

For $q\ge n$, the compensator proof of the constant-error guarantee
also applies.  In its small-threshold case $x<1$, the only relevant
failure is an empty output on a nonempty graph.  In the first round, choose any fixed
left vertex with a neighbour.  If a preceding step has already added
an edge, the output is nonempty.  Otherwise a successful sampler for
this vertex supplies an available edge.  Since it has at least one
call, the probability of an empty output is at most $1/16$.
This replaces the deterministic nonemptiness observation used in the
demand-query proof and completes the constant-error guarantee.
\end{proof}

\begin{proof}[High-probability constant-approximation part]
Use the sampler with failure probability $\delta=n^{-c-4}$ per call,
and run the one-sample algorithm for
$k=\lceil2\log n/\log\log n\rceil$ rounds.
Each simulated round now uses $O_c(n\log^3 n)$ Demand-OR queries.
Couple a call to an ideal uniform-neighbour call by using its exactly
uniform returned neighbour when it succeeds, and by drawing a fresh
uniform neighbour for the ideal call if it fails.  Equivalently,
conditional on the current input to the call, first draw a uniform
neighbour and an independent success coin, whose success probability
is the sampler's actual return probability.  This realises the joint
law required by the sampler lemma.

The real and ideal algorithms agree until the first failed call.
Every conditional failure probability is at most $\delta$, and there
are at most $kn$ calls.  Consequently the probability that the two
executions diverge is at most $kn\delta\le n^{-c-2}$ for sufficiently
large $n$.  By \Cref{ub:thm:constant-approximation}, the ideal execution
returns a $(2+\varepsilon)$-approximation with failure probability
$n^{-\omega(1)}$.  Combining the two bounds proves the claim.

This is a first-divergence coupling argument.  We do not condition the
entire adaptive execution on the event that every call succeeds;
such conditioning need not preserve the ideal execution's law.
\end{proof}

\begin{corollary}[Randomised vertex cover using Demand-OR queries]
\label{cor:demand-or-vertex-cover}
For every fixed $\varepsilon>0$, a $(2+\varepsilon)$-approximate minimum
vertex cover can be constructed with high probability in
$O(\log n/\log\log n)$ rounds using $O(n\log^3 n)$ Demand-OR queries
per round.
\end{corollary}
\begin{proof}
Simulate the matching phase of \Cref{ub:cor:vertex-cover} using the
high-probability coupling above.  In one final round, ask $(u,R')$ for
every unmatched left vertex $u$, where $R'$ is the set of unmatched
right vertices, and add $u$ to the matching endpoints exactly when
the answer is one.  These final Demand-OR queries identify the set
$U$ in \Cref{ub:cor:vertex-cover} without error, so the resulting cover
has the same guarantee.  The last round uses at most $n$ queries.
\end{proof}

\section{Deterministic round--query tradeoffs}\label{sec:deterministic}

Randomisation changes the round complexity substantially.  In the
deterministic setting we obtain matching polynomial dependence on the
number of rounds and the per-round query budget.  The remaining gap is
logarithmic, rather than a gap in the exponent of $n$.

The lower bound is a direct adversary for demand queries.  It keeps
all positive answers incident to a small right-hand set, while
preserving a linear-size matching in another graph consistent with
the same transcript.  Unlike an adversary that preserves a perfect
matching, it may discard a small fraction of left vertices.  This
freedom makes it possible to use small sets that hit \emph{most}, rather
than all, long query prefixes.  The upper bound uses a different idea:
recover several neighbours per vertex in parallel, and then simulate
a block of sequential greedy decisions without further queries.

Throughout this section, $n\ge2$, $1\le q\le n^2$, and $k\ge1$ are
integers.  Queries originate on the left side only.  Put
\[
 D=1+\lceil\log_2 n\rceil,
 \qquad \Lambda=\log\left(2+\frac{kq}{n}\right).
\]
A deterministic approximation must be valid on every graph, including
the empty graph.  The quantity being approximated is the maximum
matching size, but the algorithm must output an actual matching.

\begin{theorem}[Deterministic tradeoff]\label{thm:deterministic-tradeoff}
The following statements hold with absolute implicit constants.
\begin{enumerate}
 \item If $kq<n$, no deterministic algorithm with $k$ rounds and at
 most $q$ demand queries per round guarantees any finite multiplicative
 approximation on every graph.
 \item If $kq\ge n$, there is a deterministic algorithm using at most
 $k$ rounds and at most $q$ demand queries per round with approximation
 factor
 \begin{equation}\label{eq:deterministic-upper}
 O\!\left(\max\left\{1,\frac{n^2D^2}{k^2q}\right\}\right).
 \end{equation}
 \item Every deterministic algorithm subject to those budgets has
 worst-case approximation factor at least
 \begin{equation}\label{eq:deterministic-lower}
 \Omega\!\left(\max\left\{1,
                     \frac{n^2}{k^2q\Lambda}\right\}\right).
 \end{equation}
\end{enumerate}
Consequently, throughout the feasible range $kq\ge n$, the optimal
factor is
$\widetilde\Theta(\max\{1,n^2/(k^2q)\})$.
For $q=\Theta(n)$ and fixed $k$, the optimal factor is $\Theta(n)$,
without logarithmic losses.
\end{theorem}

Here the last assertion uses a separate one-round upper bound,
proved below.  

\subsection{Upper bound: recover lists and process batches}
\label{sec:deterministic-upper}

We start with a deterministic primitive that uses only Demand-OR
queries.  As a demand query simulates one Demand-OR query by placing
its set first, the primitive is available in both models.

\begin{lemma}[Parallel recovery of a bounded neighbour list]
\label{lem:deterministic-list-recovery}
Given $u\in L$, $S\subseteq R$, and an integer $b\ge1$, one can
recover $\min\{b,|N(u)\cap S|\}$ distinct neighbours using $D$ rounds
and at most $2b$ Demand-OR queries per round.  If fewer than $b$
neighbours exist, all of them are recovered.
\end{lemma}
\begin{proof}
Fix a balanced binary partition tree of $R$ of depth
$\lceil\log_2 n\rceil$, padding the leaves with empty cells if
necessary.  Intersect each cell with $S$.  A cell is positive when
its intersection contains a neighbour of $u$.

Query the root.  At each subsequent level, query both children of
every retained positive cell.  Retain all positive children if there
are at most $b$, and otherwise retain any $b$ of them, according to a
fixed rule.  There are at most $b$ retained cells and at most $2b$
queries at each level.  Once $b$ disjoint positive cells are retained,
each has a positive descendant, so at least $b$ positive cells remain
available at every later level.  At the leaves these yield $b$
distinct neighbours.  If the total number of neighbours is smaller
than $b$, no positive cell is ever discarded, and every neighbour is
recovered.  All choices at a level depend only on answers at earlier
levels, giving the stated round count.
\end{proof}

\begin{lemma}[Processing one batch]\label{lem:deterministic-batch}
Let $B\subseteq L$ have size at most $b$, and let $S$ be the currently
unmatched right vertices.  In $D$ rounds and at most $2b^2$
Demand-OR queries per round, one can process $B$ by a valid sequential
greedy matching execution relative to $S$.
\end{lemma}
\begin{proof}
For every vertex of $B$, recover $b$ neighbours in $S$, or all its
neighbours if fewer exist, using
\Cref{lem:deterministic-list-recovery}.  Run these recoveries in
parallel.  Then process $B$ in a fixed order without making further
queries, matching a vertex to an unused member of its recovered list
when possible.

If its list has $b$ members, at most $b-1$ right vertices have been
used earlier in this batch, so one listed neighbour remains free.
If its list has fewer than $b$ members, the list is complete.  Thus
whenever the procedure leaves a vertex unmatched, it has no available
neighbour.  These are exactly the decisions required of a sequential
greedy execution, with its neighbour choices determined by the lists.
The parallel query count is at most $b\cdot2b$ per round.
\end{proof}

\begin{lemma}[Parallel groups]\label{lem:deterministic-groups}
Suppose $1\le g\le\min\{n,\lfloor q/2\rfloor\}$.  There is a
$2g$-approximation using at most $q$ Demand-OR queries per round and
at most
\begin{equation}\label{eq:deterministic-groups-rounds}
 D\left(\frac{2\sqrt2\,n}{\sqrt{gq}}+2\right)
\end{equation}
rounds.  With $g=1$, its output is a maximal matching.
\end{lemma}
\begin{proof}
Partition $L$ into $g$ balanced groups $L_1,\ldots,L_g$.  For each
group, independently maintain a matching in $G[L_i,R]$ and process
its left vertices in batches of size at most
\[
 b=\left\lfloor\sqrt{\frac{q}{2g}}\right\rfloor\ge1.
\]
Here ``independently'' means separate deterministic executions and
separate right-hand availability sets; it does not refer to
randomness.  Run corresponding batches of all groups in parallel,
using \Cref{lem:deterministic-batch}.  The total per-round cost is at
most $2gb^2\le q$.

There are at most $n/(gb)+2$ batches per group.  Since
$b\ge\frac12\sqrt{q/(2g)}$, their total round count is at most
\eqref{eq:deterministic-groups-rounds}.  The output $M_i$ of a group
is maximal in $G[L_i,R]$, and hence has size at least $\mu_i/2$, where
$\mu_i$ is that group's maximum matching size.  This last inequality
follows by charging each edge of an optimum to an incident edge of
the maximal matching; each chosen edge receives at most two charges.
Restricting an optimum of $G$ to the left groups gives
$\sum_i\mu_i\ge\mu(G)$.  Therefore
\[
 \max_i|M_i|\ge\frac1g\sum_i|M_i|
              \ge\frac{\mu(G)}{2g}.
\]
Output a largest $M_i$.  Its edges form a matching even if different
groups used common right vertices, because only one group's matching
is output.  If $g=1$, it is a maximal matching of $G$ itself.
\end{proof}

\paragraph{A one-round fallback for $q\ge n$.}
When very few rounds are available, recovering full lists is
unnecessary.  One can instead insist that different output edges use
different right-hand blocks.

\begin{lemma}[One-round block algorithm]\label{lem:deterministic-blocks}
For $q\ge n$, a deterministic one-round algorithm using at most $q$
demand queries returns an $O(n^2/q)$-approximation.  It can also output
a vertex cover $C$ and its matching $M$ with
$|C|\le O(n^2/q)|M|$.
\end{lemma}
\begin{proof}
Set $t=\lfloor q/n\rfloor$.  Partition $R$ into $t$ nonempty blocks
of size at most $B=\lceil n/t\rceil=O(n^2/q)$.
For every $u\in L$ and every block, issue a demand query placing that
block first.  The $nt\le q$ answers determine the entire bipartite
incidence graph $H$ between $L$ and the blocks: an edge in $H$ means
that the corresponding left vertex has a neighbour in the block.
Every positive incidence also comes with a witness edge of $G$.
Compute a maximum matching $J$ of $H$ and output its witnesses as
$M$.  Because matched blocks are disjoint, these witnesses form a
matching in $G$ of size $|J|$.

A maximum matching in $G$ occupies at least $\mu(G)/B$ right blocks.
Selecting one of its edges in every occupied block gives a matching
in $H$, so $|M|=|J|\ge\mu(G)/B$.
For the cover certificate, take a minimum vertex cover of the known
bipartite graph $H$, whose size is $|J|$ by bipartite matching duality.
Replace each selected block by its right vertices, and keep every
selected left vertex.  This covers $G$ and has size at most $B|J|$.
\end{proof}

\paragraph{A low-budget fallback.}
When $q<n$, a full scan of the left vertices is possible as soon as
$kq\ge n$.  The following bound accounts explicitly for collisions
among the vertices processed in the same round.

\begin{lemma}[A single pass in batches]\label{lem:deterministic-one-pass}
If $kq\ge n$, a deterministic demand-query algorithm uses at most $k$
rounds and $q$ queries per round and returns a matching $M$ and a
vertex cover $C$ such that
\[
 |C|\le(\lceil n/k\rceil+1)|M|.
\]
\end{lemma}
\begin{proof}
Partition $L$ into at most $k$ batches of size at most
$B=\lceil n/k\rceil\le q$.  In a batch's round, each of its vertices
asks one demand query preferring the right vertices unmatched at the
beginning of the round.  Greedily add edges among the feasible
answers.  Call a batch productive if it adds at least one edge.

Let $C$ consist of the finally matched right vertices and every left
vertex in a productive batch.  If a batch adds no edge, no query in
it could have returned a neighbour unmatched at the beginning of the
round: the first such feasible answer would have been added.
Thus all neighbours of a left vertex in an unproductive batch were
already matched before that round.  It follows that $C$ covers every
edge.  There are at most $|M|$ productive batches and exactly $|M|$
matched right vertices, giving $|C|\le(B+1)|M|$.
\end{proof}

\begin{proof}[Upper-bound part of \Cref{thm:deterministic-tradeoff}]
If $k\ge n$, sequential greedy uses one query per round for at most
$n$ rounds and gives a maximal matching.  Suppose henceforth that
$k<n$.  Feasibility $kq\ge n$ then implies $q\ge2$.

First assume
\[
 k\ge64D\max\{1,n/q\}.
\]
Choose
\[
 g=\max\left\{1,\left\lceil
             \frac{128n^2D^2}{k^2q}\right\rceil\right\}.
\]
The assumed inequality implies
$g\le\min\{n,\lfloor q/2\rfloor\}$: before taking the ceiling,
the displayed fraction is at most $\min\{n,q\}/32$, and the
assertion also holds when the ceiling is one.  Apply
\Cref{lem:deterministic-groups}.  Its round count is at most
$k/4+2D\le k$, and its factor $2g$ is bounded by
\eqref{eq:deterministic-upper}.

It remains to consider smaller $k$.  If $q\ge n$, then $k<64D$;
\Cref{lem:deterministic-blocks} gives factor $O(n^2/q)$, which is
within a constant factor of \eqref{eq:deterministic-upper}.
If $q<n$, then $k<64Dn/q$.  The factor in
\Cref{lem:deterministic-one-pass} is $O(1+n/k)$, and
\[
 \frac nk
 \le64\,\frac{n^2D^2}{k^2q},
\]
using $D\ge1$.  This is again sufficient.  Every algorithm used here
respects the budget $q$ exactly, not just in expectation or up to a
constant factor.
\end{proof}

\subsection{Lower bound: two graphs with the same transcript}
\label{sec:deterministic-lower}

A partial graph assigns each pair in $L\times R$ one of three states:
present, absent, or undecided.  An answer is committed only after all
pairs preceding it in its demand ordering have been fixed absent.
Consequently later decisions on undecided pairs cannot change an
earlier answer.

\begin{lemma}[A set hitting almost all prefixes]
\label{lem:partial-hitting-set}
Let $\mathcal S$ be a family of at most $q$ subsets of an $n$-element
set, each of size $\theta\ge1$, and suppose $kq\ge n$.  There exists
a set $T$ of size at most
\[
 h=\min\left\{n,\left\lceil
       \frac n\theta\log\left(\frac{16kq}{n}\right)
                    \right\rceil\right\}
\]
that misses at most $n/(8k)$ members of $\mathcal S$.
\end{lemma}
\begin{proof}
If $h=n$, take the whole ground set.  Otherwise a uniformly random
$h$-element set misses a fixed $\theta$-set with probability at most
\[
 \frac{\binom{n-\theta}{h}}{\binom nh}
 \le (1-\theta/n)^h\le e^{-h\theta/n} \le 
e^{-\log(16kq/n)}
=
\frac{n}{16kq}.
\]
Its expected number of missed members is at most $n/(16k)$.
Therefore some set misses at most $n/(8k)$ members.  This is an
existence argument for the adversary; no random choice by the
algorithm is involved.
\end{proof}

\begin{lemma}[Indistinguishable completions]\label{lem:deterministic-completions}
Suppose $kq\ge n$ and
$\theta=\lfloor n^2/(64kq)\rfloor\ge1$.
For any deterministic $k$-round algorithm using at most $q$ demand
queries per round, there are two graphs with the same transcript,
$G_{\mathrm{small}}$ and $G_{\mathrm{large}}$, such that
\[
 \mu(G_{\mathrm{large}})\ge\lfloor n/2\rfloor,
 \qquad
 \text{$G_{\mathrm{small}}$ has a right-side vertex cover of size}
 \quad O\!\left(\frac{k^2q}{n}\Lambda\right).
\]
\end{lemma}
\begin{proof}
Fix a deterministic algorithm. We describe an adversary that answers its
queries while assigning each pair in $L\times R$ one of three states:
present, absent, or undecided. A fixed state is never changed.

Initially all pairs are undecided and all left vertices are active.
To retire a left vertex means fixing all its remaining undecided
incident pairs absent. Its neighbourhood is then completely determined, and all
future queries to it are answered from this fixed neighbourhood. Retirement does
not change any previously fixed pair.

Let
\[
    \theta=\left\lfloor\frac{n^2}{64kq}\right\rfloor\ge1,
    \qquad
    h=\min\left\{n,\left\lceil
        \frac n\theta\log\left(\frac{16kq}{n}\right)
    \right\rceil\right\}.
\]
After round $r$, write $A_r$ for the active left vertices and
$T_{\le r}=T_1\cup\cdots\cup T_r$ for the right vertices selected
by the adversary. We maintain the following invariants:
\begin{enumerate}
    \item Every present pair has its right endpoint in $T_{\le r}$, i.e., $T_{\le r}$ is a vertex cover,
    and $|T_{\le r}|\le rh$.

    \item Every vertex in $A_r$ has fewer than $n/2$ incident pairs
    fixed absent. Every vertex outside $A_r$ has its entire neighbourhood fixed.

    \item The number of active vertices satisfies
    \[
        |A_r|\ge n\left(1-\frac{5r}{32k}\right).
    \]
\end{enumerate}
We also ensure that every answer is determined by the partial graph:
a returned pair is fixed present and every preceding pair in its
query order is fixed absent; an answer $\bot$ has its entire row
fixed absent. Thus every completion of the undecided pairs preserves
all previous answers.

These properties suffice for the lemma. At the end of round $k$,
make every undecided pair absent to obtain $G_{\mathrm{small}}$.
By the first invariant, $T_{\le k}$ is a right-side vertex cover
of this graph, of size at most $kh = O\!\left(\frac{k^2q}{n}\Lambda\right)$. Alternatively, make every
undecided pair present to obtain $G_{\mathrm{large}}$. By the
second and third invariants, at least $27n/32$ left vertices have
degree greater than $n/2$ in this completion. Greedily matching any
$\lfloor n/2\rfloor$ of these vertices succeeds: when each vertex
is processed, fewer than $n/2$ right vertices have been used.
Consequently,
\[
    \mu(G_{\mathrm{large}})\ge\lfloor n/2\rfloor.
\]
Both completions give the same transcript.

It remains to describe the adversary and verify the invariants.
They hold initially, with $A_0=L$ and $T_{\le0}=\varnothing$.
Consider round $r$, whose queries are all fixed before any of
its answers are given.

\paragraph{Long queries.}
For each query to an active left vertex, list the undecided pairs
preceding its first already-present pair, in query order. If there
is no present pair, list all undecided pairs in the row. Call the
query long if this list has at least $\theta$ entries, and let its
prefix set consist of the right endpoints of the first $\theta$
entries. All other queries are short.

Apply Lemma~\ref{lem:partial-hitting-set} to the family of long-query
prefix sets, counting repeated sets once for each query. It gives
a set $T_r$ of size at most $h$ that misses at most $n/(8k)$
prefixes. Retire every active left vertex having a missed prefix.
This retires at most $n/(8k)$ vertices.

For every remaining active left vertex, fix all its undecided
pairs into $T_r$ present. Add $T_r$ to $T_{\le r-1}$ to obtain
$T_{\le r}$. Make all these positive commitments before answering
any query of the round. The first invariant is preserved, since
new present pairs enter $T_r$ and $|T_r|\le h$.

\paragraph{Answering the queries.}
Answer queries to retired vertices from their fixed rows.
For each query to a still-active vertex, return its first
present pair and fix all preceding undecided pairs absent.
If the row has no present pair, return $\bot$ and fix all its
remaining undecided pairs absent.

These answers are consistent with one another and with all earlier
answers. Indeed, the positive commitments have already been made,
and answering a query changes only undecided pairs to absent.
In particular, it cannot remove a present pair that determines
the answer to another query.

Moreover, each query to a still-active vertex fixes at most
$\theta$ new absent pairs. A long query has a present pair among
its first $\theta$ entries from the round-start list, because
its prefix was hit by $T_r$. A short query had fewer than $\theta$
undecided pairs before its earlier first present pair, or fewer
than $\theta$ undecided pairs altogether. New positive commitments
and answers to other queries can only reduce the number that it
fixes absent.

\paragraph{Retiring vertices with too many absent pairs.}
At the end of the round, retire every still-active left vertex
with at least $n/2$ incident pairs fixed absent. This restores the
second invariant. Retirement preserves the first invariant and
all answers, since it fixes only undecided pairs absent.

To verify the third invariant, distinguish pairs fixed absent
while answering queries to active vertices from the additional
pairs fixed absent during retirement. Only the former are counted
in the following budget. Since each query creates at most $\theta$
such pairs, their total number through round $r$ is at most
\[
    rq\theta.
\]
Every vertex retired by the absent-pair threshold has accumulated
at least $n/2$ such pairs before retirement. Distinct left vertices
have disjoint incident pairs, so the number retired by this rule
through round $r$ is at most
\[
    \frac{2rq\theta}{n}
    \le \frac{rn}{32k}.
\]
The missed-prefix rule has retired at most $rn/(8k)$ vertices.
Therefore
\[
    |A_r|
    \ge n-\frac{rn}{8k}-\frac{rn}{32k}
    =n\left(1-\frac{5r}{32k}\right),
\]
as required. This verifies the invariants after every round.

Finally, since $\theta=\lfloor n^2/(64kq)\rfloor\ge1$,
\[
    \frac n\theta\le\frac{128kq}{n}.
\]
Using $kq/n\ge1$ and
$\log(16kq/n)=O(\Lambda)$, we obtain
\[
    |T_{\le k}|
    \le kh
    \le k\left(
        1+\frac n\theta\log\left(\frac{16kq}{n}\right)
    \right)
    =O\!\left(\frac{k^2q}{n}\Lambda\right).
\]
Together with the two completions described above, this proves
the lemma.
\end{proof}

\paragraph{Comparison with Nisan.}
Nisan's adversary~\cite[Section~3]{Nisan21} limits the number of new
absent pairs at every vertex in each round, preserving a completion
with a perfect matching. To maintain this condition, it fixes
additional pairs present at heavily queried vertices, which
contributes to the cover of the small completion. Here we instead
bound the total number of absent pairs created by queries to active
vertices and retire the few left vertices that become too damaged.
Preserving a linear-size matching, rather than a perfect matching,
is sufficient for our approximation lower bound. This removes the
heavy-vertex contribution to the cover and permits
$\theta=\Theta(n^2/(kq))$. Allowing a few missed prefixes gives the
factor $\Lambda$, while handling demand-query prefixes directly
avoids the round overhead of simulation by OR queries.

\begin{proof}[Lower-bound and feasibility parts of
\Cref{thm:deterministic-tradeoff}]
If $kq<n$, follow the algorithm on the empty graph.  Some left vertex
is never queried.  Adding one edge incident to that vertex does not
change the transcript.  The algorithm must output the empty
matching on the empty graph and hence also on this one-edge graph,
contradicting every finite multiplicative approximation.

Now suppose $kq\ge n$.  If $\theta\ge1$, apply
\Cref{lem:deterministic-completions}.  The algorithm outputs the same
matching on both completions.  Validity on $G_{\mathrm{small}}$
forces its size to be at most $|T|$, whereas the optimum on
$G_{\mathrm{large}}$ is at least $\lfloor n/2\rfloor$.
This proves the nonconstant term of
\eqref{eq:deterministic-lower}; the trivial bound of one handles
parameters where that term is bounded.
If $\theta=0$, then $n^2/(kq)<64$, so
$n^2/(k^2q\Lambda)=O(1)$ and the same trivial bound suffices.
This proof also handles guesses of uncertified edges: any such output
must still be valid on the small completion.

For fixed $k$ and $q=\Theta(n)$, $\Lambda=\Theta(1)$ and the lower
bound is $\Omega(n)$.  If $q\ge n$, the one-round block algorithm
has factor $O(n)$.  For $q<n$ with feasible fixed $k$ and
$q=\Theta(n)$, \Cref{lem:deterministic-one-pass} gives the same factor.
\end{proof}

\subsection{Corollaries: maximal matching, certificates, and randomisation}
\label{sec:deterministic-corollaries}

\begin{corollary}[Round complexity of constant approximation]
\label{cor:deterministic-round-complexity}
For $n\le q\le n^2$, the deterministic round complexity of obtaining
a constant approximation, and that of constructing a maximal
matching, is
\[
 \widetilde\Theta(n/\sqrt q).
\]
For $q=\Theta(n)$, the explicit lower and upper bounds are
$\Omega(\sqrt{n/\log n})$ and $O(\sqrt n\log n)$, respectively.
\end{corollary}
\begin{proof}
Set $g=1$ in \Cref{lem:deterministic-groups}.  It constructs a maximal
matching in $O(Dn/\sqrt q)$ rounds; this is a 2-approximation.
Conversely, \eqref{eq:deterministic-lower} requires
$k^2q\Lambda=\Omega(n^2)$ for any fixed constant approximation.
For $k\le n$ and $q\le n^2$, $\Lambda=O(\log n)$, giving
$k=\Omega(n/\sqrt{q\log n})$, as well as the trivial requirement
$k\ge1$.  If $k>n$, the asserted lower bound already holds.
Every maximal matching is a 2-approximation, so the same lower bound
applies to constructing one.
\end{proof}

\begin{corollary}[Polynomial exponents]\label{cor:deterministic-exponents}
Fix exponents $\gamma\ge0$ and $-1\le\delta\le1$ with
$\gamma+\delta\ge0$.  Up to integer rounding, for
$k=n^\gamma$ and $q=n^{1+\delta}$ the optimal deterministic
approximation factor is
\[
 n^{\max\{0,\,1-\delta-2\gamma\}+o(1)}.
\]
\end{corollary}
\begin{proof}
Feasibility is $kq\ge n$.  Substitute the parameters in
\Cref{thm:deterministic-tradeoff}; all displayed logarithmic factors
are $n^{o(1)}$.  If rounding changes feasibility at the boundary,
round $k$ or $q$ upward.
\end{proof}

\begin{corollary}[Matching and cover certificates]
\label{cor:deterministic-certificates}
The algorithm proving \eqref{eq:deterministic-upper} can return both
a matching $M$ and a vertex cover $C$ with
\[
 |M|\le\mu(G)\le|C|
 \le O\!\left(\max\left\{1,\frac{n^2D^2}{k^2q}\right\}\right)|M|.
\]
The deterministic approximation tradeoff of
\Cref{thm:deterministic-tradeoff} also holds, up to the same
logarithmic factors, for minimum vertex cover and maximum fractional
matching.
\end{corollary}
\begin{proof}
In the group algorithm, let $C$ be the union of the endpoints of all
$M_i$.  Each $M_i$ is maximal on its left group, so its endpoints
cover all edges incident to that group.  Thus $C$ covers $G$ and
\[
 |C|\le2\sum_i|M_i|\le2g\max_i|M_i|.
\]
The one-round and one-pass fallbacks already include the corresponding
certificates.  The endpoints of sequential greedy supply a cover
when $k\ge n$.

Any matching has size at most the size of any vertex cover.  Also, if
$x$ is a fractional matching (nonnegative edge weights with incident
sum at most one at each vertex), any vertex cover bounds
$\sum_e x_e$ from above.  Therefore the certificate inequality gives
the asserted upper bounds for both objectives, without needing to
recover any additional edges.

For the lower bounds use the same two completions.  Every fractional
matching feasible on $G_{\mathrm{small}}$ has total value at most
$|T|$, since all its edges meet $T$.  The identical output cannot
approximate the linear-size integral matching of
$G_{\mathrm{large}}$ more accurately.  For vertex cover, any common
output valid on $G_{\mathrm{large}}$ has size at least
$\lfloor n/2\rfloor$, but $G_{\mathrm{small}}$ has the cover $T$.
This gives the same approximation lower bound.  If $kq<n$, the
empty-graph and one-edge argument rules out a finite approximation
for these objectives as well.
\end{proof}

\begin{corollary}[Deterministic versus randomised rounds]
\label{cor:deterministic-randomised-separation}
With $\Theta(n)$ demand queries per round, a constant approximation
to either maximum matching or minimum vertex cover requires
$\widetilde\Theta(\sqrt n)$ rounds deterministically.  With
randomisation, $O(\log n/\log\log n)$ rounds suffice for both problems
with high probability.  For matching construction, this randomised
round complexity is also necessary up to constants.
\end{corollary}
\begin{proof}
Combine \Cref{cor:deterministic-round-complexity,cor:deterministic-certificates}
with \Cref{ub:thm:constant-approximation,ub:cor:vertex-cover,cor:random-lb-consequences}.
The randomised lower bound is invoked only for matching construction.
\end{proof}

\subsection{Demand-OR queries and the one-round boundary}
\label{sec:deterministic-or}

\begin{theorem}[Deterministic Demand-OR tradeoff after two rounds]
\label{thm:deterministic-or-tradeoff}
For $q\ge n$ and $k\ge2$, the deterministic approximation bounds
\eqref{eq:deterministic-upper} and
\eqref{eq:deterministic-lower} also hold for Demand-OR queries.
Thus their optimal factor is
$\widetilde\Theta(\max\{1,n^2/(k^2q)\})$ in this range.
Moreover, for every $2\le k\le n$, the per-round Demand-OR query
complexity of obtaining a constant-factor approximation is
$\widetilde\Theta(n^2/k^2)$, without the restriction $q\ge n$.
\end{theorem}
\begin{proof}
For $k\ge64D$, use the group algorithm and the choice of $g$ in the
proof of \Cref{thm:deterministic-tradeoff}; these use only Demand-OR
queries and remain valid when $k\ge n$.  It remains to replace the
one-round demand-query fallback when $k<64D$.
Partition $R$ into $t=\lfloor q/n\rfloor$ blocks as in
\Cref{lem:deterministic-blocks}.  In round one, use $nt\le q$
Demand-OR queries to recover the left--block incidence graph, and
compute a maximum matching of that graph.  In round two, for each
selected left--block pair, query all individual pairs in that block.
Selected blocks are disjoint, so the second round uses at most $n\le q$
queries.  Recover one witness per selected incidence.  The original
block argument gives factor $O(n^2/q)$, which suffices when $k<64D$.

The lower bound follows because one demand query simulates one
Demand-OR query.  An algorithm for the weaker model would otherwise
contradict \eqref{eq:deterministic-lower} in the stronger model.

For the constant-approximation claim, if $2\le k<8D$, query all
$n^2$ pairs individually in one round.  Otherwise, set
$q=\lceil 32n^2D^2/k^2\rceil$ and apply
\Cref{lem:deterministic-groups} with $g=1$.  Its output is a maximal
matching, and it uses at most $k/2+2D\le k$ rounds.  Hence
$q=O(n^2D^2/k^2)$ suffices throughout $2\le k\le n$.
Conversely, \eqref{eq:deterministic-lower} requires
$q=\Omega(n^2/(k^2\log n))$ for any fixed constant approximation,
since $\Lambda=O(\log n)$ when $k\le n$ and $q\le n^2$.
\end{proof}

The restriction $q\ge n$ in this two-round fallback is explicit.  For
$q<n$, \Cref{lem:deterministic-groups} still supplies the stated
Demand-OR upper bound whenever its round requirement is met, but the
one-pass demand-query fallback does not transfer automatically.

\begin{theorem}[Exact nonadaptive Demand-OR threshold]
\label{thm:deterministic-or-one-round}
A deterministic one-round Demand-OR algorithm that always returns a
valid matching and guarantees any finite multiplicative approximation
on every bipartite graph requires at least $n^2$ queries.
Conversely, $n^2$ queries suffice to recover the entire graph and find
a maximum matching.
\end{theorem}
\begin{proof}
The upper bound queries each pair separately.  For the lower bound,
suppose fewer than $n^2$ queries are issued.  Some left vertex $u$
participates in $t<n$ queries, with sets $S_1,\ldots,S_t\subseteq R$.
Make every other left vertex isolated.

Start with $U=R$.  Whenever some $S_i\cap U$ is a singleton $\{v\}$,
remove $v$ from $U$.  That query's intersection becomes empty and
can never witness another removal.  There are therefore at most
$t$ removals, and $U$ is nonempty.  At termination each $S_i\cap U$
has size zero or at least two.  For every $v\in U$, the two
neighbourhoods
\[
 N(u)=U\qquad\text{and}\qquad N(u)=U\setminus\{v\}
\]
produce identical answers to every query.  Answers at other left
vertices are zero in both cases.

No pair incident to $u$ is guaranteed present under this transcript:
a pair outside $U$ is absent in the first graph, and a pair to
$v\in U$ is absent in its corresponding second graph.  Thus an
always-valid deterministic algorithm must output the empty matching.
But the graph with neighbourhood $U$ has maximum matching size one,
contradicting every finite approximation factor.
\end{proof}

The underlying obstruction is deterministic nonadaptive single-element
recovery; the broader round--query tradeoffs for that primitive are
studied by Assadi, Chakrabarty, and Khanna~\cite{assadi2021graph}.
Here it produces a sharp model distinction: $n$ one-round demand
queries suffice for an $n$-approximation, but fewer than $n^2$
one-round deterministic Demand-OR queries cannot guarantee any finite
approximation.  From two rounds onward, with $q\ge n$, the polynomial
tradeoffs of the two models agree.

\begin{proposition}[One-round deterministic Demand-OR vertex cover]
\label{prop:deterministic-or-vertex-cover}
For every $n\le q\le n^2$, there is a deterministic one-round
Demand-OR algorithm using at most $q$ queries that outputs an
$O(n^2/q)$-approximate minimum vertex cover.
\end{proposition}
\begin{proof}
Partition $R$ into $t=\lfloor q/n\rfloor$ blocks of size at most
$B=\lceil n/t\rceil=O(n^2/q)$.  Query $(u,S)$ for each $u\in L$ and
each block $S$.  These $nt\le q$ responses specify the bipartite
incidence graph $H$ between left vertices and right blocks.  Find a
minimum vertex cover $C_H$ in $H$.  Output the left vertices of $C_H$
together with all right vertices in its selected blocks.  This is a
vertex cover of $G$ of size at most $B|C_H|$.  By K\"onig's theorem,
$|C_H|=\mu(H)\le\mu(G)$: every matching of $H$ can be realised by
choosing one original edge from each of its distinct matched blocks.
Thus the output has size at most $B\mu(G)=B\tau(G)$, where
$\tau(G)$ is the minimum vertex-cover size.  No witness edges need
to be recovered by the queries.
\end{proof}


\subsection*{AI Acknowledgment}


ChatGPT 5.6 Sol and 6 Astra were used extensively for typesetting and editing various parts of this work. An earlier version of this work was written entirely without LLM assistance. We subsequently interacted with ChatGPT, which led to an improved analysis of our randomised $k$-round upper bound (by a factor of $\Theta(k)$) and the development of the deterministic round--query bounds presented here.

\bibliographystyle{alpha}
\bibliography{ref}

\appendix

\end{document}